\documentclass{article}

\usepackage{arxiv}

\usepackage[utf8]{inputenc} % allow utf-8 input
\usepackage[T1]{fontenc}    % use 8-bit T1 fonts
\usepackage{hyperref}       % hyperlinks
\usepackage{url}            % simple URL typesetting
\usepackage{booktabs}       % professional-quality tables
\usepackage{amsfonts}       % blackboard math symbols
\usepackage{nicefrac}       % compact symbols for 1/2, etc.
\usepackage{microtype}      % microtypography
\usepackage{graphicx}
\usepackage{doi}

\title{Smoothing methods to estimate the hazard rate under double truncation}

\author{\href{https://orcid.org/0000-0002-0570-0650}{\includegraphics[scale=0.06]{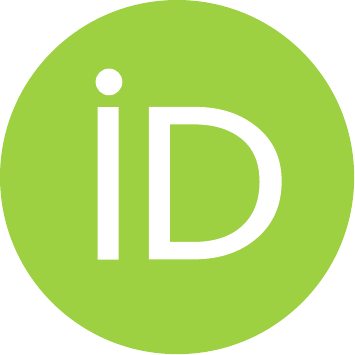}\hspace{1mm}Carla Moreira}\\%\thanks{Use footnote for providing further
	CMAT-Centre of mathematics\\ 
University of Minho\\ 
Braga, Portugal\\
EPIUnit - Instituto de Sa\'ude P\'ublica\\
Universidade do Porto\\
    Porto, Portugal\\
	\texttt{carlamgmm@gmail.com} \\
	\And
	\href{https://orcid.org/0000-0002-4686-8417}{\includegraphics[scale=0.06]{orcid.pdf}\hspace{1mm}Jacobo ~de U\~na-\'Alvarez} \\
	Department of Statistics and OR\\
    SiDOR research group \\
    CINBIO- University of Vigo\\
     Vigo, Spain\\
	\texttt{jacobo@uvigo.es} \\
	\AND
	\href{http://orcid.org/0000-0002-2992-5299}{\includegraphics[scale=0.06]{orcid.pdf}\hspace{1mm}Ana Cristina ~Santos} \\
	EPIUnit - Instituto de Sa\'ude P\'ublica\\
    Departamento de Ci\^encias da Sa\'ude P\'ublica e Forenses e Educa\c{c}\~ao M\'edica\\
    Faculdade de Medicina\\
     Universidade do Porto\\
    Porto, Portugal\\
	\texttt{acsantos@med.up.pt} \\
\AND
    \href{https://orcid.org/0000-0003-4699-6571}{\includegraphics[scale=0.06]{orcid.pdf}\hspace{1mm}Henrique ~Barros} \\
	EPIUnit - Instituto de Sa\'ude P\'ublica\\
    Departamento de Ci\^encias da Sa\'ude P\'ublica e Forenses e Educa\c{c}\~ao M\'edica\\
    Faculdade de Medicina\\
     Universidade do Porto\\
    Porto, Portugal
	\texttt{hbarros@med.up.pt } \\
}

\renewcommand{\shorttitle}{\textit{arXiv} Template}
\newtheorem{prop}{Proposition}
\newtheorem{thm}{Theorem}

\newcommand{\argmax}{\operatorname{argmax}\displaylimits}
\newcommand{\argmin}{\operatorname{argmin}\displaylimits}

\hypersetup{
pdftitle={A template for the arxiv style},
pdfsubject={q-bio.NC, q-bio.QM},
pdfauthor={David S.~Hippocampus, Elias D.~Striatum},
pdfkeywords={First keyword, Second keyword, More},
}

\begin{document}
\maketitle

\begin{abstract}
In Survival Analysis, the observed lifetimes often correspond to individuals for which the event occurs within a specific calendar time interval. With such interval sampling, the lifetimes are doubly truncated at values determined by the birth dates and the sampling interval. This double truncation may induce a systematic bias in estimation, so specific corrections are needed. A relevant target in Survival Analysis is the hazard rate function, which represents the instantaneous probability for the event of interest. In this work we introduce a flexible estimation approach for the hazard rate under double truncation. Specifically, a kernel smoother is considered, in both a fully nonparametric setting and a semiparametric setting in which the incidence process fits a given parametric model.  Properties of the kernel smoothers are investigated both theoretically and through simulations. In particular, an asymptotic expression of the mean integrated squared error is derived, leading to a data-driven bandwidth for the estimators.  The relevance of the semiparametric approach is emphasized, in that it is generally more accurate and, importantly, it avoids the potential issues of nonexistence or nonuniqueness of the fully nonparametric estimator. Applications to the age of diagnosis of Acute Coronary Syndrome (ACS) and AIDS incubation times are included.
\end{abstract}


\maketitle

\section{Introduction}
Acute coronary syndrome (ACS) is still one of the main causes of death in Europe and worldwide. Coronary heart disease mortality has decreased in the last decades in high-income countries because of primary prevention and improvement in treatment of patients with ACS \cite{Nichols2014}.  EPIHeart \cite{Araujo18} is a prospective cohort study which includes patients with confirmed diagnosis of type 1 (primary spontaneous) ACS who were consecutively admitted to the Cardiology Department of two tertiary hospitals in Portugal between August 2013 and December 2014. EPIHeart cohort comprises 939 cases with ages at diagnosis ranging from $30$ to $94$ years.
Due to the aforementioned sampling scheme, the age at diagnosis in EPIHeart cohort is  observed conditionally on being larger than the time from birth to the study onset (August 2013) and smaller than the age at the end of the study (December 2014); i.e, the event time is doubly truncated by the sampling interval. This restriction in the observation of the ages at infarction may entail biases in estimation, unless appropriate corrections are used \cite{Mandel2017, RennertXie2019}.

In general, interval sampling occurs when only those individuals whose event lies within a certain calendar time window are observed \cite{Zhu2014}. Individuals with
event out of such sampling interval are not observed, and no information about them is available to the investigator. As mentioned, interval sampling induces double truncation on the event times; this phenomenon is less known, and much more complicated, than one-sided truncation, where only left or right observational limits are present. In particular, the nonparametric maximum-likelihood estimator (NPMLE) for doubly truncated data has no explicit form, and iterative methods are needed \cite{Efron99}. \cite{Moreira010} implemented three different iterative algorithms to compute the NPMLE. However, convergence of these algorithms for a particular dataset does not imply that the NPMLE actually exists. Indeed, the NPMLE may not exist, and may not be unique \cite
{Xiao19}. Therefore, in practice a preliminary inspection of these potential issues is required. Complications under double truncation arise in theoretical developments too; see for instance \cite{dUAVK19} for recent updates and a discussion. Summarising, double truncation requires attention and is far from simple.

One of the goals of EpiHeart study is to evaluate the risk of infarction along time from such cohort and, therefore, the estimation of the hazard rate is of interest. The hazard rate function reports the instantaneous probability of death or failure along time, and it plays an essential role in Epidemiology and particularly in Survival Analysis.  The shape of the hazard rate function helps in apprehending the mechanism that affects survival and hence, in absence of shape constraints, nonparametric methods turn out to be particularly useful for the estimation of the true curve. In this context, kernel based estimation of the hazard rate function, following the spirit of kernel density estimation \cite{Wand95}, has received particular attention in the literature.

Nonparametric estimation of density and hazard rate functions under right censoring was investigated by  \cite{Diehl88} and  \cite{Wang94}, while  \cite{Cai98} considered the issue of dependent censoring.   \cite{Zhou999} studied kernel density and hazard rate estimation based on the product-limit estimator from left truncated data \cite{Lynden71}. The mean squared error for the kernel estimator of the hazard rate from left truncated and right censored data was derived in  \cite{Lemdani07}. \cite{Moreira2012} introduced kernel-type density estimation for doubly truncated data. These authors considered a purely nonparametric density estimator, as well as a semiparametric estimator which exploits information on the truncation distribution \cite{Moreira110}. The two estimators
were shown to be asymptotically equivalent but, as expected, the semiparametric estimator exhibited a better finite-sample performance. \cite{Moreira13} explored several bandwidth selection procedures for kernel density estimation under double truncation, which are appropriate modifications of the normal reference rule, the least squares cross-validation procedure, plug-in procedures, and a bootstrap based method. However, smoothing methods for the hazard rate under double truncation have not been investigated yet.

In this paper we propose and investigate nonparametric and semiparametric kernel smoothers for the hazard rate when the data are subject to double truncation. We also illustrate how the application of the semiparametric approach may mitigate the aforementioned issues for the NPMLE. The rest of the paper is organized as follows. In Section 2 we introduce the needed notation and we review the NPMLE and the semiparametric maximum likelihood estimator (SPMLE) for doubly truncated data; the smooth hazard rate estimators are introduced too. In Section 3 the main theoretical properties of the proposed estimators for the hazatd rate are given.  The finite-sample performance of the estimators is investigated through simulations in Section 4. Section 5 is devoted to the analysis of two real datasets: the EPIHeart cohort data, and the well-known CDC AIDS blood transfusion data \cite{Lawless89}. The main conclusions of our research are given in the Section 6, whereas the technical proofs and the details on bandwidth selection and the smoothed bootstrap are deferred to the Appendix.

%The analysis of doubly truncated data has received a lot of attention in the recent literature. Among other relevant contributions, we point out the development of more efficient estimators for special forms of double truncation \cite{Austin2014}, the estimation of the association parameter and the bivariate survival under interval sampling \cite{Zhu2014, Zhu2012}, and the adaptation of Cox regression to the setting of double truncation \cite{Mandel2017,Rennert17}.

\section{Notation and estimators}
%\label{Section2}

Let $X^{\ast }$ be the random variable of ultimate interest, with continuous df
$F$, and assume that it is doubly truncated by the random pair
$\left( U^{\ast },V^{\ast }\right) $ with joint df $T$, where
$U^{\ast }$ and $V^{\ast }$ ($U^{\ast }\leq V^{\ast }$) are the left
and right truncation variables respectively. This means that the
triplet $\left( U^{\ast },X^{\ast },V^{\ast }\right) $ is observed
if and only if $U^{\ast }\leq X^{\ast }\leq V^{\ast }$, while no
information is
available when $X^{\ast }<U^{\ast }$ or $X^{\ast }>V^{\ast }$.
In the special case of interval sampling we have $V^{\ast }=U^{\ast }+ \tau$, for some constant $\tau>0$ (the width of the sampling interval).  It is assumed that $\left( U^{\ast },V^{\ast }\right)$ is independent of $X^{\ast }$; we also assume that the supports of $X^*$, $U^*$ and $V^*$ are such that the identifiability conditions  for $F$ are satisfied \citep{Woodroofe85}. \\

\subsection{The hazard function estimator}
%\label{Subsectio23}
The ordinary kernel estimator for the hazard function is given by

\begin{equation}
\lambda_{h}(x)=\int K_{h}(x-t)\Lambda_{n}(dt)= \int K_{h}(x-t)\frac{F_n(dt)}{1-F_n(t^-)}\\
%&\equiv%
%\alpha_n \frac1n{\sum_{i=1}^{n}K_{h}(x-X_{i})\frac{G_{n}(X_{i})^{-1}}{1-F_n(X_i^{-})}}
%\label{npkde}
\label{npkde1}
\end{equation}

\noindent where $K_{h}(t)=K(t/h)/h$ is the re-scaled kernel function,
$h=h_{n}$ is a deterministic bandwidth sequence with $h_n
\rightarrow 0$ as $n \rightarrow \infty$, $\Lambda_n$ is the empirical cumulative hazard, and  $F_n$ is the empirical cumulative distribution function \citep{Wand95}. For doubly truncated data $F_n$ is the Efron and Petrosian estimator, which is the NPMLE in such a setting \citep{Efron99}. We revisit this estimator in the following lines.

%\subsection{The NPMLE of the cumulative df}

Let $\left( U_i,X_i,V_i\right)$, $1 \leq i \leq n$, be the observed data, which are independent copies with the conditional distribution of $\left( U^{\ast },X^{\ast },V^{\ast }\right)$ given $U^{\ast }\leq X^{\ast }\leq V^{\ast }$.  Here, without loss of generality we assume that the NPMLE is a discrete distribution supported by the set of observed data.  Let $\varphi=(\varphi_1,\ldots,\varphi_n)$ be a distribution putting probability $\varphi_i$ on $X_i, i=1,\ldots,n$. Similarly, let $\psi=(\psi_1,\ldots,\psi_n)$ be a distribution putting joint probability $\psi_i$ on $(U_i,V_i), i=1,\ldots,n$.
Under the assumption of independence between $X^*$ and
$(U^*,V^*)$ the full likelihood, $\mathcal {L}(\varphi,\psi)$, can be decomposed as a product of the conditional likelihood of the $X_i$'s given the
  $(U_i,V_i)$'s, say $\mathcal {L}_1(\varphi)$, and the marginal likelihood of the $(U_i,V_i)$'s, say $\mathcal {L}_2(\varphi,\psi)$:
 \begin{equation}
 \mathcal {L}(\varphi,\psi)=\prod_{j=1}^{n}\frac{\varphi_j}{\Phi_j}\times\prod_{j=1}^{n}\frac{\Phi_j\psi_j}{\sum_{i=1}^{n}\Phi_i\psi_i}=\mathcal {L}_1(\varphi)\times \mathcal {L}_2(\varphi,\psi)
 \label{Ln_formula}
 \end{equation}
where $\Phi_i$ is defined through
$\Phi_i=\displaystyle\sum_{m=1}^{n}{\varphi_mJ_{im}}$,
$i=1,\ldots,n$ with $J_{im}=I(U_i\leq X_m\leq V_i)$   the indicator of the event $U_i
\leq X_m \leq V_i$.

\bigskip
The conditional NPMLE of $F$ \citep{Efron99} is defined as the
maximizer of $\mathcal {L}_1(\varphi)$ in equation
(\ref{Ln_formula}): $ \hat \varphi=\argmax_{\varphi} \mathcal
{L}_1(\varphi)$. The conditional NPMLE
$F_n(x)=\displaystyle\sum_{i=1}^{n}\hat \varphi_iI(X_i\leq x)$
maximizes indeed the full likelihood, which can be also written as
the product
\[
\mathcal
{L}(\varphi,\psi)=\prod_{j=1}^{n}\frac{\psi_j}{\Psi_j}\times
\prod_{j=1}^{n}\frac{\Psi_j\varphi_j}{\sum_{i=1}^{n}\Psi_i
\varphi_i}=\mathcal {L}_1^{*}(\psi)\times \mathcal
{L}_2^{*}(\psi,\varphi)
\]
where $\Psi_i=\displaystyle\sum_{m=1}^{n} \psi_mJ_{mi}$, $i=1\ldots,n$.
Here, $\mathcal L_1^{*}(\psi)$ denotes the conditional likelihood of
the $(U_i,V_i)$'s given the $X_i$'s and $\mathcal
{L}_2^{*}(\psi,\varphi)$ refers to the marginal likelihood of the
$X_i$'s. Introduce $\hat \psi=(\hat \psi_1,...,\hat \psi_n)$ as the
maximizer of $\mathcal {L}^*_1(\psi)$; then,
$T_n(u,v)=\displaystyle\sum_{i=1}^{n}\hat \psi_iI(U_i\leq
u,V_i\leq
v)$ is the NPMLE of $T$ \cite{Shen08}.\\
\bigskip

The NPMLE of $F$ also admits the representation

\begin{equation}
F_n(x)=\alpha_{n} \int_{a_F}^{x}\frac{F^{\ast }_{n}(dt)}{G_{n}(t )}
\label{ipwe}
\end{equation}

\noindent where $a_F$ is the lower limit of the support of $F$, $F^{\ast }_{n}$ is the ordinary empirical df of the
$X_i$'s,

\begin{eqnarray*}
G_n(t)=\int_{\left\{ u\leq t\leq v\right\} }T_n(du,dv)
\end{eqnarray*}

\noindent is the NPMLE for the conditional
probability of sampling a specific $X^*$-value, $X^{\ast}=t$, which is given by
$G(t)=P(U^{\ast }\leq t\leq V^{\ast })$; and
$\alpha_{n}=(\int_{a_F}^{\infty }G_{n}^{-1}(t)F^*_{n}(dt))^{-1}$ is an
estimator for the no-truncation probability $\alpha=P(U^* \leq X^*\leq V^*)$.  \\

\cite{Shen08} investigated the asymptotic properties of $F_n$ in the particular case in which both $X^*$ and $(U^*,V^*)$ have a density. Note however that for interval sampling the couple $(U^*,V^*)$ falls on a line and, therefore, the density of the truncation pair does not exist. Recently, \cite{dUAVK19} revisited and completed the asymptotic theory for the NPMLE in the  more general setting in which covariables are present; they formally established the weak convergence of both $F_n$ and $G_n$ under primitive assumptions and they repaired several gaps and inconsistencies in \cite{Shen08}.\\

\subsection{Limitations of the NPMLE}

In practice the NPMLE may have have some limitations. When analysing a particular doubly truncated dataset, the existence or uniqueness of the NPMLE may be compromised. For instance, the convergence of the iterative algorithms proposed by \cite{Turnbull76} or \cite{Efron99} does not imply that a NPMLE actually exists. In fact, such an estimate may not exist, in which case the estimates provided by the iterative algorithm may be misleading.   \cite{Xiao19} presented a necessary and sufficient graphical condition, based on \cite{Vardi85}, to determine the existence and uniqueness of the NPMLE.
The graphical condition is based on graphs theory applied to the observed triplets $(X_i, U_i, V_i), i=1, \ldots, n$. Considering that each of these triplets represents one of the $n$ vertices of the graph $\mathcal{G}$,  a directed edge from vertex $i$ to vertex $j$ exists if and only if $X_j \in [U_i, V_i]$.  A graph $\mathcal{G}$ is strongly connected if, for any two vertices $i$ and $j$, there exists a directed path from $i$ to $j$ and a directed path from $j$ to $i$, see \cite{Xiao19} for more details. Specifically, the result is as follows.

\begin{prop}
(Xiao and Hudgens, 2019) There exists a unique NPMLE if and only if the graph $G$,  is strongly connected.
\label{prop:exuniq}
\end{prop}

A simple necessary condition for the existence and uniqueness of the NPMLE can be derived from Proposition \ref{prop:exuniq}. Let $S_j= \displaystyle \sum_{i=1}^n J_{ij}$ and $\tilde S_j= \displaystyle \sum_{i=1}^n I(U_j\leq X_i\leq V_j)$.
If $S_j=1$ or $\tilde S_j=1$ for some $j$ then the NPMLE does not exist or is not unique.

Even when existing, the NPMLE is of little use when its variance is extremely large, something with may occur for special truncation patterns. In Section 5  we provide illustrations of the referred limitations. An alternative to the NPMLE is the semiparametric approach, which is discussed in the following subsection.

\subsection{The SPMLE of the cumulative df}
\cite{Moreira110} derived the asymptotic results for the SPMLE of the cumulative df under double truncation, and concluded that it may be more efficient than the Efron–Petrosian NPMLE. On other hand, the use of the semiparametric approach circumvents the eventual limitation of the non-existence or non-uniqueness of the NPMLE. In the nonparametric approach, the NPMLE of $F$ was represented as an inverse probability weighted estimator (IPWE), where the $X_i$'s are upweighted according to their estimated sampling probabilities $G_n(X_i)$; see equation (\ref{ipwe}). In the semiparametric setting $T$ is assumed to belong to a parametric family of df's $\left\{T_{\theta }\right\} _{\theta \in \Theta }$, where $\theta $
is a vector of parameters and $\Theta $ stands for the parametric
space.   As a consequence, $G(t)$ becomes
\begin{eqnarray*}
G_{\theta}(t)=\int_{\left\{ u\leq t\leq v\right\} }T_{\theta}(du,dv).
\end{eqnarray*}
The parameter $\theta$ is estimated by the maximizer $\hat \theta$
of the conditional likelihood of the $(U_i,V_i)$'s given the
$X_i$'s, that is,

\begin{eqnarray*}
\mathcal{L}_1^{*}(\psi)\equiv \mathcal{L}_1^{*}(\theta)= \prod_{i=1}^{n}\frac{g_{\theta}(U_{i},V_{i})}{G_{\theta}(X_{i} )}
\end{eqnarray*}
where $g_{\theta}(u,v)=\frac{\partial^2}{\partial u \partial v}P(U^{\ast }\leq u,V^{\ast }\leq v)=T_{\theta }(du,dv)$ stands
for the joint density of $\left( U^{\ast },V^{\ast }\right) $. In the case of interval sampling one has $V^*=U^* + \tau$  for some constant $\tau >0$ (the sampling interval width) so the joint density of $(U^*,V^*)$ does not exist; in this case one should rather use $g_{\theta}(u,v)=\frac{\partial}{\partial u }P(U^{\ast } \leq u)= L_{\theta}(du)$
 and $G_{\theta}(x)=L_{\theta}(x)-L_{\theta}(x-\tau)$ in $\mathcal{L}_1^{*}(\theta)$, where $L_{\theta}(.)$ stands for the df of the parametric model assumed for $U^*$.

%\todo{Añadí este Remark}; see \textit{Remark} 1 for the definition of $g$ for interval sampling case).\\

\bigskip

Once $\theta$ is estimated, the SPMLE of $F$ is
introduced through

\begin{equation*}
F_{\widehat \theta}(x)=\alpha_{\widehat \theta} \int_{a_F}^x
\frac{F_n^*(dt)}{G_{\widehat \theta}(t)}, \label{semip_est}
\end{equation*}
where $\alpha_{\widehat \theta}=(\int_{a_F}^{\infty }G_{\widehat
\theta}^{-1}(t)F^*_{n}(dt))^{-1}$.  \cite{Moreira110} established the
asymptotic normality of both $\hat \theta$ and $F_{\widehat
\theta}$.  As a drawback, the semiparametric estimator requires preliminary
 specification of a parametric family, which may eventually introduce a bias component when it is far away from reality \citep{Moreira110}.\\

%\todo{Añadí este Remark} \noindent \textit{Remark} 1\\
%In the case of interval sampling $V^*=U^* + \tau$,  for some fixed constant $\tau >0$,  we rather have $g_{\theta}(u,v)=\frac{\partial^2}{\partial u \partial v}P(U^{\ast } \leq u)= L_{\theta}(du)$ and $G_{\theta}(x)=L_{\theta}(x)− L_{\theta}(x−\tau)$, where $L_{\theta}(.)$ is the df of the parametric model assumed for $U^*$.

Following (\ref{npkde1}), we introduce the semiparametric kernel estimator for the hazard function as

\begin{equation}
\lambda_{\widehat \theta,h}(x)=\int K_{h}(x-t)\frac{F_{\widehat \theta}(dt)}{1-F_{\widehat \theta}(t^-)}=
\alpha_{\widehat \theta} \frac1n{\sum_{i=1}^{n}K_{h}(x-X_{i})\frac{G_{\widehat \theta}(X_{i})^{-1}}{1-F_{\widehat \theta}(X_i^{-})}}.\label{spkde}
\end{equation}

\section{Hazard rate estimators: main properties}
%\label{Section3}

\bigskip
Both $G_{n}$ and $G_{\widehat{\theta }}$ defined in the previous section are $\sqrt{n}$%
-consistent estimators of $G$. For $G_{\widehat \theta } $ this
follows from the $\sqrt{n}$-consistency of $\widehat \theta$,
provided that $G_{\theta}$ is a smooth function of $\theta$
\citep{Moreira110}. For $G_{n}$, the result may be obtained by
noting that
\begin{eqnarray*}
G_{n}(x)=\alpha _{n}^{-1}\int \int_{\left\{ u\leq x\leq v\right\} }\frac{T_{n}^{\ast }(du,dv)}{\int_{\left\{ u\leq t\leq v\right\}
}F_{n}(dt)}
\end{eqnarray*}

\noindent and

\begin{eqnarray*}
\alpha _{n}=\int \int \frac{T_{n}^{\ast }(du,dv)}{\int_{\left\{ u\leq t\leq v\right\} }F_{n}(dt)},
\end{eqnarray*}

\noindent where $T_{n}^{\ast }$ is the ordinary empirical df of the truncation
times. Hence, $\sqrt{n}$-consistency of $G_{n}$ is a consequence of
that of $F_{n}$ and $T_{n}^{\ast }$; see \cite{dUAVK19} for formal derivations. Since both
$G_{n}$ and $G_{\widehat{\theta }}$ approach to $G$ at a $\sqrt{n}
$-rate, which is faster than the nonparametric rate $\sqrt{nh}$, the
asymptotic properties of $\lambda_{h}$ and $\lambda_{\widehat{\theta },h}$ are expected to
be the same, and will coincide with those  of the artificial estimator based on
the true $G$. The same heuristic argument suggests that $F_{n}$, $F_{\widehat{\theta }}$, $\alpha_{n}$ and $\alpha_{\widehat{\theta }}$ can be replaced by their limits $F$ and $\alpha$ for asymptotic analysis. This is in parallel with the approach in \cite{Moreira2012} for density estimation.

\bigskip

Introduce the asymptotically equivalent version of $\lambda_{h}$ and $\lambda_{\widehat{\theta },h}$ through
\begin{equation}
\overline{\lambda}_{h}(x)=
\alpha \frac1n{\sum_{i=1}^{n}K_{h}(x-X_{i})\frac{G(X_{i})^{-1}}{1-F(X_i^{-})}}  . \label{aekde}
\end{equation}

%where
%\begin{eqnarray*}
%\overline{\Lambda}_{n}(x)=%
%\alpha \frac1n{\sum_{i=1}^{n}\frac{G(X_{i})^{-1}I_{[X_i \leq x]}}{{1-F(X_i^{-})}}}.
%\end{eqnarray*} \todo{Respuesta: No se utiliza la función cumulativa para nada!}

\noindent As discussed, under regularity one has $(nh)^{1/2}(\lambda_h(x)-\bar \lambda_h(x))=o_P(1)$ and $(nh)^{1/2}(\lambda_{\hat \theta,h}(x)-\bar \lambda_h(x))=o_P(1)$. The function $G(.)$ may be constant; for example, this happens when $V^*-U^*$ is degenerated (that is, with interval sampling) provided that the left-truncation time $U^*$ is uniformly distributed in a suitable interval. In such a case, the correction for truncation vanishes and the usual kernel hazard estimators for complete data is obtained. This is not surprising, since a constant $G$ indicates that there is no sampling bias. In general, however, the function $G$ will not be flat and the correction for double truncation becomes relevant. In the next result we establish the strong consistency and the
asymptotic normality of $\overline{\lambda}_{h}(x)$. Throughout this Section we implicitly assume
$G(x)>0$ for each $x$ in the support of $X^*$. Note that this condition is needed to ensure the identifiability of $F$ along its whole support.

\bigskip
%\textbf{Theorem 2.1.}
\begin{thm}(i) If $K$ is bounded on a compact support, $h$
is such that $\sum_{n=1}^{\infty }\exp (-\eta hn)<\infty $ for each
$\eta >0$, $G$
is continuous at $x$, and $x$ is a Lebesgue point of $\lambda$, then $\overline{\lambda}%
_{h}(x)\rightarrow \lambda(x)$ with probability 1.

(ii) If, in addition to the conditions in (i), $K$ is an even function, $%
h=o(n^{-1/5})$, $G^{-1}/(1-F)\lambda$ has a second derivative which is bounded
in a neighbourhood of $x$, and $\lambda(x)>0$, then
\begin{eqnarray*}
\left( nh\right) ^{1/2}\left( \overline{\lambda}_{h}(x)-\lambda(x)\right)
\rightarrow N(0,\alpha \frac{G(x)^{-1}}{1-F(x)}\lambda(x)R(K))
\end{eqnarray*}
in distribution, where $R(K)=\int K(t)^{2}dt$.\\
\label{Theorem1}
\end{thm}

\textit{Proof}. See the Appendix.

\bigskip
The asymptotic mean and variance of (\ref{aekde}) are given in the
following result. We refer to the following standard regularity
assumptions.

\begin{itemize}

\item [](A1) The kernel function $K$ is a density function with $\int tK(t)dt=0$, $\mu _{2}(K)=\int t^{2}K(t)dt<\infty $, and $R(K)=\int
K(t)^{2}dt<\infty $.

\item [](A2) The sequence of bandwidths $h=h_{n}$ satisfies $h\rightarrow 0$ and $nh\rightarrow \infty $ as $n\rightarrow \infty $.

\item [](A3) The functions $\lambda$ and $\frac{G^{-1}}{1-F}\lambda$ are twice continuously differentiable around $x$.
\end{itemize}

%\textbf{Theorem 2.2.}
\begin{thm}Under (A1)-(A3) we have, as $n\rightarrow \infty
$,
\begin{eqnarray*}
E\left[ \overline{\lambda}_{h}(x)\right] =\lambda(x)+\frac{1}{2}h^{2}\lambda^{\prime
\prime }(x)\mu _{2}(K)+o(h^{2}),
\end{eqnarray*}
\begin{eqnarray*}
Var\left[ \overline{\lambda}_{h}(x)\right] =(nh)^{-1}\alpha
\frac{G(x)^{-1}}{1-F(x)}\lambda(x)R(K)+o((nh)^{-1}).
\end{eqnarray*}
\label{Theorem2}
\end{thm}

\begin{proof}

%\todo {Cambié la frase al inicio. Continua parecido al EJE}
The proof follows standard steps, see e.g. \cite{Wand95}. A second-order
Taylor expansion of $\lambda$, respectively of $G^{-1}\lambda/(1-F))$,  around $x$ is used, and the assumptions on
the kernel and
the bandwidth are enough to conclude. Details are omitted.
\end{proof}

\bigskip

From Theorem \ref{Theorem2} it can be seen that the asymptotic variance of $%
\overline{\lambda}_{h}(x)$ is affected by the double truncation issue, while the asymptotic bias is that of the complete data case (no truncation). The variance can be smaller or larger than the one obtained without truncation (constant $G$) depending on the particular $x$ value. This is intuitive, since the sampling bias due to the double truncation may result in an oversampling of certain lifetime values, while other may be undersampled.\\

%\todo {Cambiado para diferenciar de EJE}
The global error of $\overline{\lambda}%
_{h}$ can be measured
through the mean integrated squared error (MISE), namely
\begin{eqnarray*}
MISE(\overline{\lambda}_{h})=\int MSE(\overline{\lambda}_{h}(x))dx,
\end{eqnarray*}
where
\begin{eqnarray*}
MSE(\overline{\lambda}_{h}(x))=\left[ E[\overline{\lambda}_{h}(x)]-\lambda(x)\right]
^{2}+Var\left( \overline{\lambda}_{h}(x)\right) .
\end{eqnarray*}
Under regularity, the following
asymptotic expression for the $MISE(\overline{\lambda}_{h})$ is immediately derived from the previous results:
\begin{eqnarray*}
AMISE(\overline{\lambda}_{h})=\frac{1}{4}h^{4}R\left( \lambda^{\prime \prime
}\right) \mu _{2}(K)^{2}+(nh)^{-1}\alpha R(K)\int \frac{G^{-1}}{1-F}\lambda
\end{eqnarray*}
where $R\left( \lambda^{\prime \prime }\right) =\int \left( \lambda^{\prime
\prime }\right) ^{2}$. Minimization of $AMISE(\overline{\lambda}_{h})$ w.r.t. $h$
leads to the asymptotically optimal bandwidth
\begin{eqnarray*}
h_{AMISE}=\left[ \frac{\alpha R(K)\int \frac{G^{-1}}{1-F}\lambda} {R\left( \lambda^{\prime
\prime }\right) \mu _{2}(K)^{2}}\right] ^{1/5}n^{-1/5}.
\end{eqnarray*}
Of course, this expression depends on unknown quantities that must
be estimated in practice. There exist several criteria to select the
bandwidth from the data at hand. In the Appendix, a least-squares cross-validation (LSCV) bandwidth selector is derived. The cross-validation bandwidth is used in the real data analyses of Section 5.%\ref{Section4}.

%Because of the $\sqrt{nh}$-equivalence
%between $G_n$,$G_{\widehat{\theta }}$ and $G$, the same asymptotic
%expression will hold for $\lambda_h$ and $\lambda_{\widehat{\theta }%
%,h}$ under proper conditions.

%\bigskip

%Interestingly, H\"{o}lder's inequality gives $\alpha \int \frac{G^{-1}}{1-F^-}\lambda \geq
%1$, which indicates that the global error when estimating the
%density in the doubly truncated scenario is at least as large as
%that pertaining to the no truncated situation. This does not mean
%that for a particular $x$ the MSE of $\overline f_h(x)$ may not be
%smaller than in the i.i.d. situation, since $\alpha G(x)^{-1}f(x)<1$
%may happen.

\section{Simulations}
\label{Section4}

%\todo {Cambiar para diferenciar de EJE}
In this section we illustrate the finite sample behavior of the purely nonparametric estimator $\lambda_h(x)$ and the
semiparametric estimator $\lambda_{\widehat \theta, h}(x)$ through simulations. We analyze the
influence of the bandwidth in the MISE of the estimators, and we measure the amount of efficiency
which is gained by using  the semiparametric information.

\bigskip

We simulate $U^*$ independently of $X^*$ and then we take $V^*=U^*
+ \tau$ for some constant $\tau >0$. Such scenario represents interval sampling, and follows the spirit
of the two real data examples presented in Section 5, with terminating events (acute coronary syndrome or AIDS diagnosis) falling between two specific dates. Different models
are simulated.
We take $\tau=0.25$ and $U^* \sim U(0,1)$, $X^* \sim 0.75 Beta(3/4,1)+0.25$ for Model 1 and $U^* \sim U(0,1)$, $X^* \sim 0.75N(0.5,0.15)+0.25$ for Model 2.
In Model 3 we take  $U^* \sim U(0.25,1)$, $X^* \sim 0.75 Beta(3/4,1)+0.25$ and $V^* = U^* + \tau$, considering a decreasing sequence of widths for the sampling interval: $\tau=0.25$, 0.15 and 0.10 (Model 3.1, 3.2, 3.3 respectively). As parametric information on the truncation distribution we take a $Beta(\theta _{1},1)$ model for $
U^{\ast }$. For each model, we simulate 1000 trials with final sample size $n=100$, 250, or 500.

%We point out that, due to the random truncation, in Models 3.1 to 3.3  relatively small and moderate values of the lifetime are less probably observed, while in Models 1 and 2 there is no observational bias on $X^*$ (i.e. $G(.)=1$; see Remark 2.1 in \cite{Moreira2012}).\\

The functions $G$ corresponding to the aforementioned models,  based on a Monte Carlo approximation from a single sample of size $n = 20,000$, are shown in Figure \ref{ObservationalBias}.
The depicted functions indicate that small  values of the variable of interest $(X^*)$ are observed with a relatively small probability in the last three models, while there is no observational bias in Models 1 and 2 ($G$ remains constant).

\begin{figure}[ht]
\begin{minipage}[b]{0.48\linewidth}
\includegraphics[width=\textwidth]{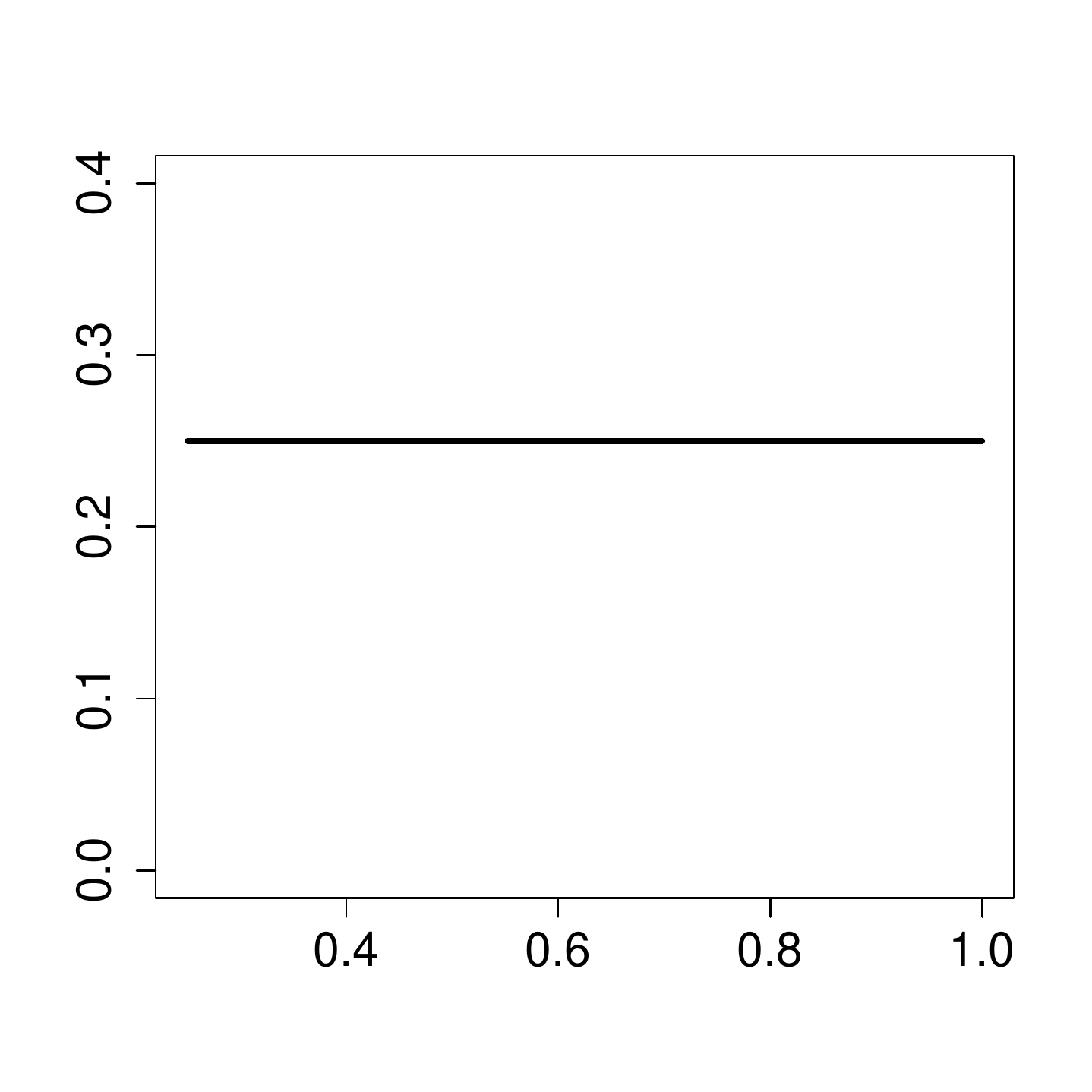}
\end{minipage}
\begin{minipage}[b]{0.48\linewidth}
\includegraphics[width=\textwidth]{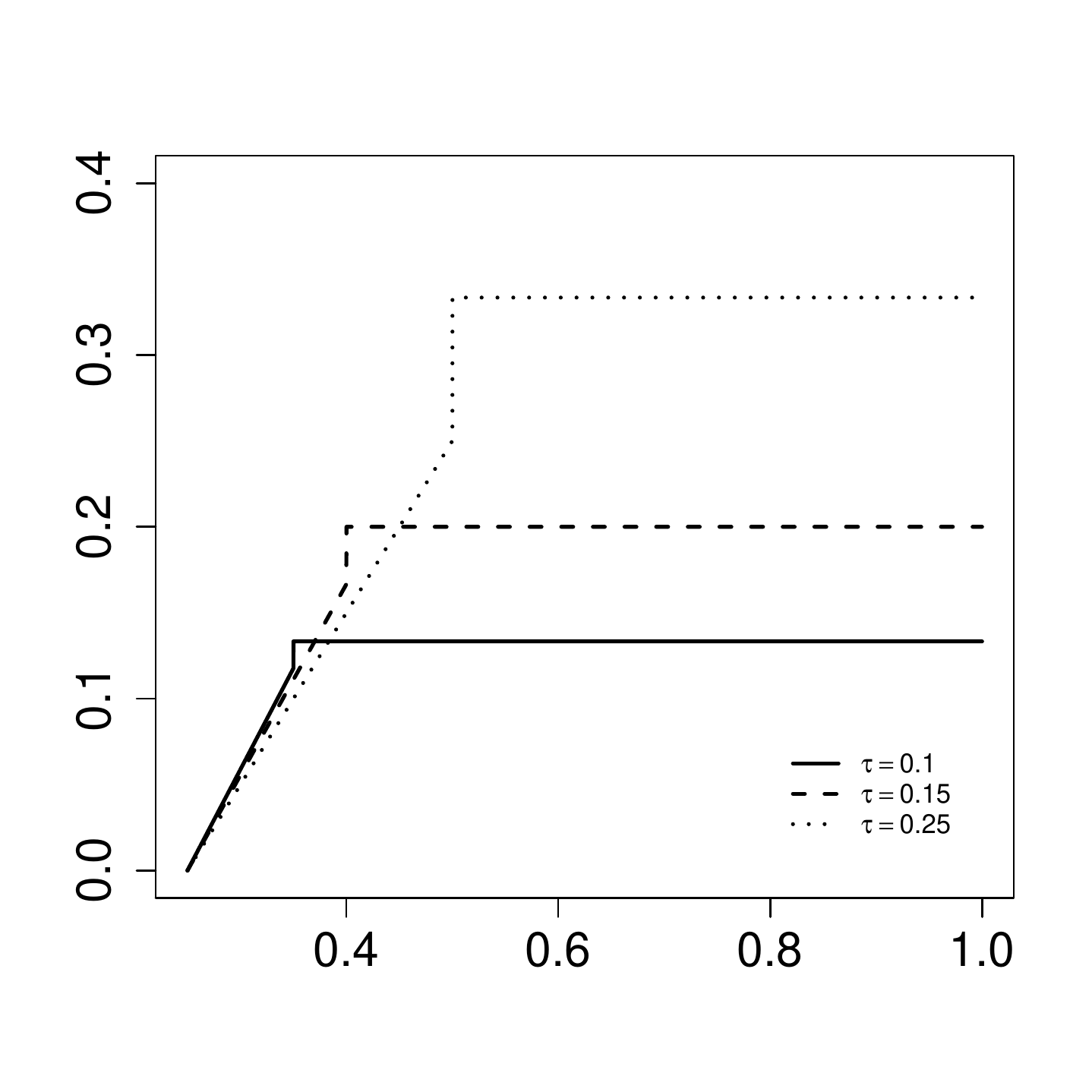}
\end{minipage}
\caption{Function $G$ for the simulated models. Left: Models 1 and 2. Right: Model 3.1, 3.2 and 3.3.}
\label{ObservationalBias}
\end{figure}

\bigskip

%\todo{Un poco cambiado}

%\todo{Una vez que presentamos las graficas de la función G, podemos eliminar este párrafo?} This means that, for each trial, the number of simulated data is much larger than $n$, actually $N\approx n\alpha^{-1}$ are needed on average, where recall that $\alpha$ stands for the proportion of no truncation. For the simulated models, the proportion of truncation ranges between 44\% and 88\%. More specifically, the following right and left truncation proportions occur: $37\%$ (right) and $44\%$ (left) for Model 1.1; $38\%$ and $40\%$ for Model 1.2; $53\%$ and $22\%$ for Model 2.1; and $45\%$ and $28\%$ for Model 2.2 \cite{Moreira10}.\\

In Table \ref{Table1} we report the optimal bandwidths (in the sense
of the MISE) and the corresponding minimum MISEs for the
nonparametric estimator and the semiparametric estimator. The theoretical
MISE function is approximated by the average of the integrated squared error (ISE) along the
$M=$1000 trials, namely

\bigskip

$\overline{ISE}$($\lambda_h)=\dfrac 1 M \displaystyle \sum_{m=1}^M \int
\left(\lambda_h^m-\lambda \right)^2$\hspace{0.2 cm},\\
\hspace{0.2 cm}$\overline{ISE}$($\lambda_{\widehat \theta,h})=\dfrac 1 M \displaystyle \sum_{m=1}^M \int \left(\lambda_{\widehat \theta,h}^m-\lambda\right)^2$,\\

%\todo{en los dous integrales de la formula anterior, No falta la orden a que integrar??}

\noindent where $\lambda_h^m$ and $\lambda_{\widehat \theta,h}^m$ are respectively the
nonparametric and the semiparametric estimators when based on the
$m$-th Monte Carlo trial.

\bigskip

From Table \ref{Table1} it is seen that the optimal bandwidths and
the MISEs decrease when increasing the sample size; besides, the
semiparametric estimator has an error which is smaller than that
of the nonparametric estimator. It is also seen that the
optimal bandwidths for the semiparametric estimator are generally smaller than
those of the nonparametric estimator, according to the extra amount
of information. As the sample size grows, the relative efficiency of
the nonparametric estimator approaches to one; this is in agreement
to the asymptotic equivalence of the semiparametric and the
nonparametric hazard estimators discussed in Section 2. Results for Models 3.1 to 3.3 reveal that the MISE increases as the width of the sampling interval decreases, due to the large variance attached to a small observational window. The semiparametric estimator however behaves much better than the nonparametric estimator, offering moderate MISE values even in the extreme situation with $\tau=0.1$.

\bigskip

In Table \ref{Table2} we report the biases and the variances of the nonparametric and semiparametric hazard estimators at some selected time points, corresponding to the quartiles of $\lambda$, for sample sizes n=100, 500, along the 1,000 Monte Carlo trials. It is seen that the squared bias is always of a smaller order when compared to the variance, so the resulting mean squared errors (MSEs) are mainly determined by the dispersion of the estimates. For all the cases, the local MSEs of the semiparametric estimator are smaller than those pertaining to the nonparametric estimator, as expected.

%It is also seen that  as the observational window decrease, i.e, moving from Model 3.1 to 3.3,  the MSE does not increase  in all presented deciles, as in previous table for the global error. This is  beacuse we are measuring the error in a particular decile, being representative of that specif point. PODEMOS DECIR ALGO MAS SOBRE ESTO? QUEDA ALGO CONTRAINTUITIVO ASI

\begin{table}[h]
\centering
%\scriptsize
\caption {Optimal bandwidths ($h_{opt}$) and minimum MISEs of the
hazard estimators: nonparametric estimator ($NP$) and semiparametric
estimator ($SP$). Averages along 1000 trials of a sample size $n$.}
\label{Table1}
\begin{tabular}{c c c c c c}
Model& n& \multicolumn{2}{c}{$h_{opt}$}  & \multicolumn{2}{c}{$MISE(h_{opt})$}\\
%\cline{3-4}
%\cline{5-6}
&&NP&SP&NP&SP\\
\hline
&100&0.100&0.090&7.723&4.649\\
1 &250&0.020&0.020&3.432&2.381\\
&500&0.020&0.020&1.532&1.031\\
\hline
&100&0.091&0.078&9.392&4.713\\
2 &250&0.090&0.076&5.473&2.763\\
&500&0.070&0.070&2.483&1.754\\
\hline
&100&0.073&0.066&8.726&6.026\\
3.1 &250&0.014&0.014&7.788&5.489\\
&500&0.011&0.011&4.290&3.823\\
\hline
&100&0.093&0.068&15.968&10.759\\
3.2&250&0.014&0.014&11.595&9.925\\
&500&0.011&0.011&5.869&4.150\\
\hline
&100&0.18&0.066&28.651&10.808\\
3.3 &250&0.115&0.014&17.0180&10.043\\
&500&0.011&0.011&8.517&4.162\\
\hline
\end{tabular}
\end{table}

\begin{table}[h]
\centering
%\scriptsize
\caption {Biases ($ \times 10^3$) and variances ($ \times 10$) of the
nonparametric estimator ($NP$) and semiparametric
estimator ($SP$) at the quartiles of $F$, for sample sizes n=100, 500, along 1000 Monte Carlo trials.}
\label{Table2}
\begin{tabular}{c c c r r r r}
& && \multicolumn{2}{c}{$NP$}  & \multicolumn{2}{c}{$SP$}\\
\cline{4-7}
Model& n&x& Bias&Var  & Bias&Var\\
\hline
%\cline{5-6}
%&&&EP&SP&EP&SP\\
&&$q_{.25}$&$37.20$&$2.44$&$149.80$&$2.01$\\
&100&$q_{.50}$&$72.20$&$9.99$&$468.10$&$8.31$\\
&&$q_{.75}$&$110.30$&$37.95$&$155.60$&$32.75$\\
\cline{3-7}
1&&$q_{.25}$&$20.34$&$1.65$&$8.59$&$0.44$\\
&500&$q_{.50}$&$28.69$&$3.04$&$16.24$&$0.75$\\
&&$q_{.75}$&$72.09$&$6.58$&$63.28$&$3.06$\\
\hline
&&$q_{.25}$&$62.29$&$6.93$&$69.50$&$4.86$\\
&100&$q_{.50}$&$66.05$&$11.98$&$13.30$&$2.86$\\
&&$q_{.75}$&$726.60$&$44.92$&$325.00$&$16.95$\\
\cline{3-7}
2&&$q_{.25}$&$20.16$&$1.45$&$19.65$&$1.22$\\
&500&$q_{.50}$&$8.70$&$3.85$&$0.02$&$0.69$\\
&&$q_{.75}$&$38.00$&$14.02$&$25.50$&$5.37$\\
\hline
&&$q_{.25}$&$12.66$&$7.37$&$5.88$&$2.39$\\
&100&$q_{.50}$&$130.62$&$10.86$&$34.18$&$2.88$\\
&&$q_{.75}$&$299.28$&$24.22$&$159.63$&$11.04$\\
\cline{3-7}
3.1&&$q_{.25}$&$2.37$&$1.40$&$2.47$&$0.56$\\
&500&$q_{.50}$&$17.93$&$2.07$&$9.66$&$0.68$\\
&&$q_{.75}$&$67.40$&$5.15$&$67.18$&$2.94$\\
\hline
&&$q_{.25}$&$157.70$&$32.72$&$1792.10$&$1.47$\\
&100&$q_{.50}$&$435.79$&$39.31$&$959.42$&$3.53$\\
&&$q_{.75}$&$882.44$&$64.47$&$1070.40$&$15.54$\\
\cline{3-7}
3.2&&$q_{.25}$&16.050&38.89&115.50&0.17\\
&500&$q_{.50}$&$82.48$&$0.68$&$82.48$&$0.68$\\
&&$q_{.75}$&$188.52$&$3.26$&$188.52$&$3.26$\\
\hline
&&$q_{.25}$&$2731.90$&$34.94$&$1705.30$&$0.88$\\
&100&$q_{.50}$&$6355.11$&$41.19$&$947.58$&$3.43$\\
&&$q_{.75}$&$2796.60$&$37.41$&$2553.54$&$0.06$\\
\cline{3-7}
3.3&&$q_{.25}$&$32.41$&$8.87$&$46.79$&$0.25$\\
&500&$q_{.50}$&$114.02$&$12.00$&$50.42$&$57.34$\\
&&$q_{.75}$&$296.83$&$24.30$&$120.98$&$2.79$\\
\hline
\end{tabular}
\end{table}

In Figures \ref{Teoric_Model21} to \ref{Teoric_Model233} we report for each simulated model: (i) the target hazard function together with its semiparametric and nonparametric estimators averaged along the 1000 Monte Carlo trials (bottom row); and (ii) the ratio between the MISEs of the semiparametric and the nonparametric estimators along a grid of bandwidths (top row). From these Figures \ref{Teoric_Model21} to \ref{Teoric_Model233} several interesting features can be seen. First, for each given smoothing degree, the MISE of the semiparametric estimator is less than that of the nonparametric estimator; the relative benefits of using the semiparametric information are more clearly seen when working with relatively smaller bandwidths, when the variance component of the MISE is larger. This illustrates how the semiparametric estimator achieves a variance reduction w.r.t. the NPMLE.  Also importantly, we see that the ratios of the MISEs approach to one as the sample size increases.  This was expected, since (as discussed in Section 3) both estimators are asymptotically equivalent.

%Interesting is that, moving from Model 3.1 to Model 3.3, in which the variance increases, the ratio has more difficulty to move to 1, being again in accordance with the aforementioned advantage of the semiparametric estimator.  % The averaged estimators depicted in Figures \ref{Teoric_Model21} to \ref{Teoric_Model233} reveal that the semiparametric estimator fits better the target than its nonparametric competitor when the sample size is moderate. \\

\begin{figure}[ht]
\begin{minipage}[b]{0.32\linewidth}
\includegraphics[width=\textwidth]{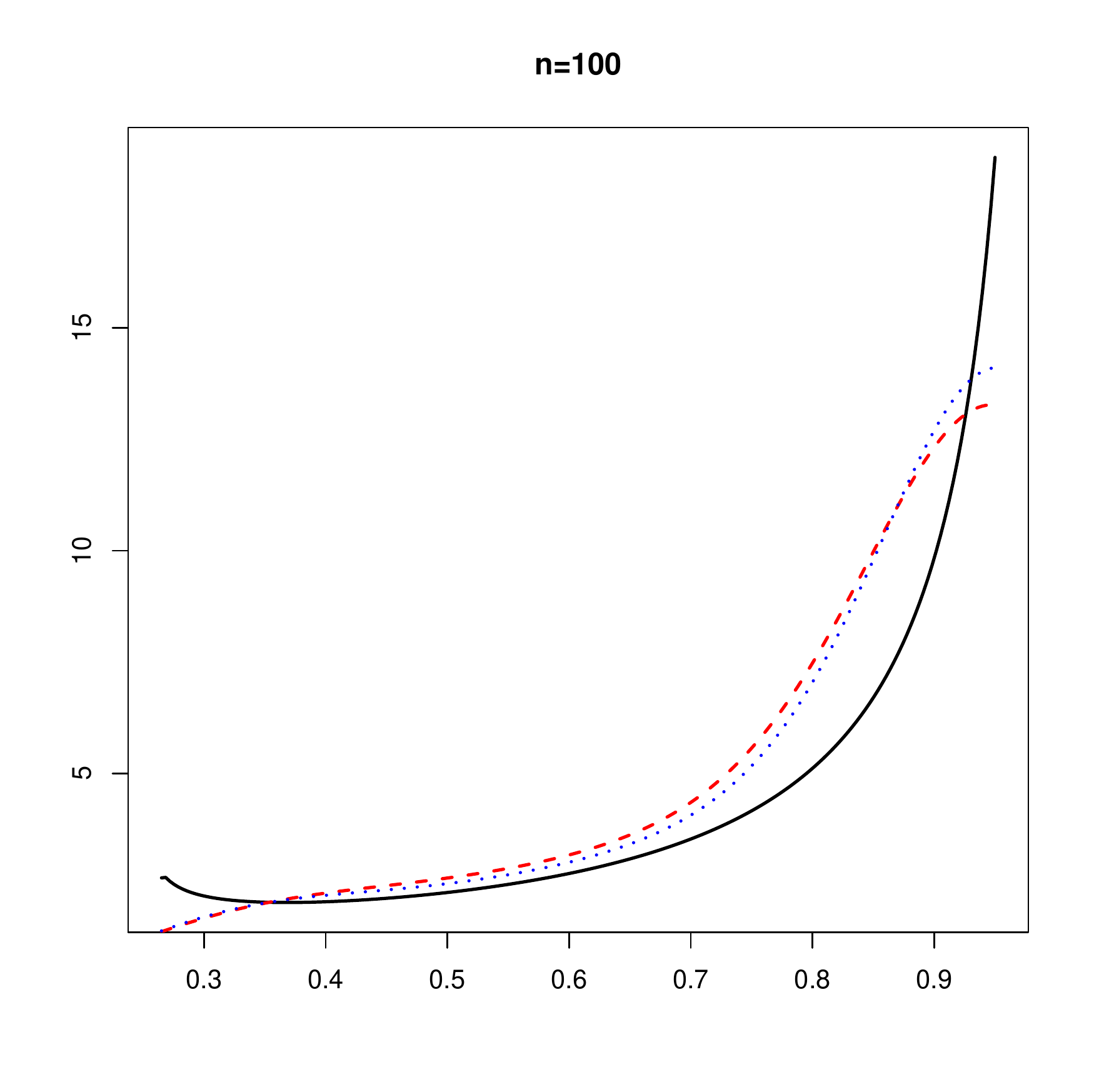}
\end{minipage} \hfill
\begin{minipage}[b]{0.32\linewidth}
\includegraphics[width=\textwidth]{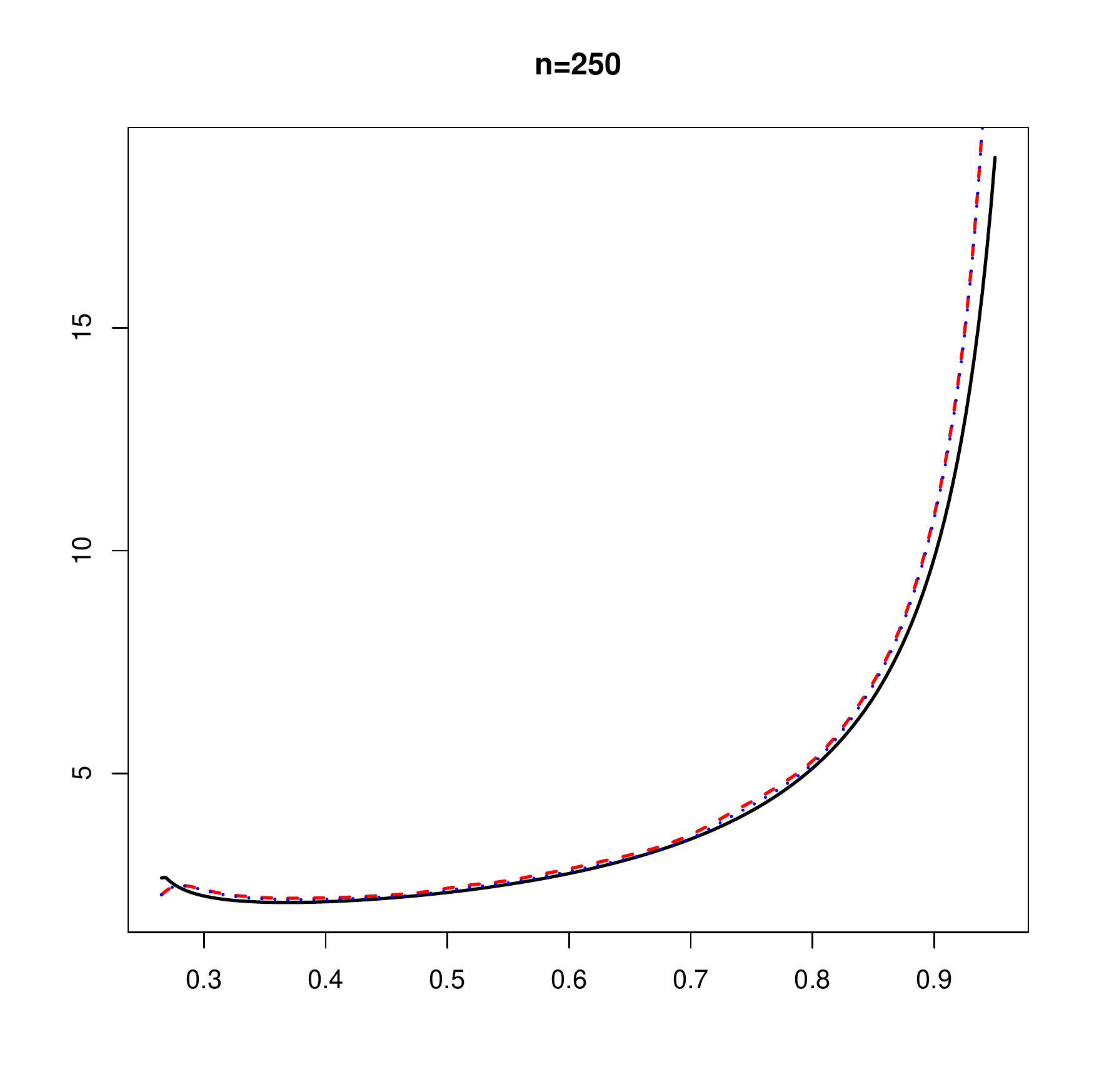}
\end{minipage} \hfill
\begin{minipage}[b]{0.32\linewidth}
\includegraphics[width=\textwidth]{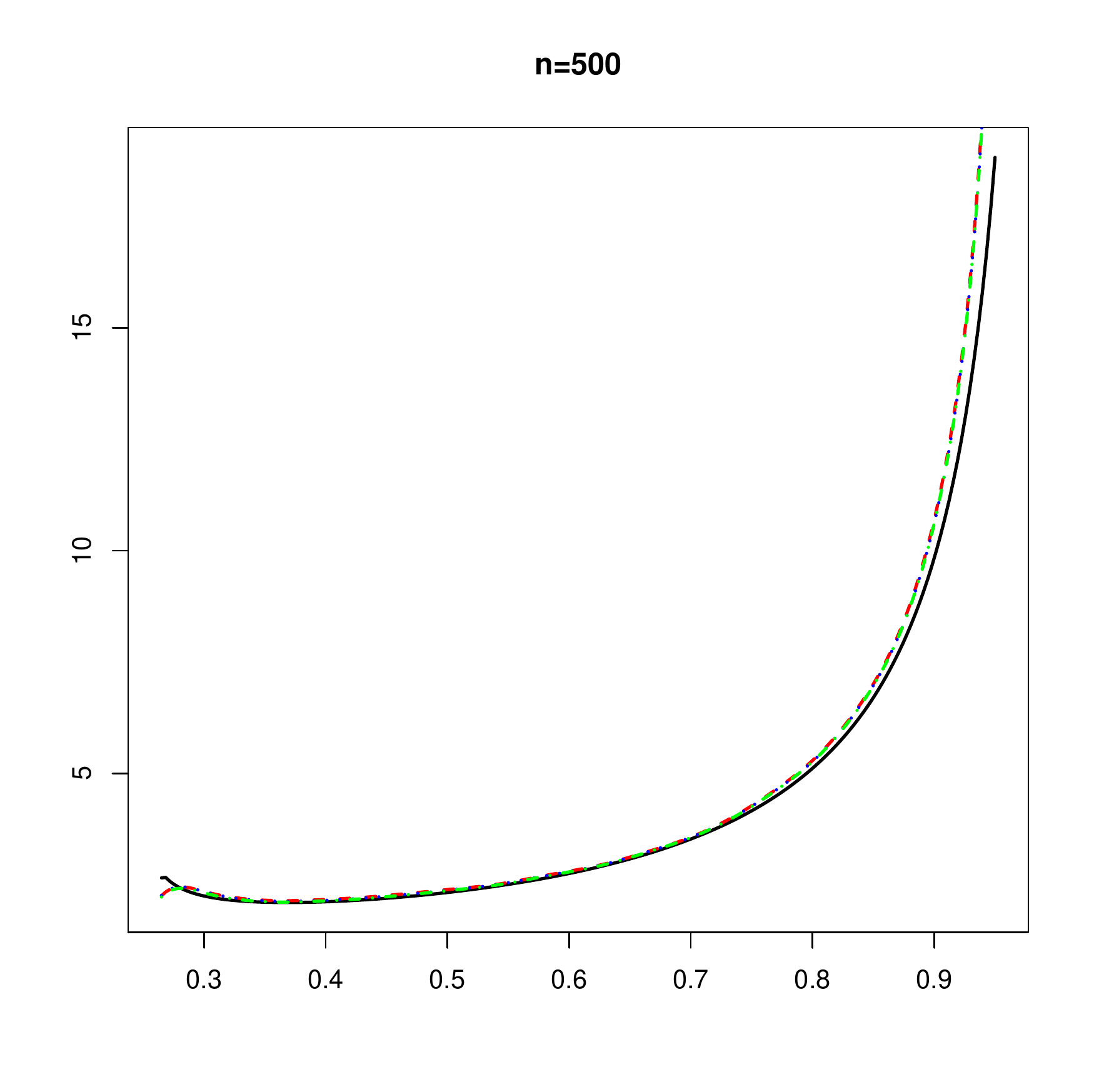}
\end{minipage}
\begin{minipage}[b]{0.32\linewidth}
\includegraphics[width=\textwidth]{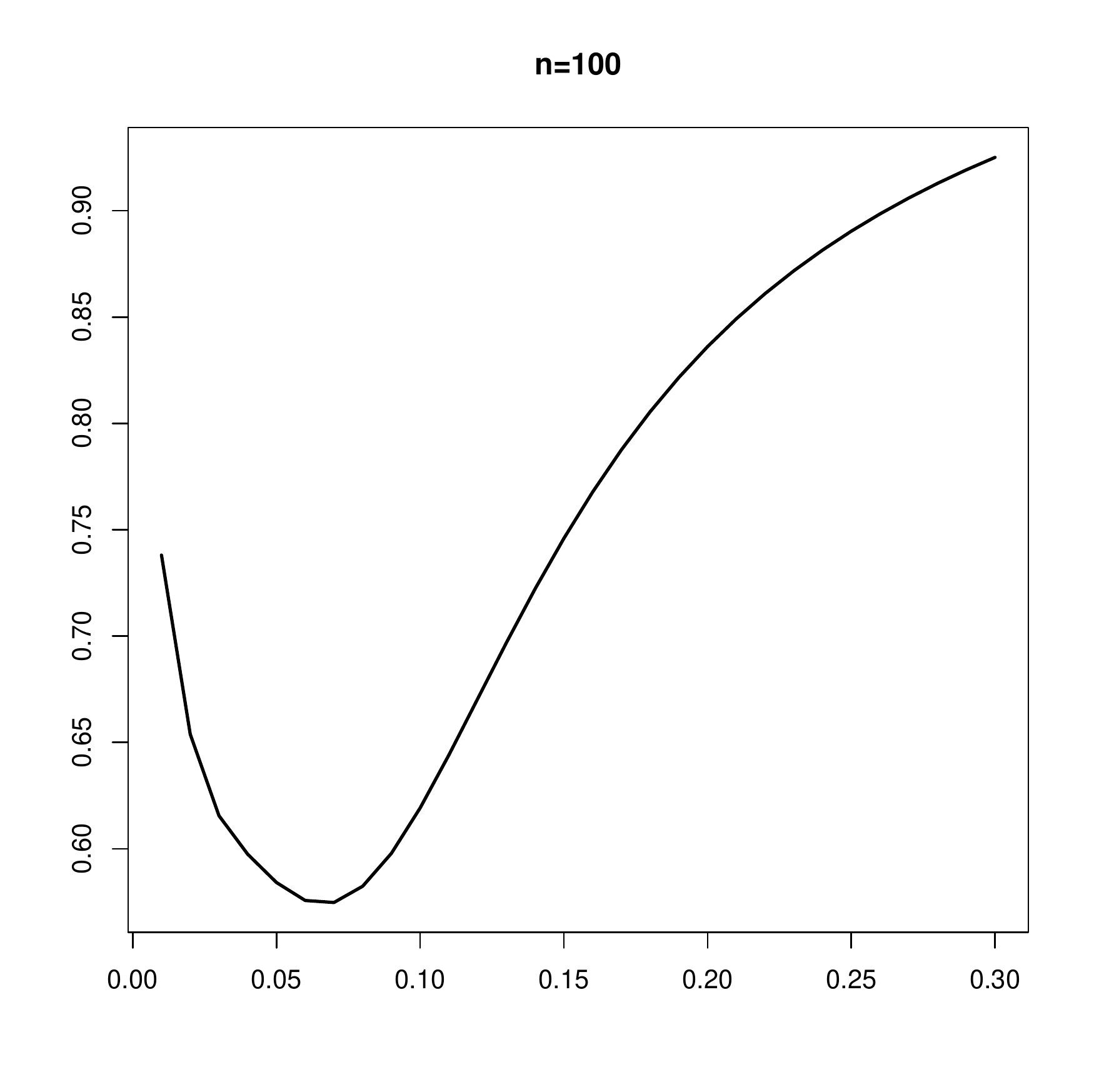}
\end{minipage} \hfill
\begin{minipage}[b]{0.32\linewidth}
\includegraphics[width=\textwidth]{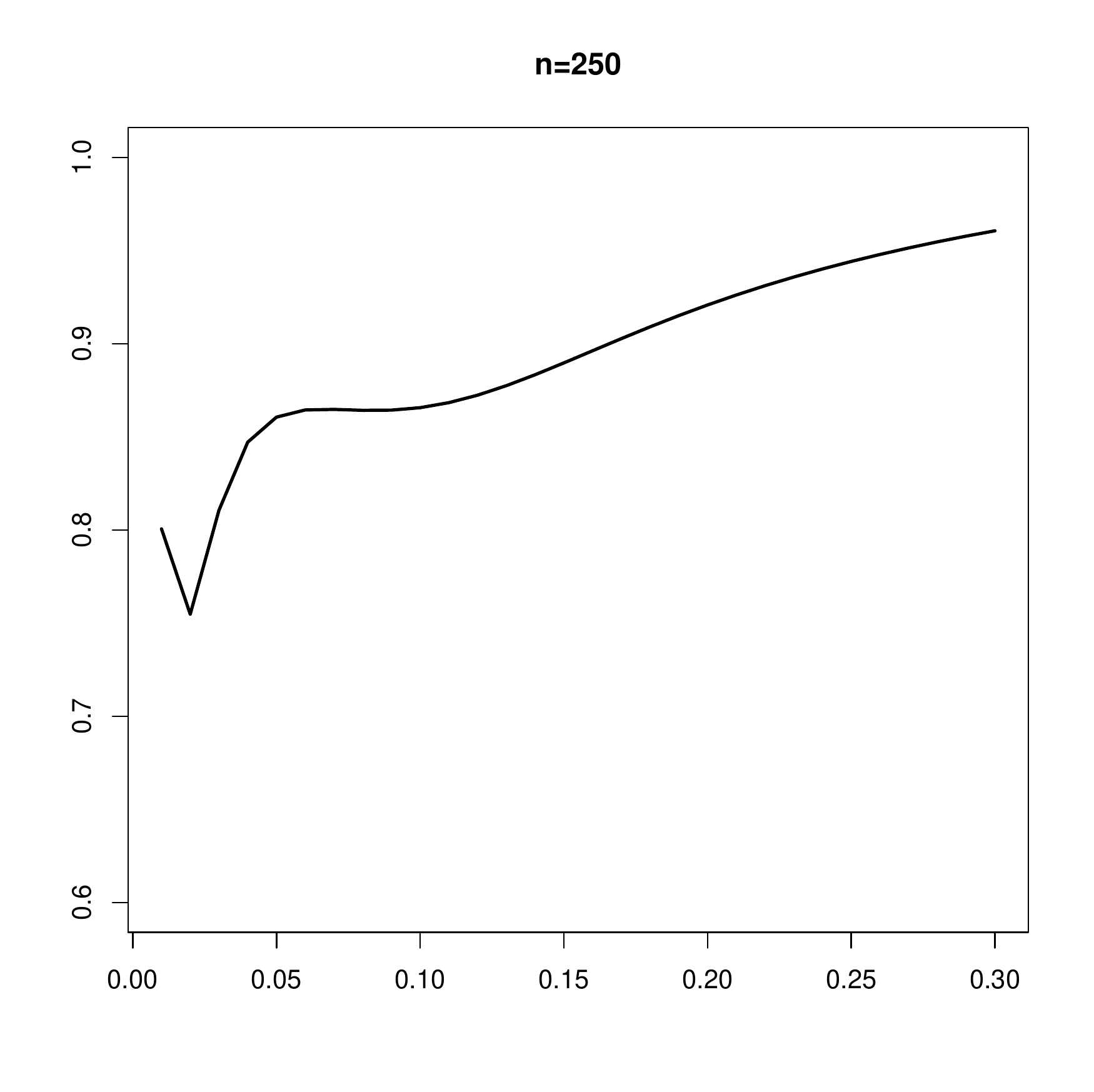}
\end{minipage} \hfill
\begin{minipage}[b]{0.32\linewidth}
\includegraphics[width=\textwidth]{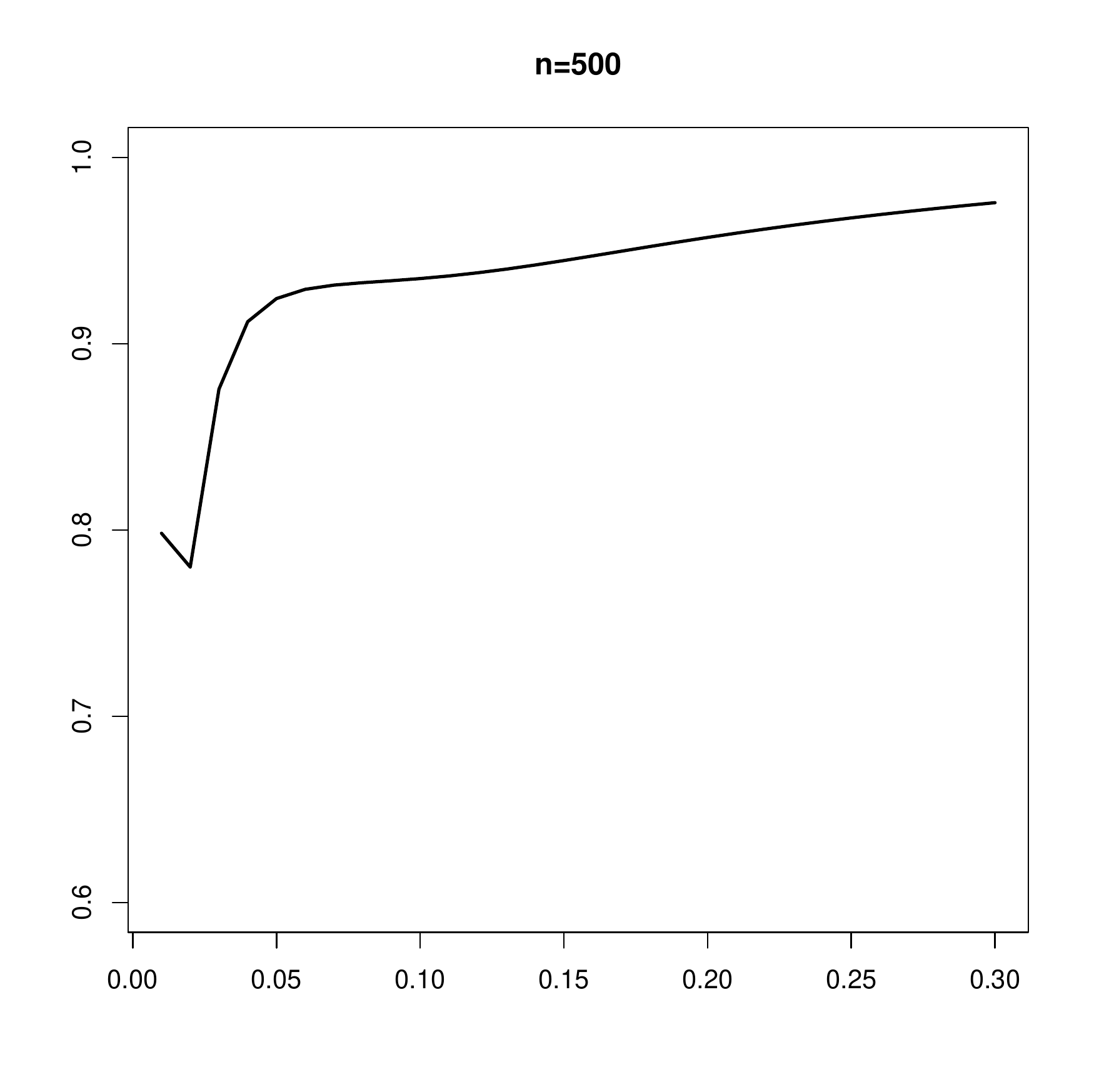}
\end{minipage}
\caption{(i) the target hazard (solid line) together with
its semiparametric (dashed line), nonparametric (dotted line) and naive (dashed-dotted line for n=500) estimators averaged along the
1000 Monte Carlo trials for Model 1 (top row); (ii) The ratio between the MISE's of the semiparametric and the nonparametric estimators along a grid of bandwidths (bottom row). }
\label{Teoric_Model21}
\end{figure}

\begin{figure}[ht]
\begin{minipage}[b]{0.32\linewidth}
\includegraphics[width=\textwidth]{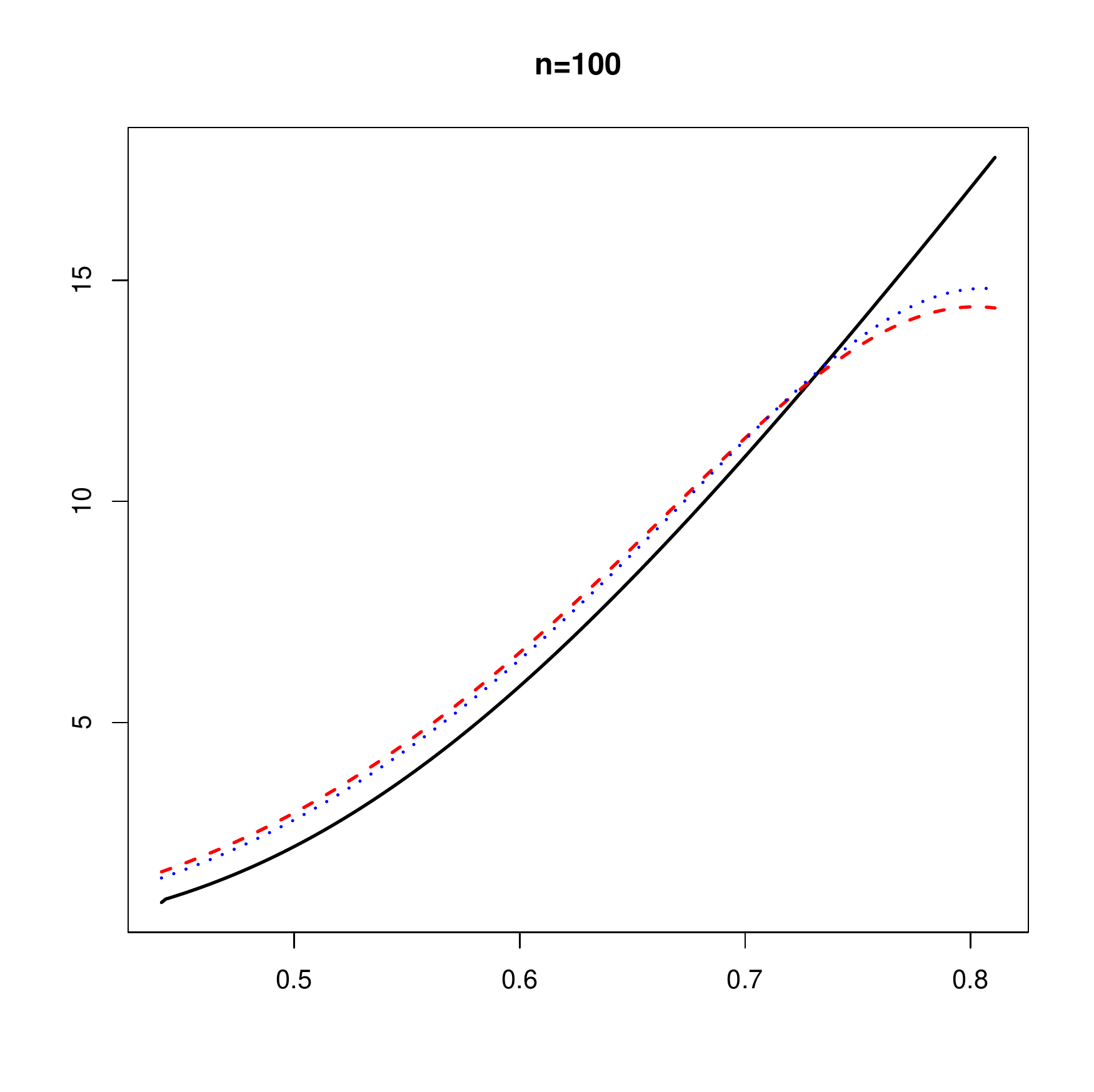}
\end{minipage} \hfill
\begin{minipage}[b]{0.32\linewidth}
\includegraphics[width=\textwidth]{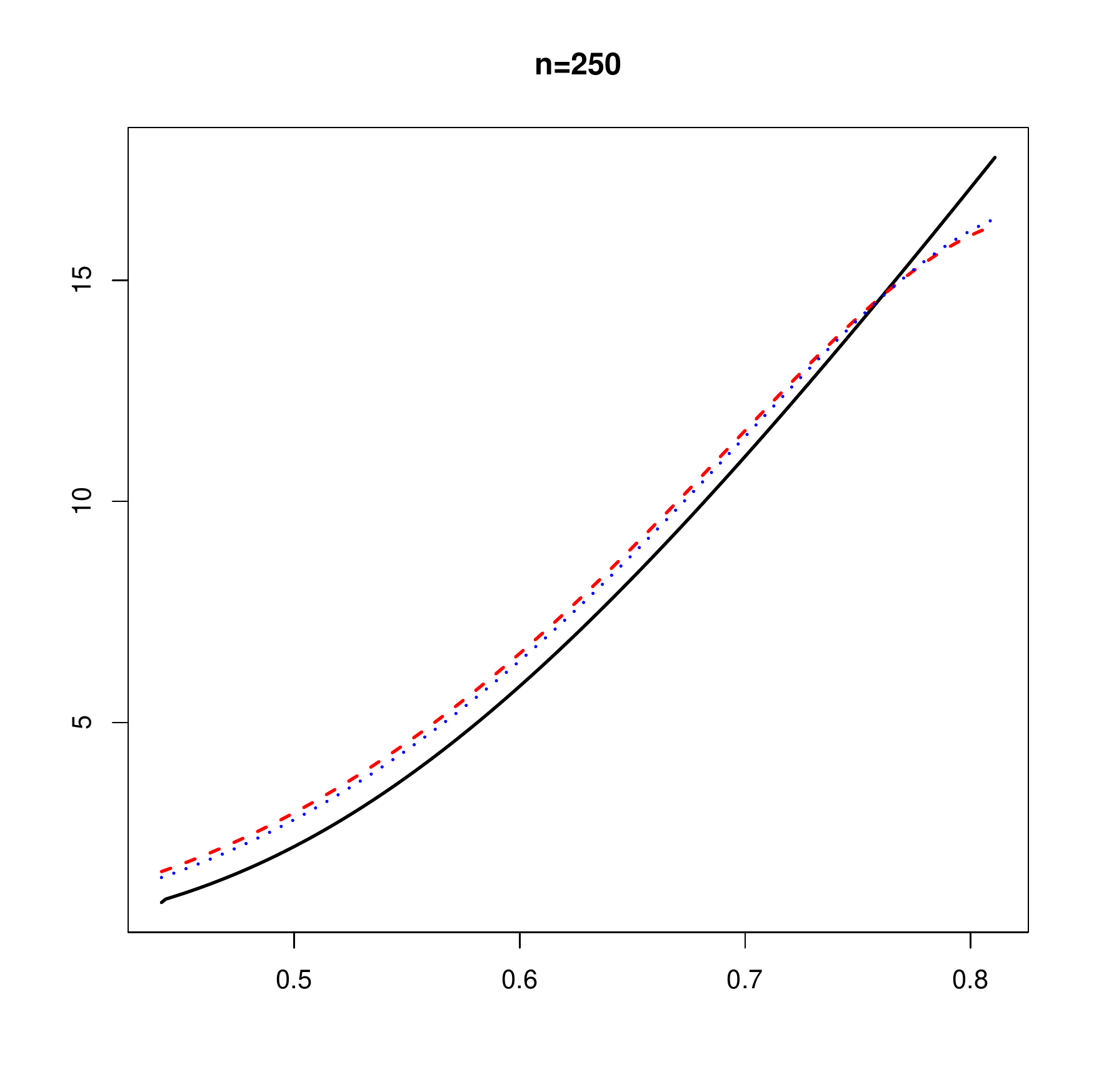}
\end{minipage} \hfill
\begin{minipage}[b]{0.32\linewidth}
\includegraphics[width=\textwidth]{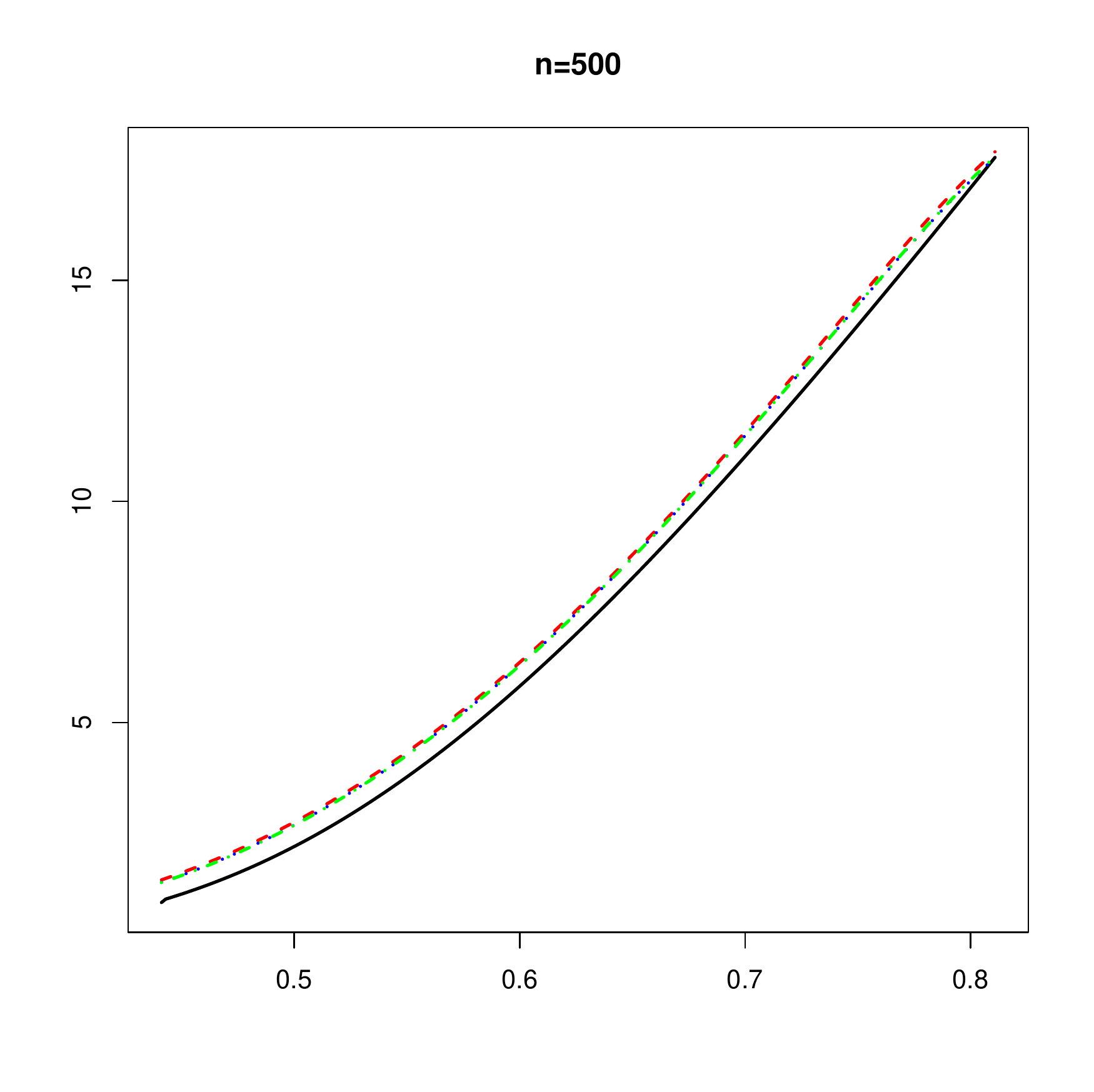}
\end{minipage}
\begin{minipage}[b]{0.32\linewidth}
\includegraphics[width=\textwidth]{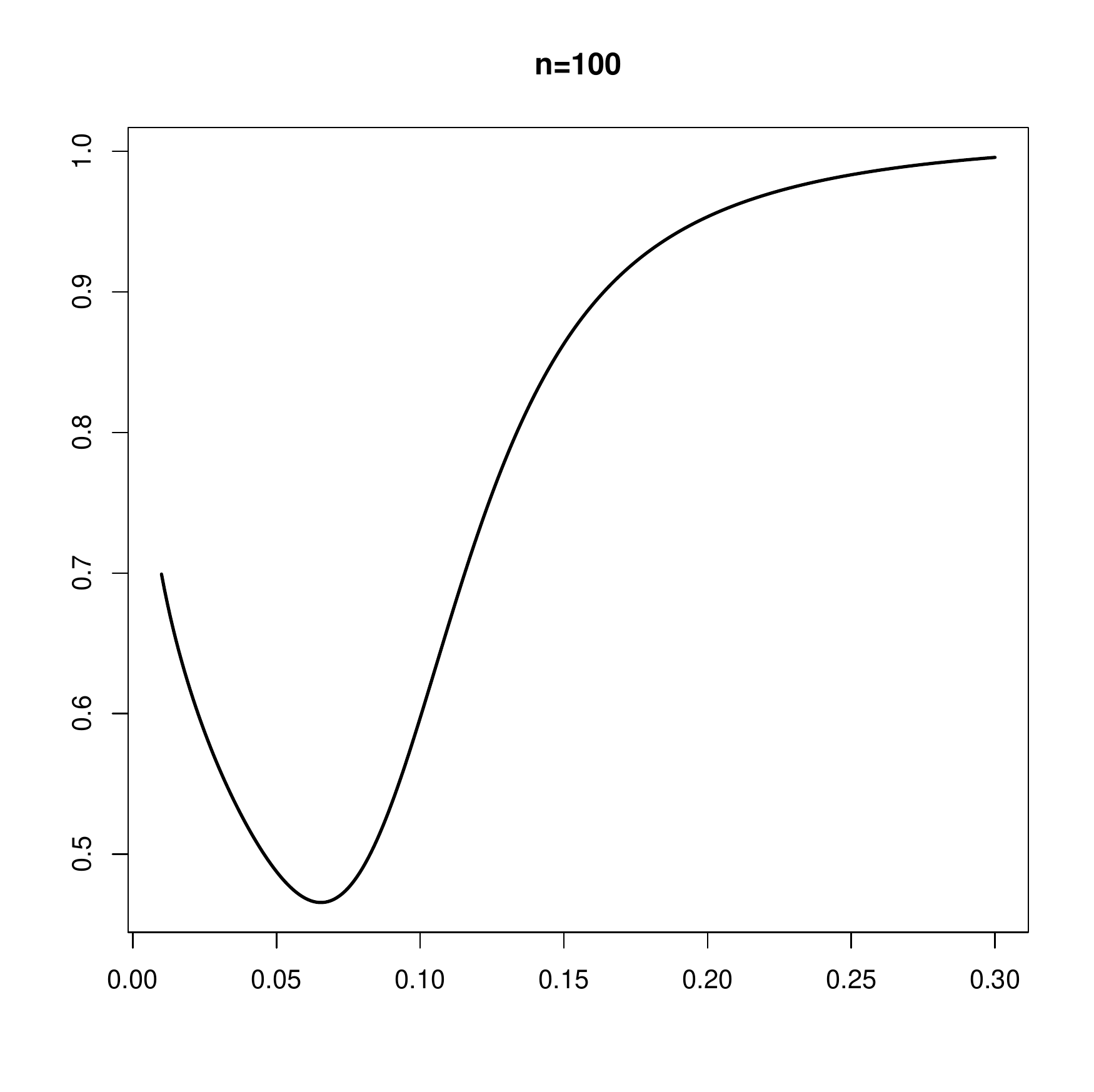}
\end{minipage} \hfill
\begin{minipage}[b]{0.32\linewidth}
\includegraphics[width=\textwidth]{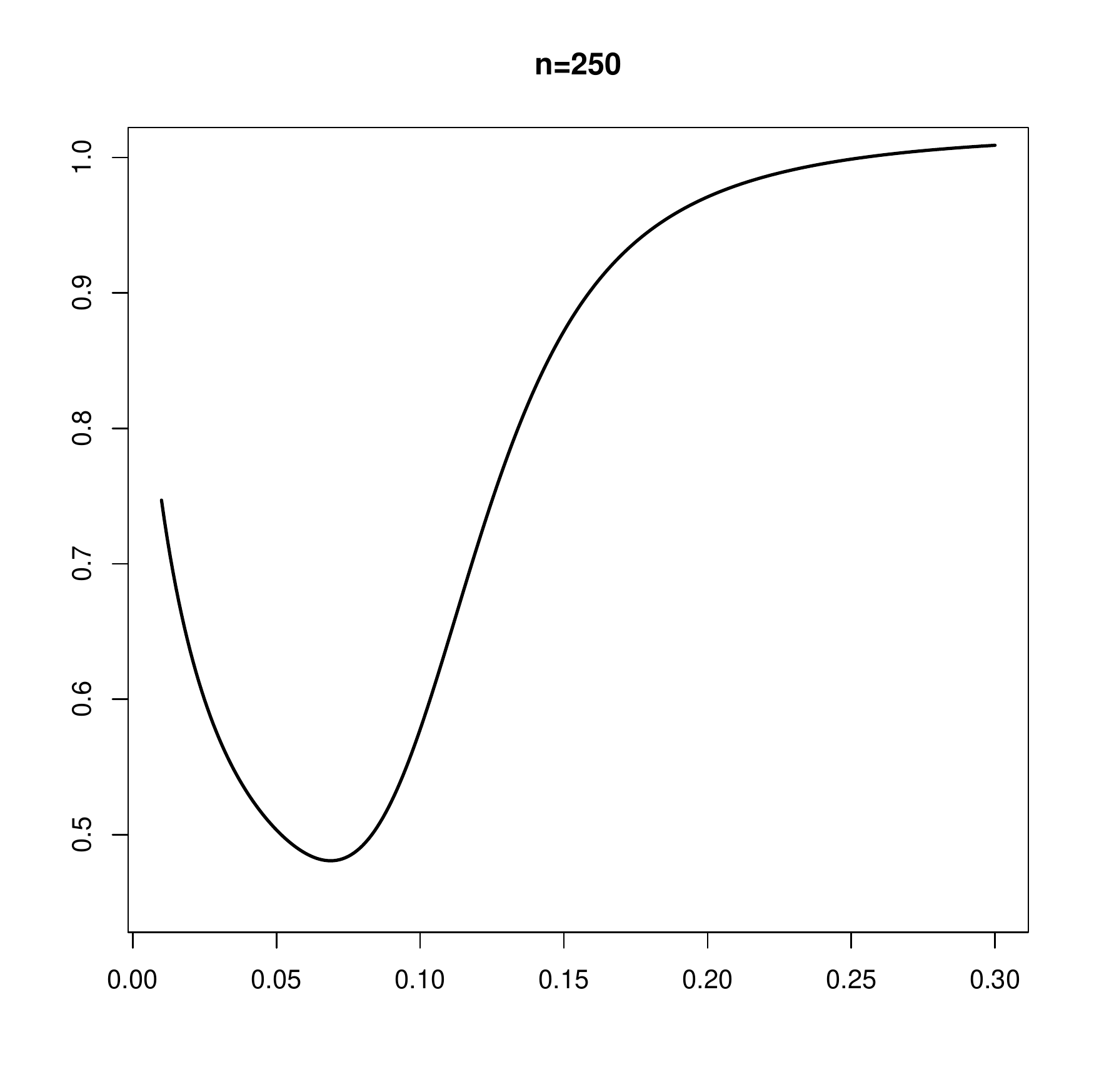}
\end{minipage} \hfill
\begin{minipage}[b]{0.32\linewidth}
\includegraphics[width=\textwidth]{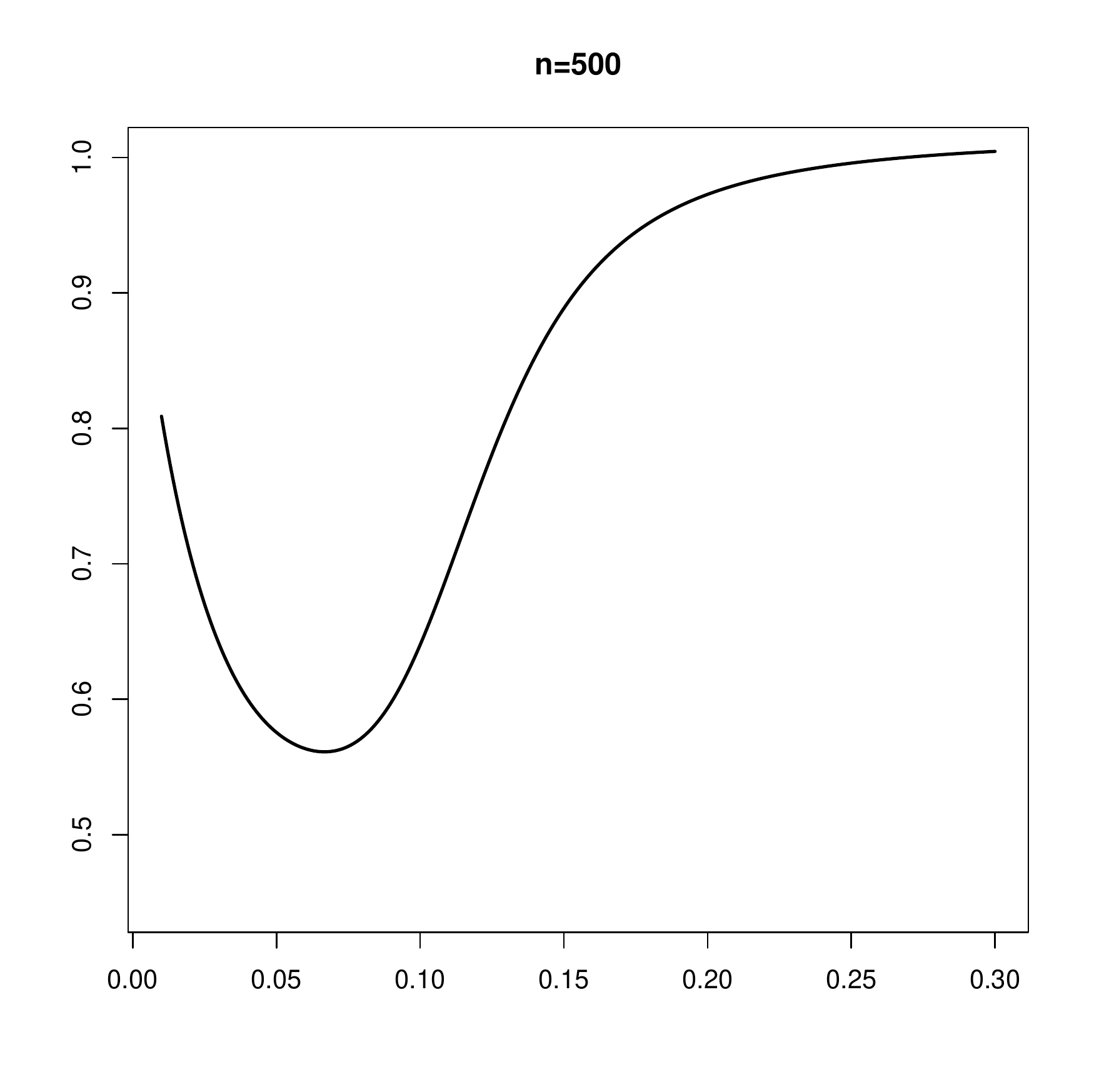}
\end{minipage}
\caption{(i) the target hazard (solid line) together with
its semiparametric (dashed line), nonparametric (dotted line) and naive (dashed-dotted line for n=500) estimators averaged along the
1000 Monte Carlo trials for Model 2 (top row); (ii) The ratio between the MISE's of the semiparametric and the nonparametric estimators along a grid of bandwidths (bottom row). }
\label{Teoric_Model22}
\end{figure}

\begin{figure}[ht]
\begin{minipage}[b]{0.32\linewidth}
\includegraphics[width=\textwidth]{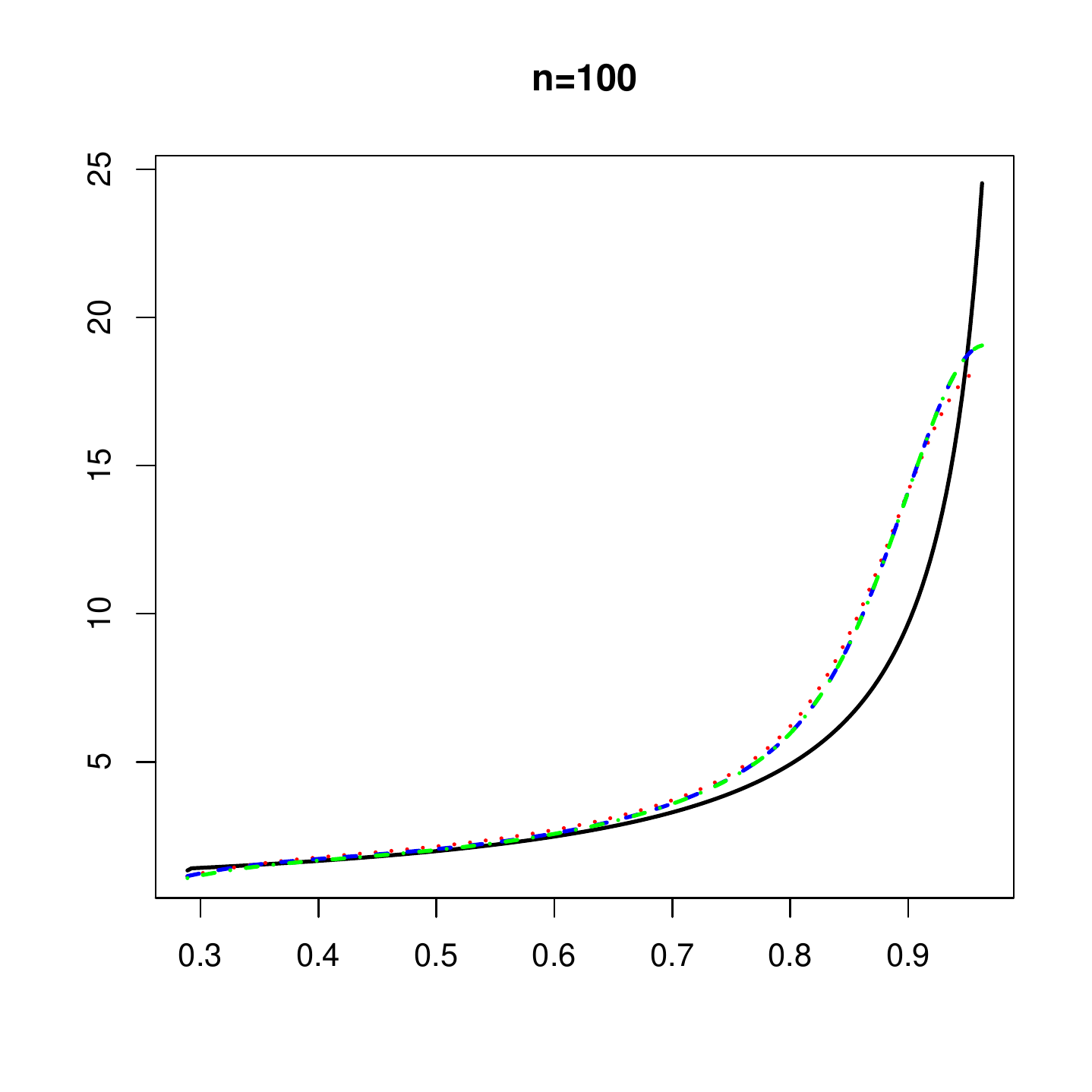}
\end{minipage} \hfill
\begin{minipage}[b]{0.32\linewidth}
\includegraphics[width=\textwidth]{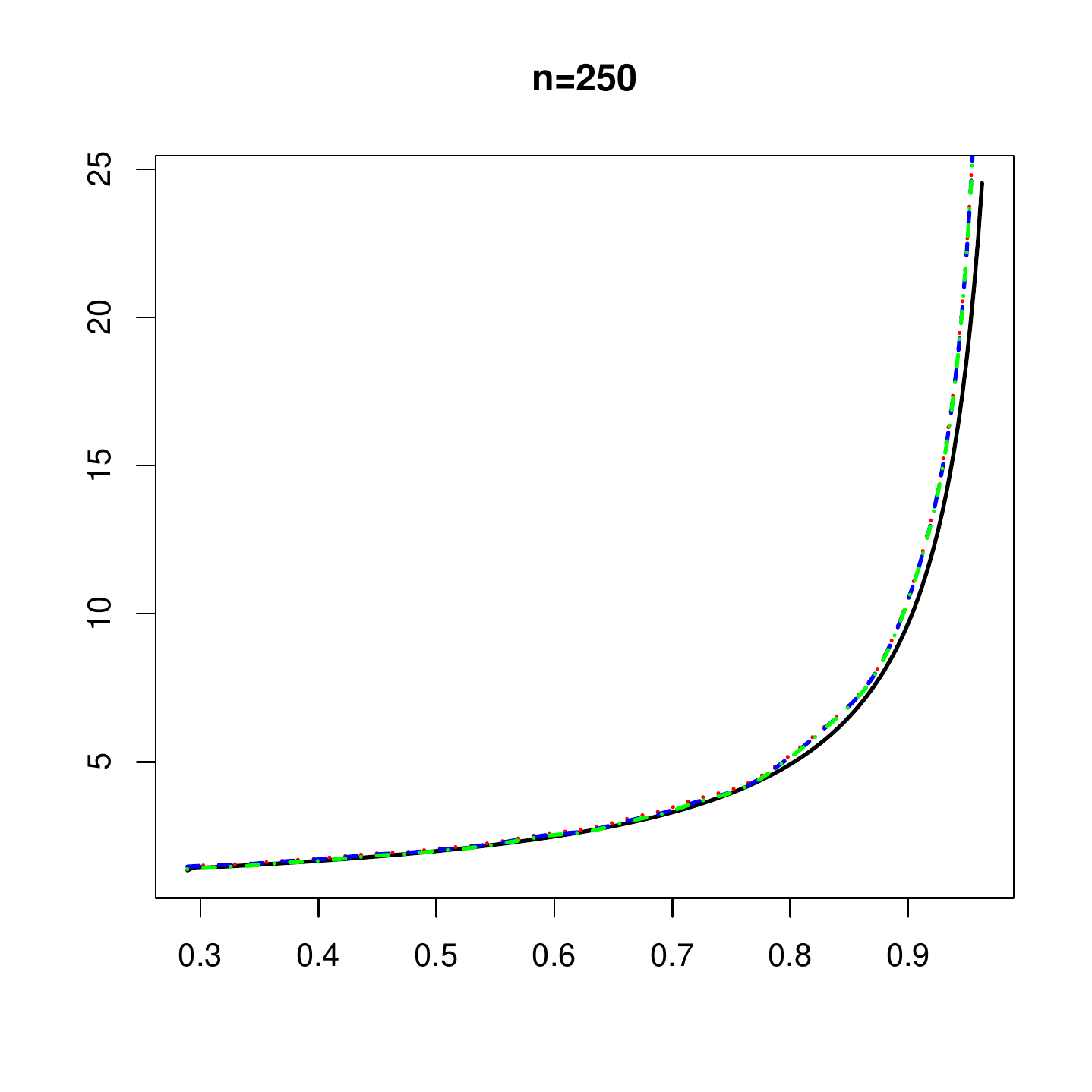}
\end{minipage} \hfill
\begin{minipage}[b]{0.32\linewidth}
\includegraphics[width=\textwidth]{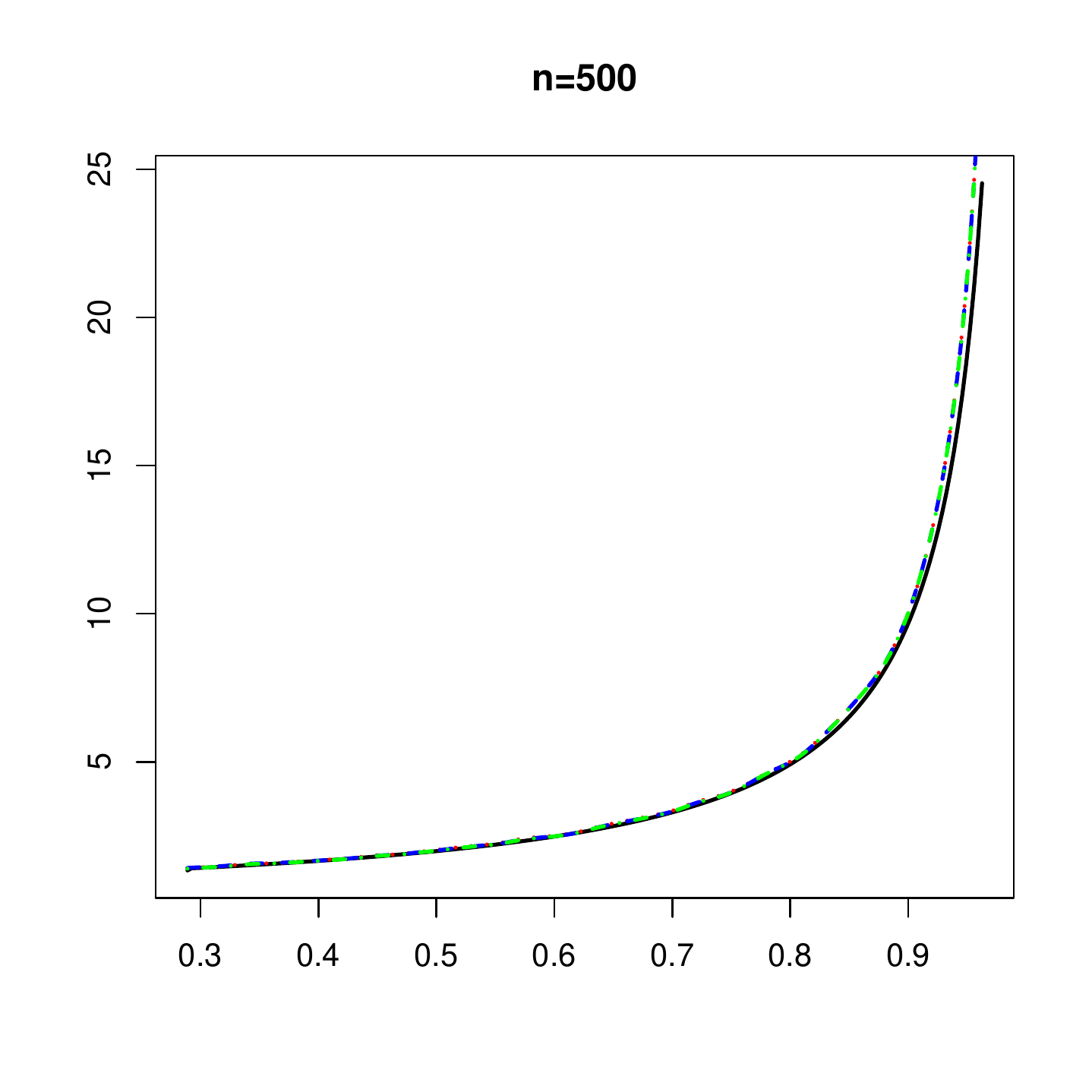}
\end{minipage}
\begin{minipage}[b]{0.32\linewidth}
\includegraphics[width=\textwidth]{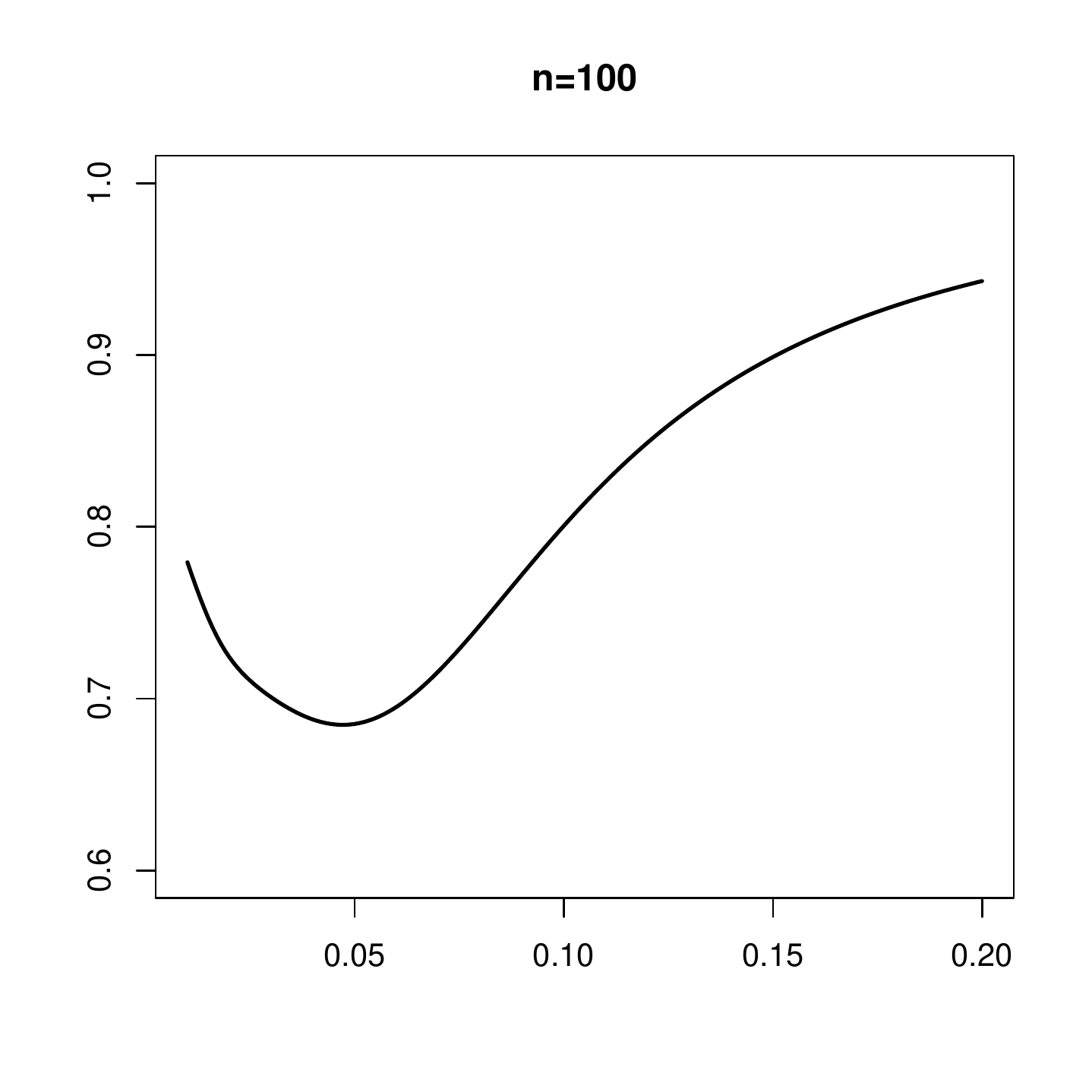}
\end{minipage} \hfill
\begin{minipage}[b]{0.32\linewidth}
\includegraphics[width=\textwidth]{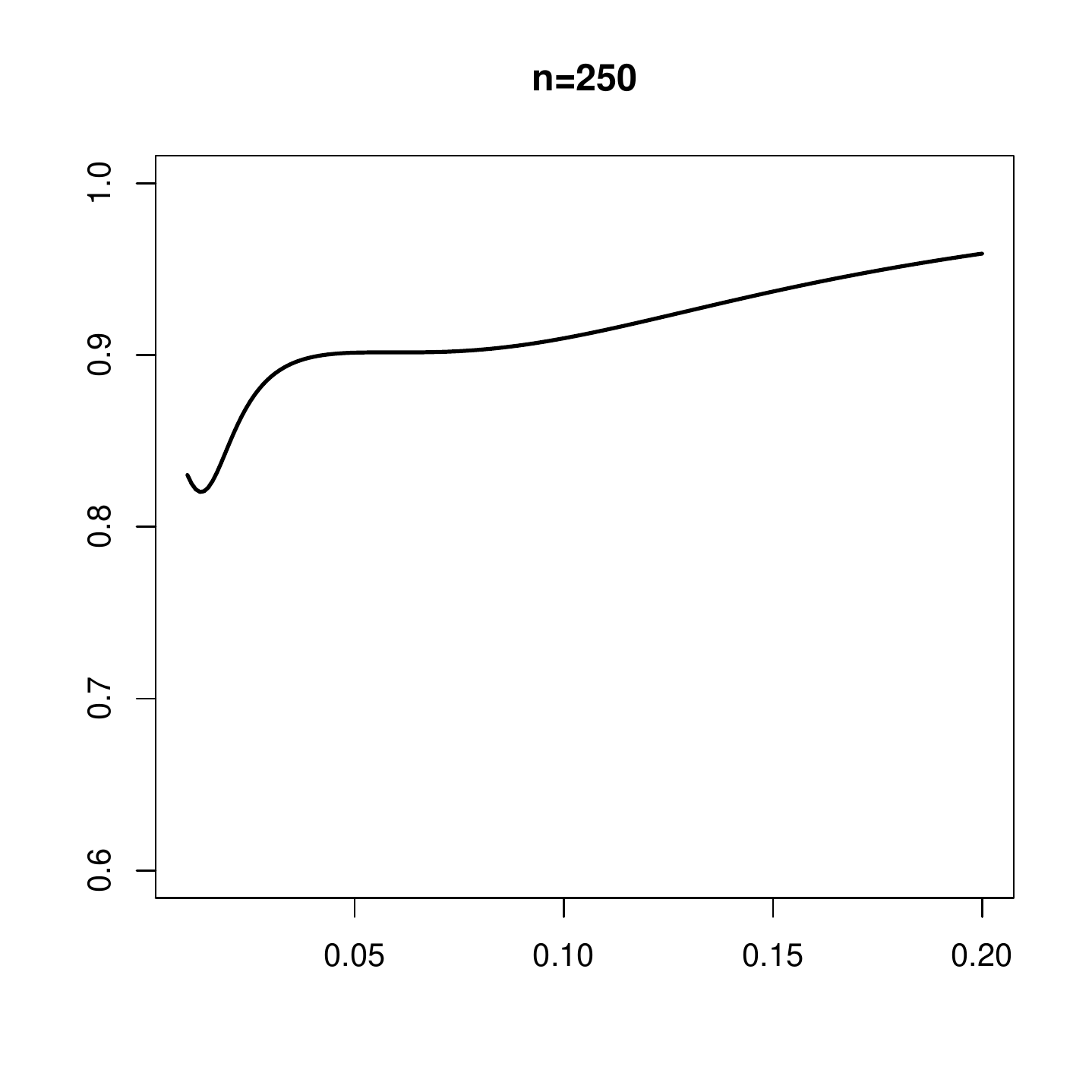}
\end{minipage} \hfill
\begin{minipage}[b]{0.32\linewidth}
\includegraphics[width=\textwidth]{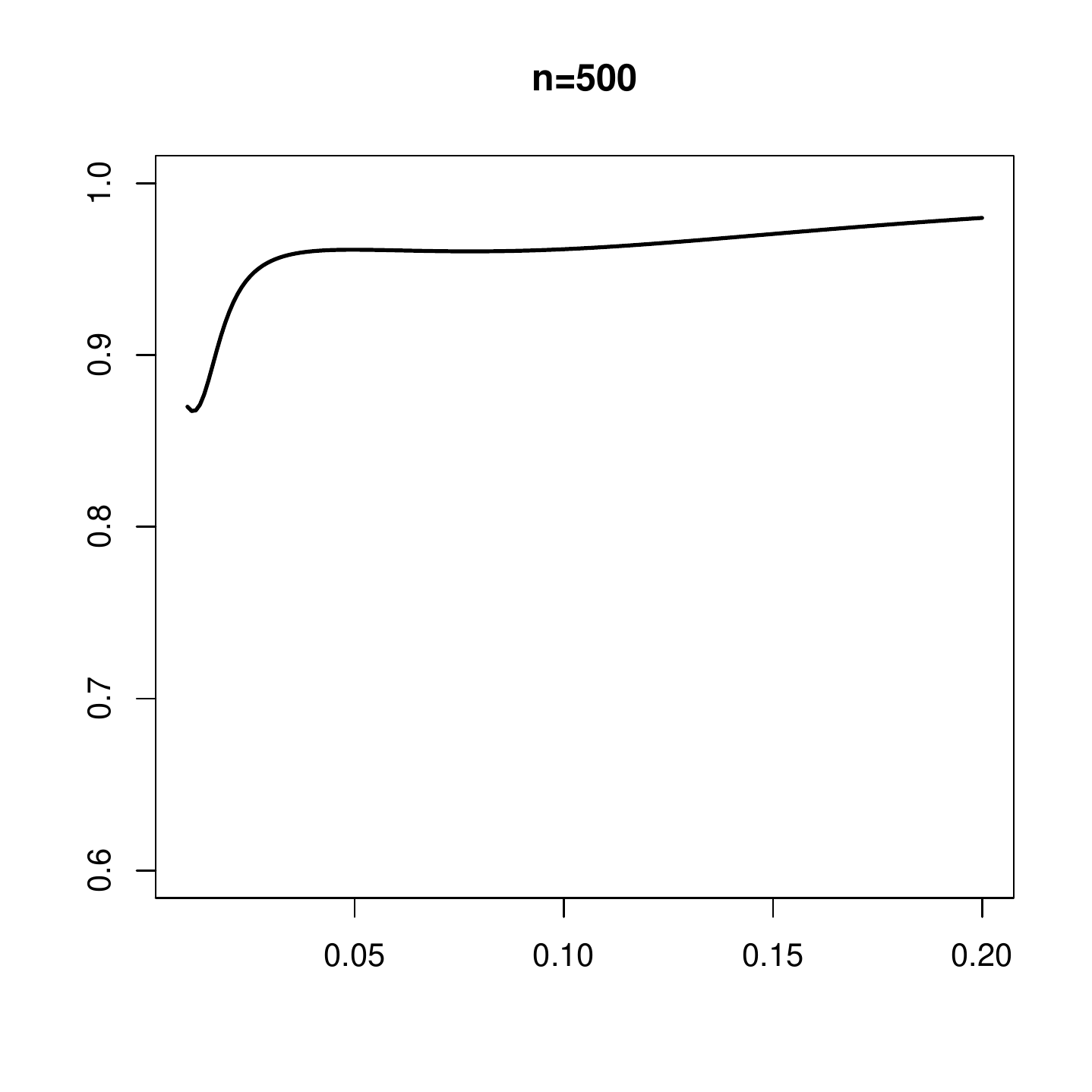}
\end{minipage}
\caption{(i) the target hazard (solid line) together with
its semiparametric (dashed line), nonparametric (dotted line) and naive (dashed-dotted line for n=500) estimators averaged along the
1000 Monte Carlo trials for Model 3.1 (top row); (ii) The ratio between the MISE's of the semiparametric and the nonparametric estimators along a grid of bandwidths (bottom row). }
\label{Teoric_Model231}
\end{figure}

\begin{figure}[ht]
\begin{minipage}[b]{0.32\linewidth}
\includegraphics[width=\textwidth]{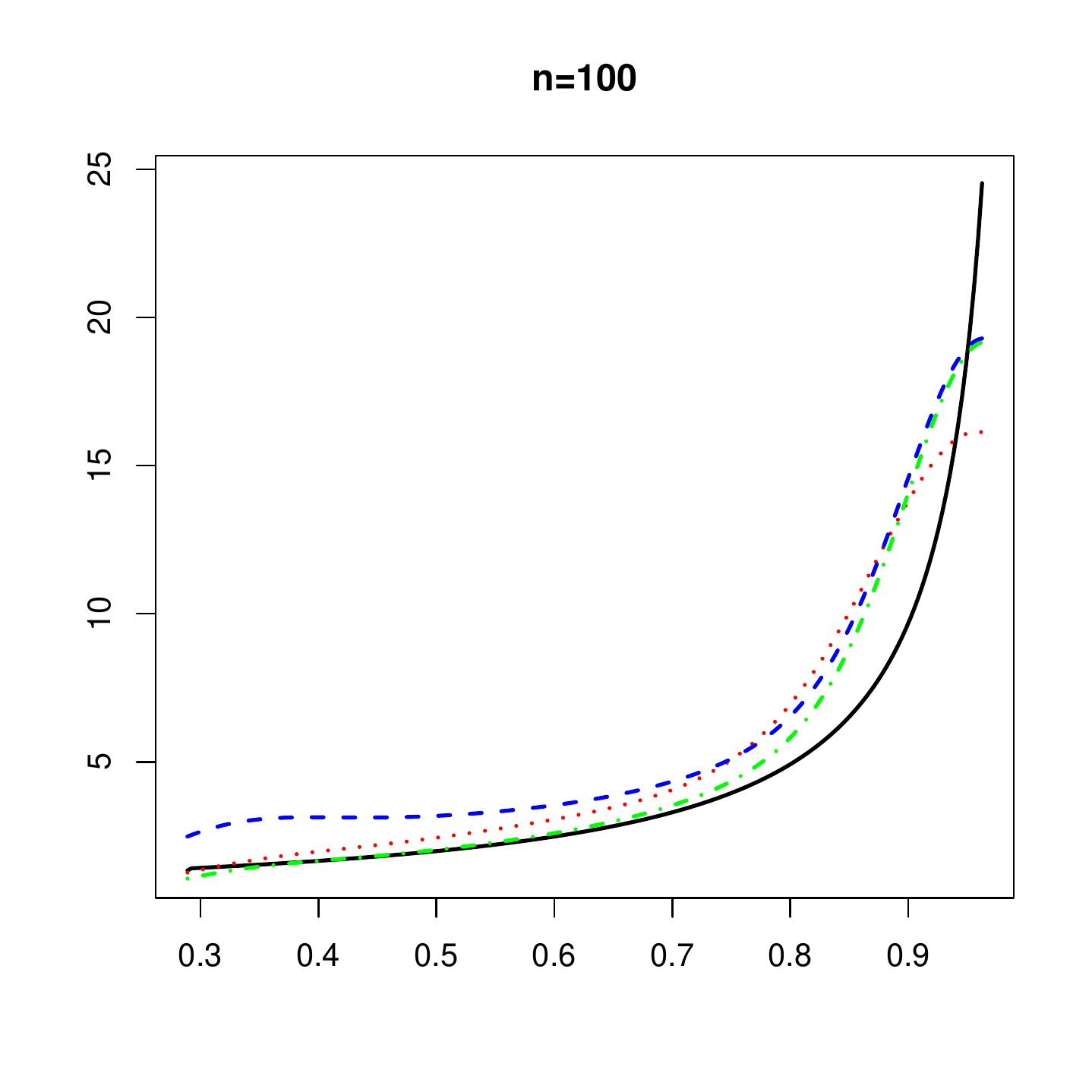}
\end{minipage} \hfill
\begin{minipage}[b]{0.32\linewidth}
\includegraphics[width=\textwidth]{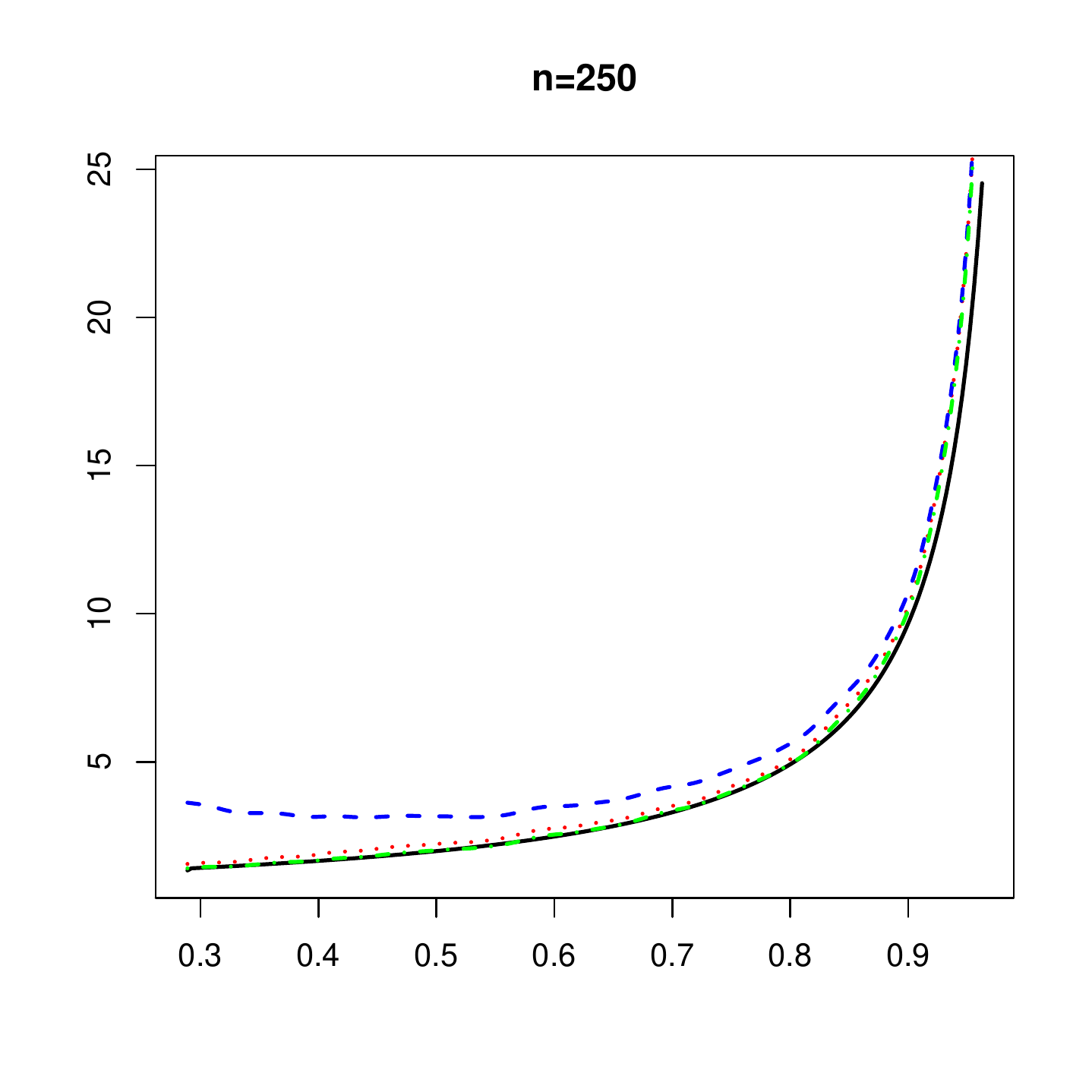}
\end{minipage} \hfill
\begin{minipage}[b]{0.32\linewidth}
\includegraphics[width=\textwidth]{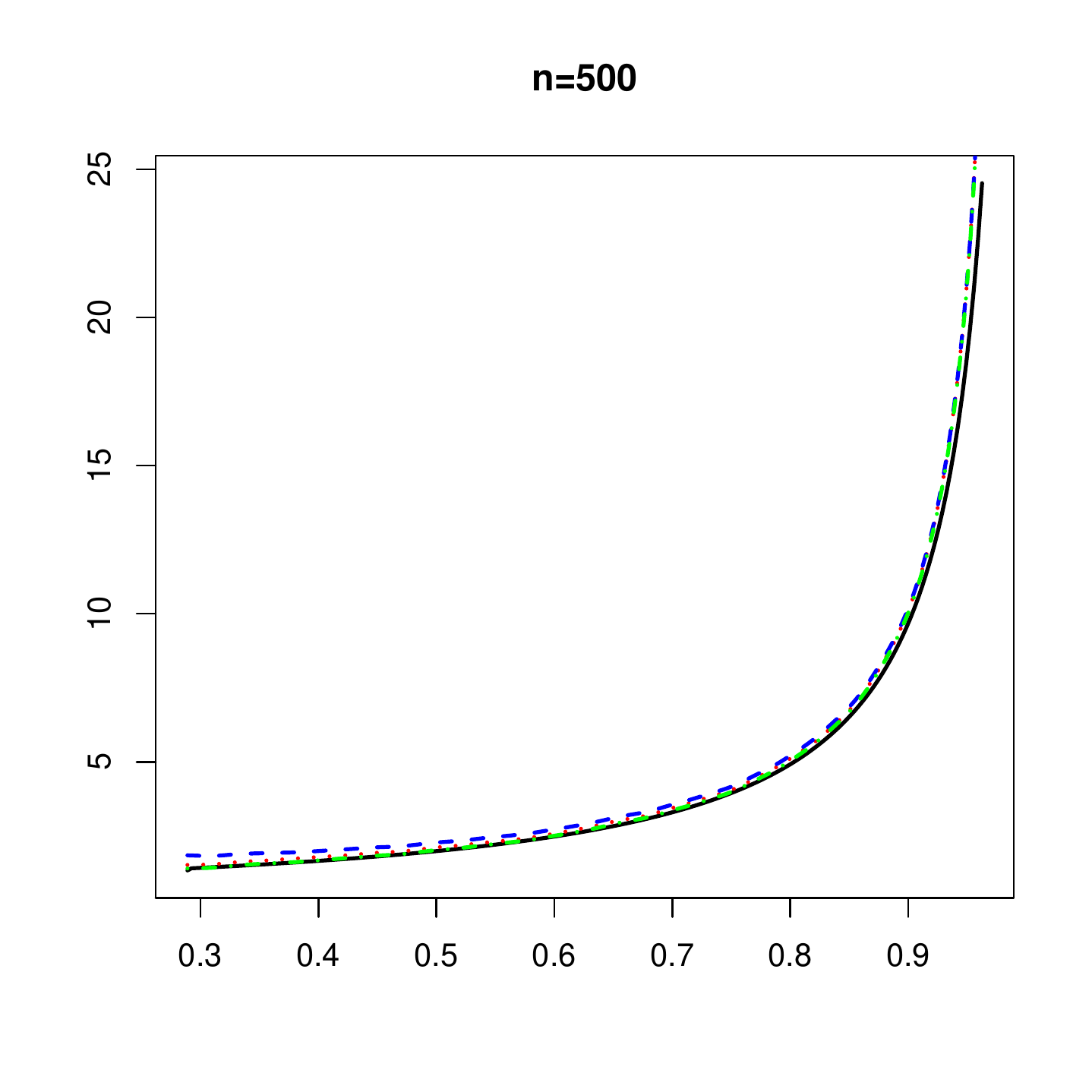}
\end{minipage}
\begin{minipage}[b]{0.32\linewidth}
\includegraphics[width=\textwidth]{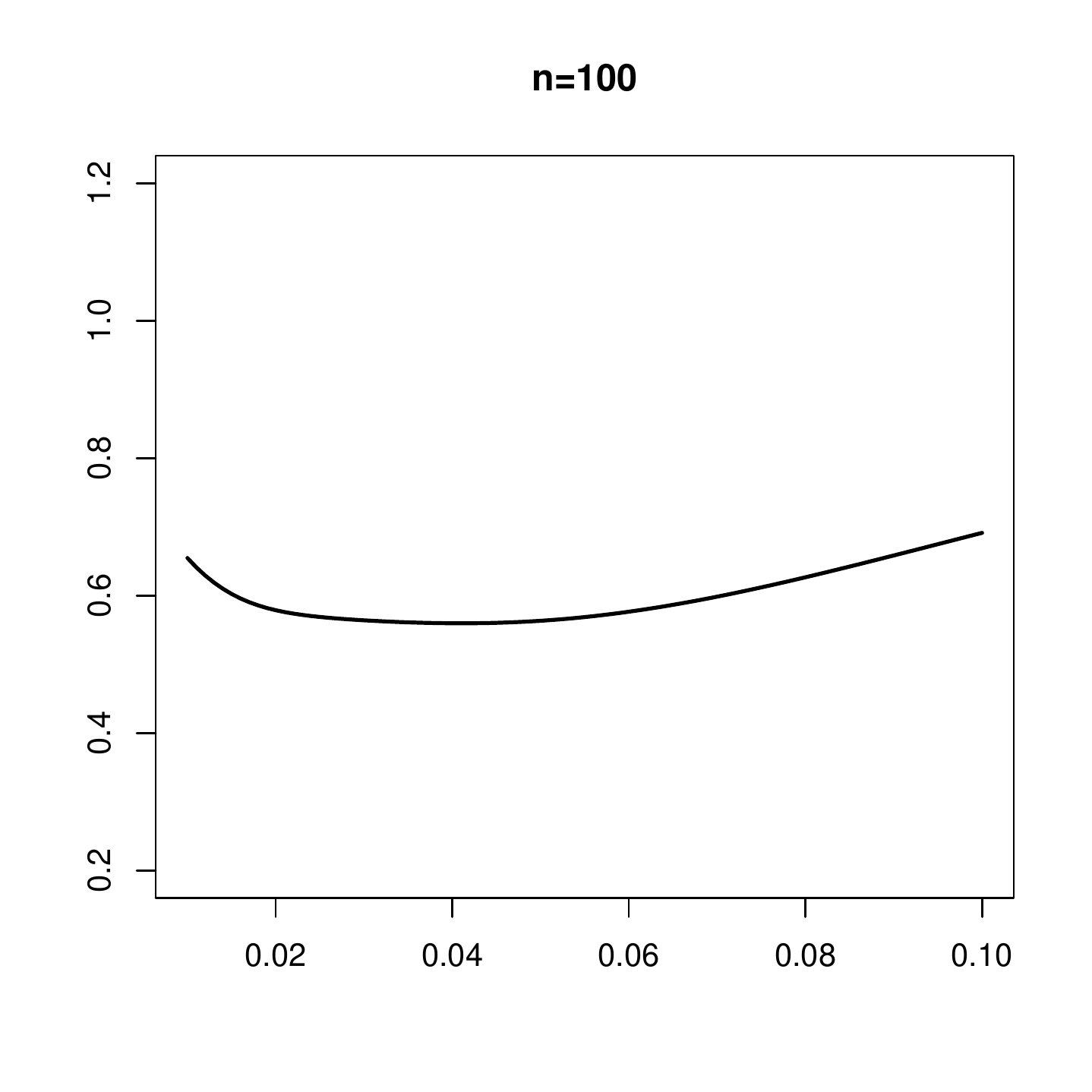}
\end{minipage} \hfill
\begin{minipage}[b]{0.32\linewidth}
\includegraphics[width=\textwidth]{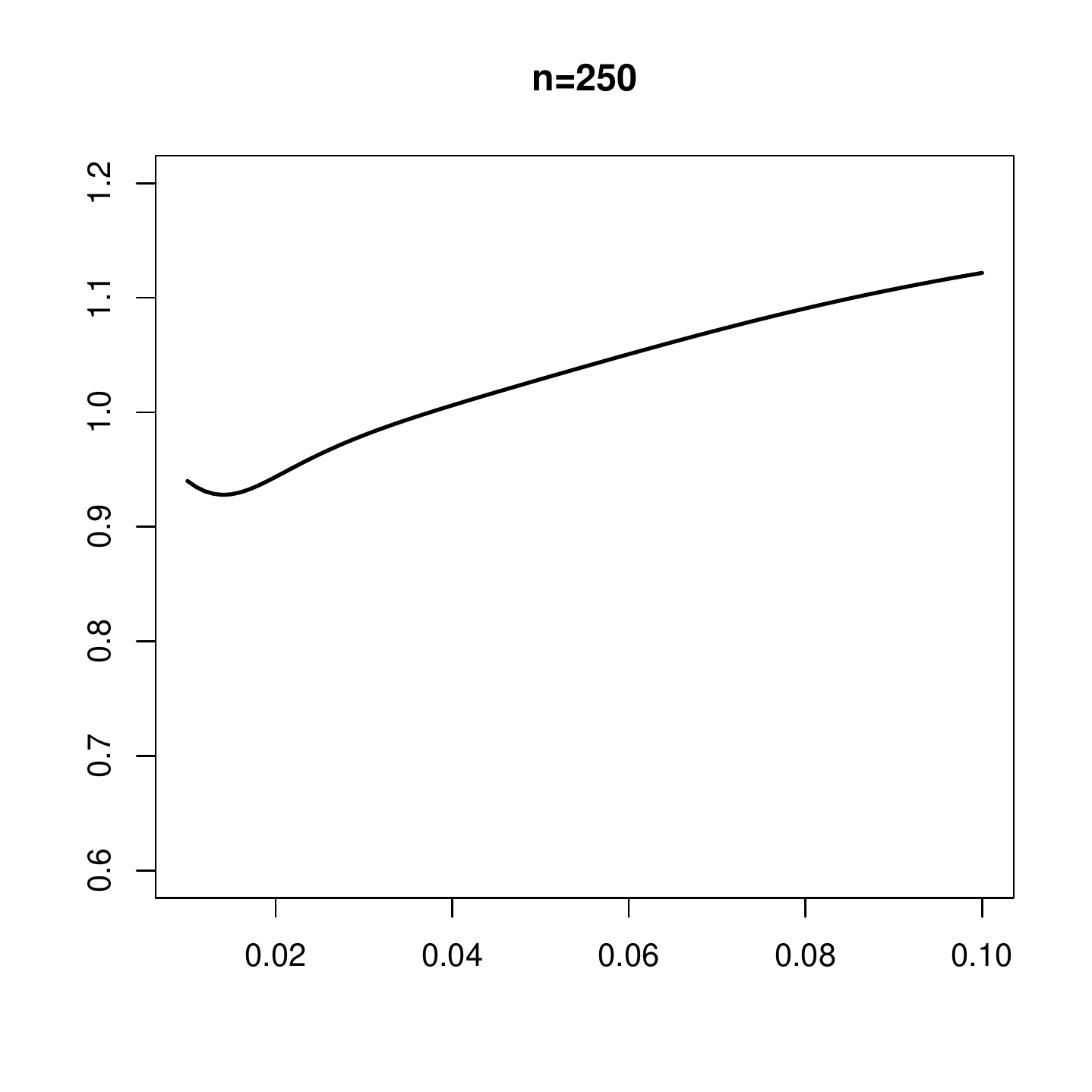}
\end{minipage} \hfill
\begin{minipage}[b]{0.32\linewidth}
\includegraphics[width=\textwidth]{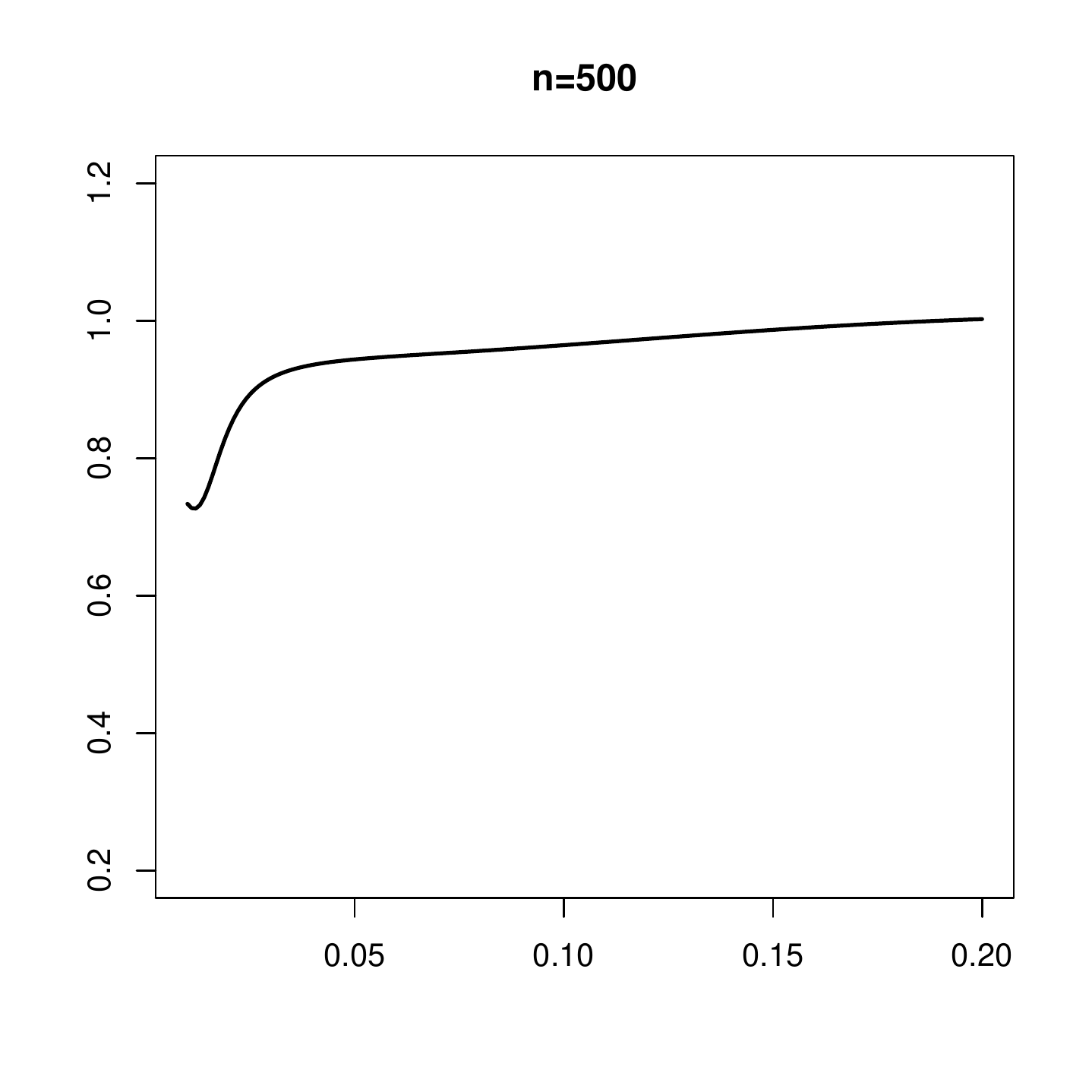}
\end{minipage}
\caption{(i) the target hazard (solid line) together with
its semiparametric (dashed line), nonparametric (dotted line) and naive (dashed-dotted line for n=500) estimators averaged along the
1000 Monte Carlo trials for Model 3.2 (top row); (ii) The ratio between the MISE's of the semiparametric and the nonparametric estimators along a grid of bandwidths (bottom row). }
\label{Teoric_Model232}
\end{figure}

\begin{figure}[ht]
\begin{minipage}[b]{0.32\linewidth}
\includegraphics[width=\textwidth]{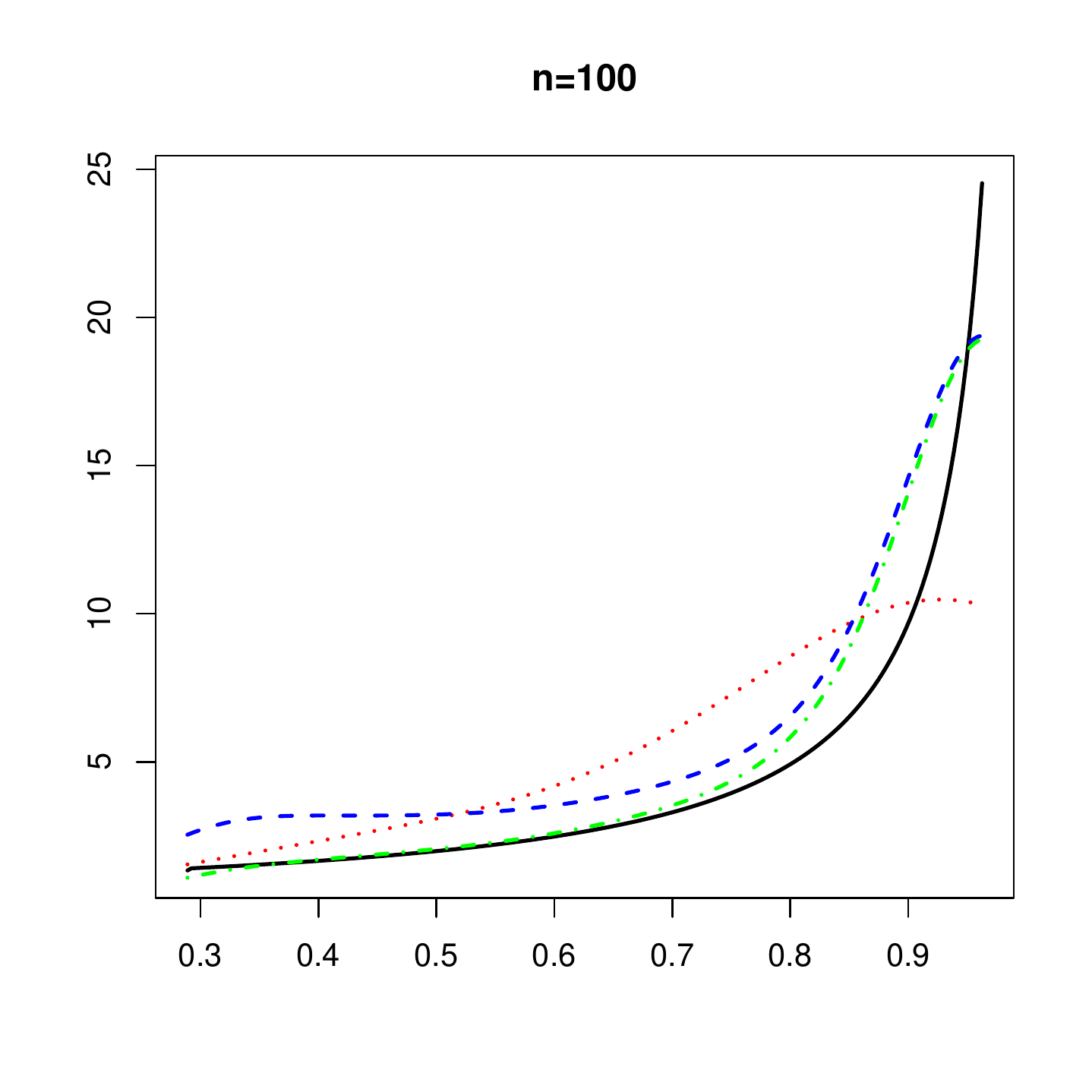}
\end{minipage} \hfill
\begin{minipage}[b]{0.32\linewidth}
\includegraphics[width=\textwidth]{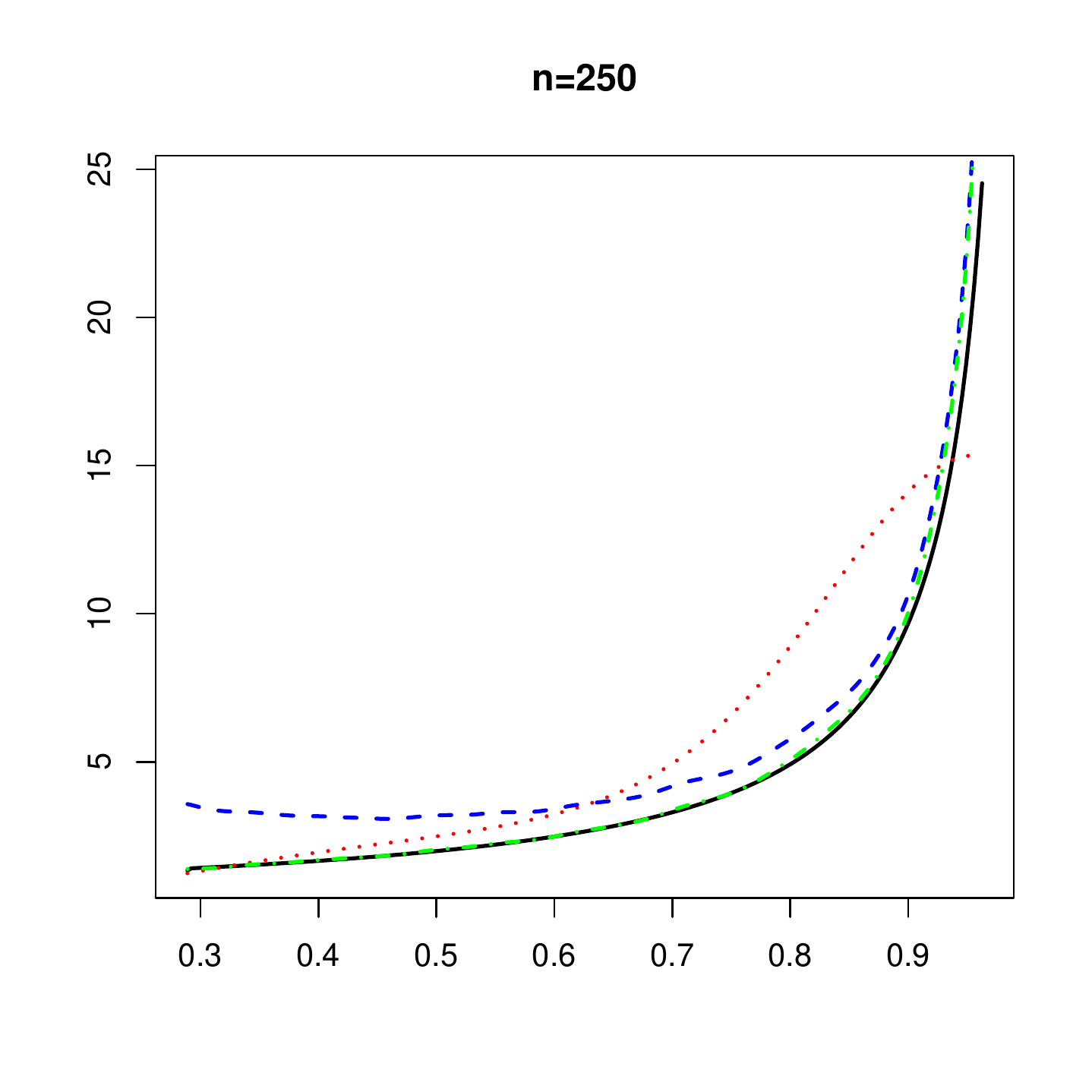}
\end{minipage} \hfill
\begin{minipage}[b]{0.32\linewidth}
\includegraphics[width=\textwidth]{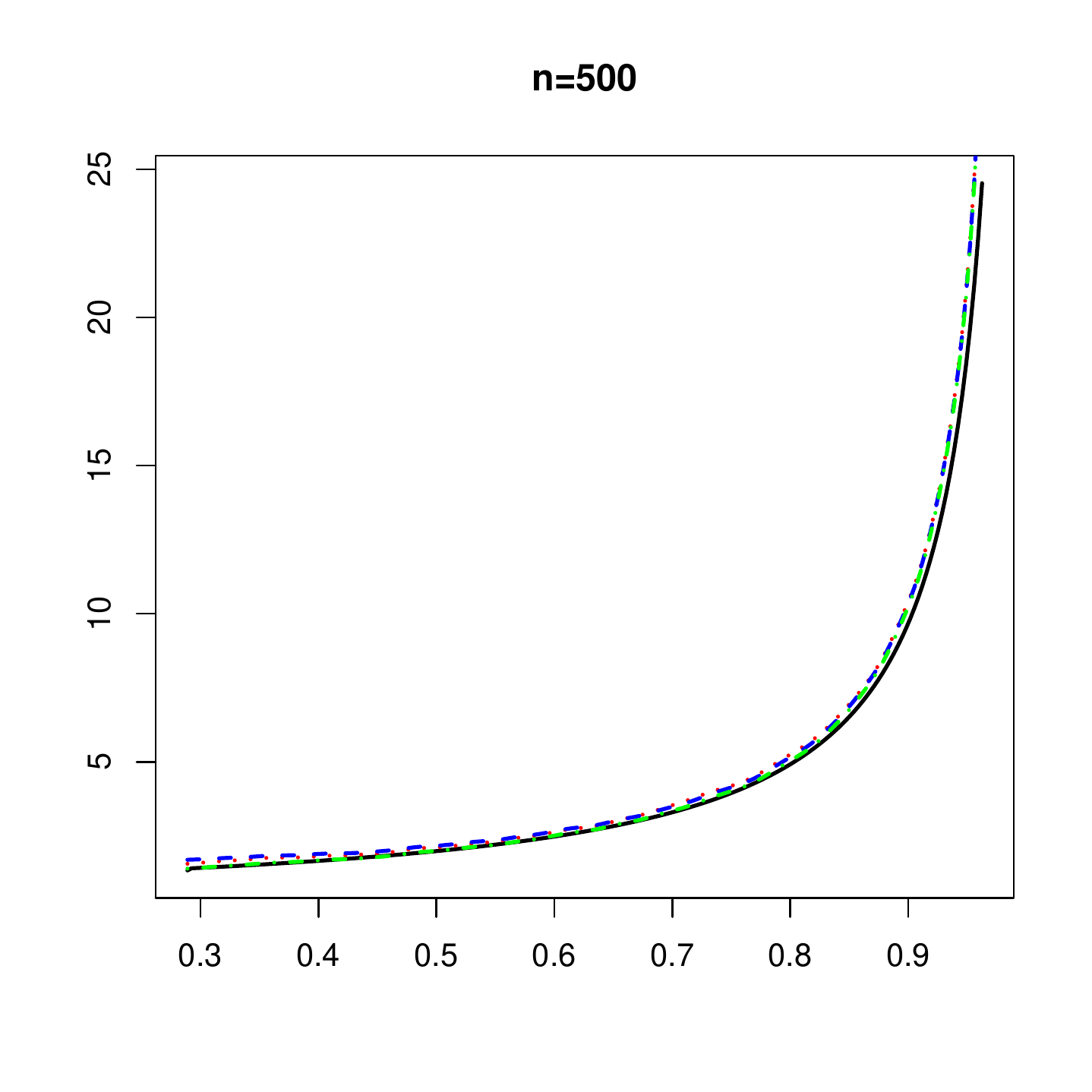}
\end{minipage}
\begin{minipage}[b]{0.32\linewidth}
\includegraphics[width=\textwidth]{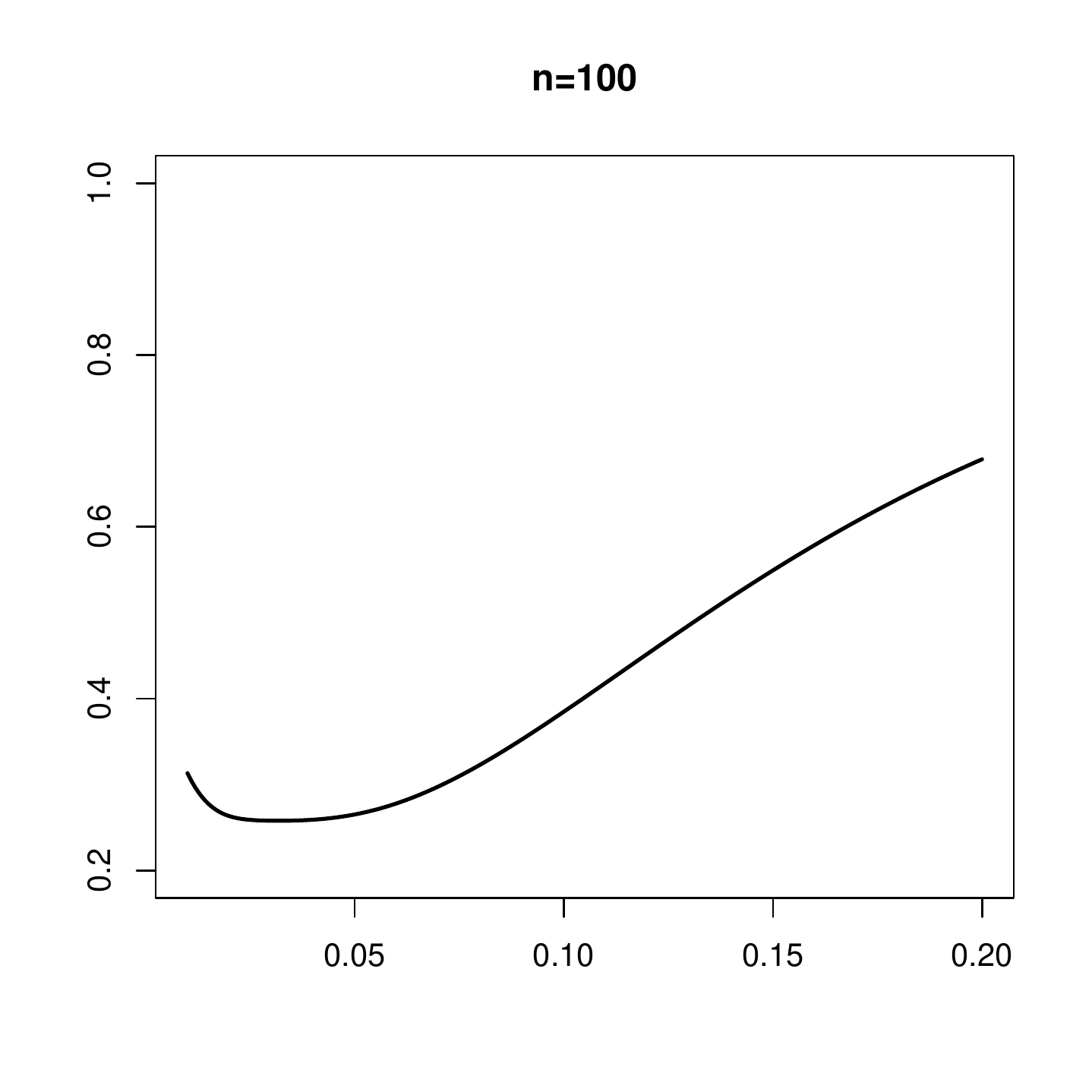}
\end{minipage} \hfill
\begin{minipage}[b]{0.32\linewidth}
\includegraphics[width=\textwidth]{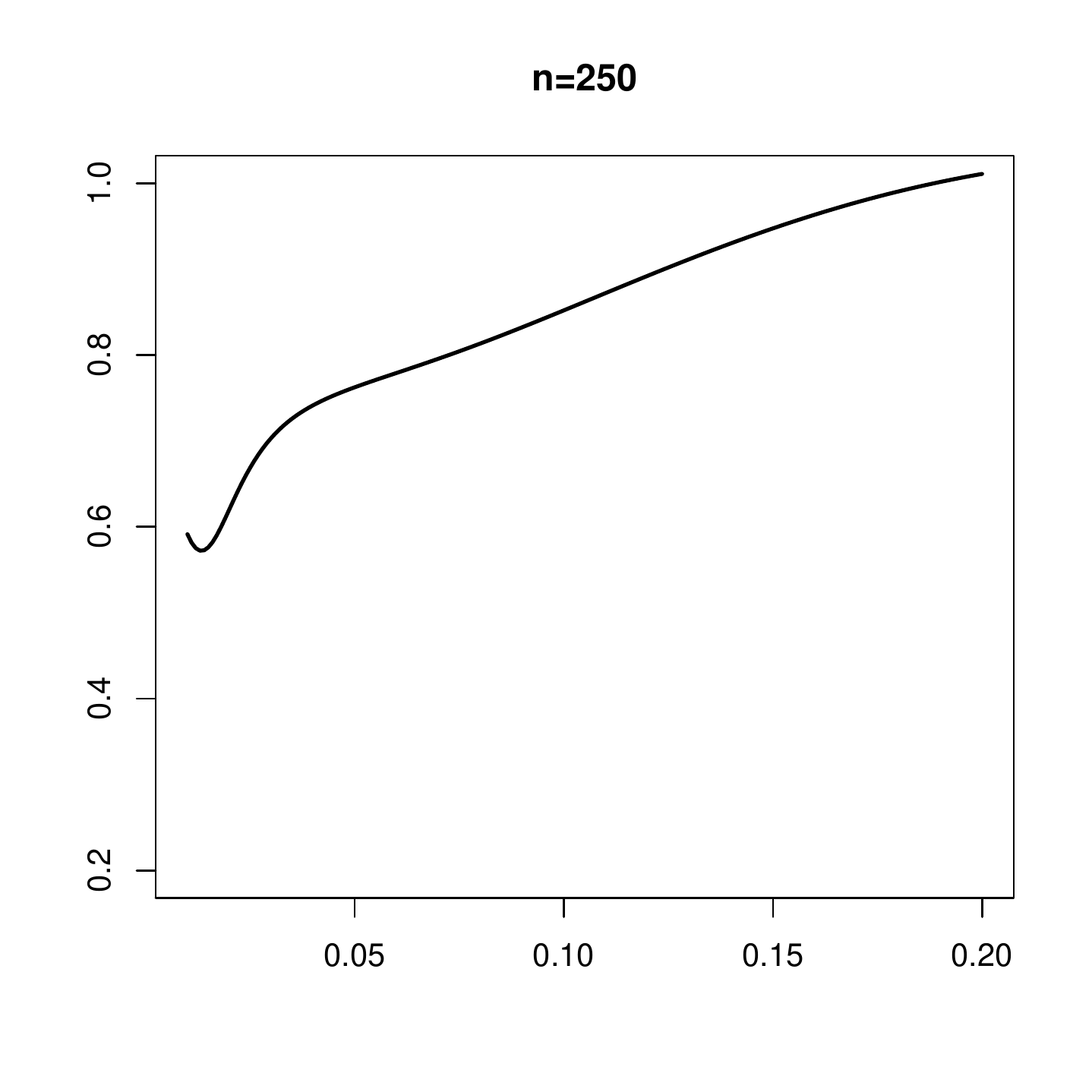}
\end{minipage} \hfill
\begin{minipage}[b]{0.32\linewidth}
\includegraphics[width=\textwidth]{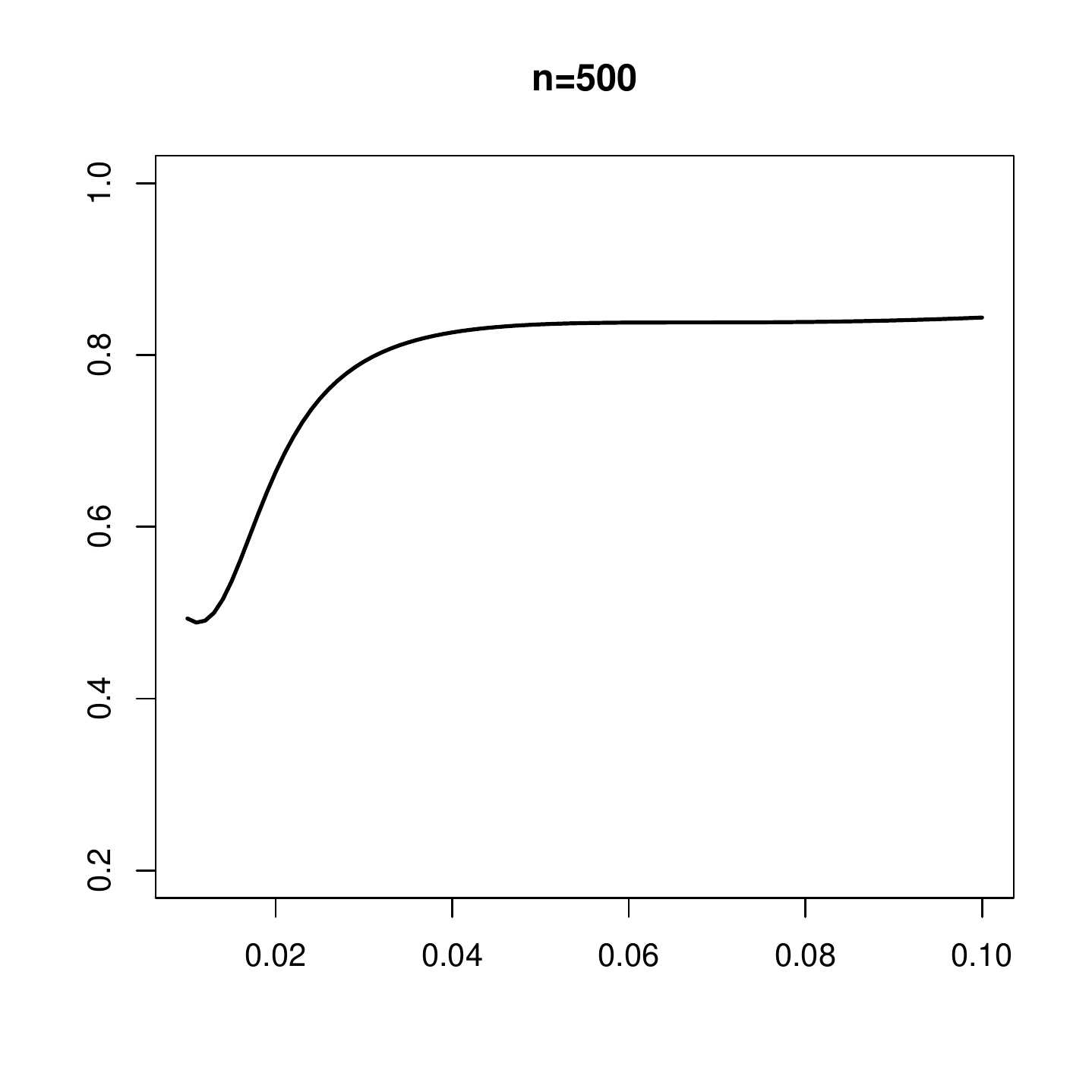}
\end{minipage}
\caption{(i) the target hazard (solid line) together with
its semiparametric (dashed line), nonparametric (dotted line) and naive (dashed-dotted line for n=500) estimators averaged along the
1000 Monte Carlo trials for Model 3.3 (top row); (ii) The ratio between the MISE's of the semiparametric and the nonparametric estimators along a grid of bandwidths (bottom row). }
\label{Teoric_Model233}
\end{figure}
Simulations above are informative about the relative performance of the two proposed estimators when the parametric information on the truncation distribution is correctly specified. However, in practice, some level of misspecification in the parametric model may occur. To investigate the sensitivity of the semiparametric estimator to the misspecification degree, we have repeated the simulation of Model 1 but changing the $U(0,1)$ distribution of $U^*$ for a $Beta(1,a)$ distribution, with $a\neq 1$, so the parametric information $Beta(\theta_1,1)$ on $U^*$ is misspecified.  Note that the misspecification degree increases as $a$ departs from $1$. Results on the bandwidth, the MISE, and the local MSE of
both the semiparametric and nonparametric hazard estimators are reported in
Tables \ref{Table3} and \ref{Table4} for the case $n=500$ (results based on 1,000 trials). From Table \ref{Table3},
it is seen that the semiparametric estimator may be still equivalent or even preferred
to the nonparametric estimator in all cases of misspecification. Table \ref{Table4}
indicates that, when the parametric information is misspecified, the variance of
the semiparamtric estimator remains smaller than that of the nonparametric estimator for almost all the cases (exceptions for $a=5$ are found in the first two quartiles).

\begin{table}[ht]
\centering
%\scriptsize
\caption {Optimal bandwidths ($h_{opt}$) and minimum MISEs of the
hazard estimators: nonparametric estimator ($NP$) and semiparametric
estimator ($SP$). Averages along 1000 trials of a sample size $n=500$. Similar as Model 1, but $U^*$ is simulated as a $Beta(1, a)$ random variable.}
\label{Table3}
\begin{tabular}{ c c c c c}
 a& \multicolumn{2}{c}{$h_{opt}$}  & \multicolumn{2}{c}{$MISE(h_{opt})$}\\
\cline{2-3}
\cline{4-5}
&NP&SP&NP&SP\\
\hline
$1/5$&0.027&0.037&1.632&1.492\\
$1/2$&0.028&0.030&1.586&1.443\\
$1$&0.031&0.031&1.309&1.116\\
$3/2$&0.025&0.025&2.794&1.850\\
5&0.024&0.026&6.024&3.386\\
\hline
\end{tabular}
\end{table}

\begin{table}[ht]
\centering
%\scriptsize
\caption {Bias ($ \times 10^3$) and variances ($ \times 10$)  of the
nonparametric estimator ($NP$) and semiparametric
estimator ($SP$) at the quartiles of $F$, for sample sizes n=500, along 1000 Monte Carlo trials. Misspecified parametric model.}
\label{Table4}
\begin{tabular}{ c c r r r r}
\hline
 && \multicolumn{2}{c}{$NP$}  & \multicolumn{2}{c}{$SP$}\\
\cline{3-4}
\cline{5-6}
a&x& Bias&Var  & Bias&Var\\
%&&&EP&SP&EP&SP\\
\hline
&$q_{.25}$&$0.44$&$0.51$&$11.80$&$0.44$\\
$1/5$&$q_{.50}$&$7.02$&$1.26$&$6.45$&$0.99$\\
&$q_{.75}$&$1825.00$&$11.64$&$4687.00$&$6.86$\\
\cline{2-6}
&$q_{.25}$&$0.50$&$0.44$&$0.92$&$0.35$\\
$1/2$&$q_{.50}$&$7.18$&$1.08$&$1.83$&$0.73$\\
&$q_{.75}$&$36.84$&$4.18$&$2.89$&$2.93$\\
\cline{2-6}
&$q_{.25}$&$0.80$&$0.39$&$0.99$&$0.30$\\
1&$q_{.50}$&$7.70$&$1.06$&$8.80$&$0.87$\\
&$q_{.75}$&$38.28$&$4.47$&$37.98$&$3.79$\\
\cline{2-6}
&$q_{.25}$&$0.58$&$40.42$&$24.23$&$0.34$\\
$3/2$&$q_{.50}$&$0.02$&$1.13$&$55.54$&$1.03$\\
&$q_{.75}$&$39.09$&$4.88$&$171.48$&$4.59$\\
\cline{2-6}
&$q_{.25}$&$1.57$&$48.41$&$1.98$&$51.30$\\
5&$q_{.50}$&$14.75$&$1.36$&$293.40$&$1.53$\\
&$q_{.75}$&$3176.00$&$8.90$&$2037.00$&$8.58$\\
\hline
\end{tabular}
\end{table}

\section{Real data illustration} \label{Section5}

\subsection{Acute Coronary Syndrome data}
\indent For illustration purposes, in this section we consider the aforementioned data on the
age at diagnosis of ACS. In Portugal, with a population of 10.3 million inhabitants, there are 38 public hospitals with resources for structured care of patients with ACS, out of which 16 have catheterisation laboratory facilities. Public hospitals provide  treatment for the majority of the acute coronary events and the number of patients submitted to primary Percutaneous coronary intervention (PCI) increased by 37.0\% from 2009 to 2013, although at the regional level access to this procedure varied. The EPIHeart cohort is a prospective study  assembled between August 2013 and December 2014  to the Cardiology Department of two tertiary hospitals in two regions in Northern Portugal (Hospital de S\~ao Jo\~ao, Porto, covering the metropolitan area of Porto in the coast;  and Hospital de S\~ao Pedro, Vila Real, covering the interior, northeastern region). The inclusion criteria to the cohort were admission with a diagnosis of ACS type I, aged 18 years or older, living in the catchment area of the referred hospitals (Porto, Vila Real, Bragan\c{c}a or Viseu) with confirmed diagnosis of type 1 (primary spontaneous) ACS. Data was collected through structured interviews within the first 48 hours after admission. Of the 1297 patients initially considered, 939 were included in the cohort due the inclusion criteria. The age at diagnosis (ranging
from $30$ to $94$ years old) was doubly truncated by ($U^*,V^*$),
where $V^*$ stands for the elapsed time (in years) between birth and
end of the study (December 2014), and $U^*=V^*-1.42$.

For this dataset it happens $\sum_{j=1}^n I(U_j\leq X_i\leq V_j)=1$ for the three largest values of $X_i$, corresponding to ages $X_i=90.70$, $91.99$ and $93.80$. Then, the NPMLE $F_n$ does not exist or is not unique for the ACS data \citep{Xiao19}. We redefined our sample restricting to the largest dataset in which the Proposition \ref{prop:exuniq} is satisfied, i.e, the conditional NPMLE exists and is unique. Thus, our final sample is composed by 917 patients, $680$ male and $237$ female, with ages at ACS diagnosis between $39$ and $90$ years.

The nonparametric and semiparametric kernel estimators for the
hazard function of $X^*$ computed from the $n=917$ patients together with the pointwise confidence bands at 95\%
level are given in Figure
\ref{Gfunction}, top row. For the semiparametric estimator a $Beta(\theta_1,\theta_2)$
model for $U^*$ was assumed, and the parameters were estimated by maximizing the
conditional likelihood of the truncation times.  The 95\% pointwise confidence bands were computed from the smoothed bootstrap; the parametric information was included in the bootstrap when dealing with $\lambda_{\widehat \theta,h}(x)$ (see the Appendix for more details).  The reason to use a smoothed bootstrap procedure (as opposed to a non-smoothed one) is the same as in \cite{Silverman87}, namely without smoothing the bootstrap would be inconsistent.

Generally speaking it is seen that the hazard of ACS increases with age. A local mode, located approximately at 78 years old, is suggested by the nonparametric estimator. This mode can be medically ignored since it is a result of the large variability of the estimator; the bootstrap confidence intervals around this mode are very wide indeed. This is in well agreement with our findings in Section 3, Models 3.1 to 3.3, in which large MISE values were found when the sampling interval was too narrow. The situation here is even worse in that the width of the sampling interval ($1.42$ years) is only 3\% the width of the support of the target variable $X^*$. In situations like this the NPMLE may be expected to be irrelevant, and the SPMLE becomes a useful alternative.

%Note that in this real data application the width of the observational window is $1.42$ years when the ages at diagnosis ranged approximately between 39 and 90 years old.  In such case, the nonparametric estimator is useless due to its huge variance and the semiparametric one is an alternative.

The optimal bandwidths, derived from the LSCV method in the Appendix, are $h=0.028$ and $h=0.031$ for the nonparametric and semiparametric estimators, respectively. For comparison purposes the naive estimator is also depicted in Figure \ref{Gfunction}  ($h=0.042$). It is seen that the three estimators are close to each other, suggesting no impact of the double truncation issue in the hazard function; this can be further investigated through the estimation of the biasing function $G(.)$.\\

%although some differences can be seen in the aforementioned local mode point. \\

The biasing function,  together with the 95\% pointwise confidence bands based on the bootstrap are displayed in Figure \ref{Gfunction}, bottom panel. It can be seen that both biasing functions $G_n$ and  $G_{\hat \theta}$ are roughly flat, and this explains why the semiparametric hazard estimator mimics the naive one which does not correct for double truncation. Note that, although $G_n$ exhibits some bumps, these are not significant according to the confidence limits.\\

%On the other hand, $G_n$ exhibits some bumps around $G_{\hat \theta}$, which result in a more wiggly hazard estimate. Summaryzing, one may say that the sampling bias induced by the interval sampling is not critical for the ACS data.

\begin{figure}[h]
\begin{minipage}[b]{0.48\linewidth}
\includegraphics[width=\textwidth]{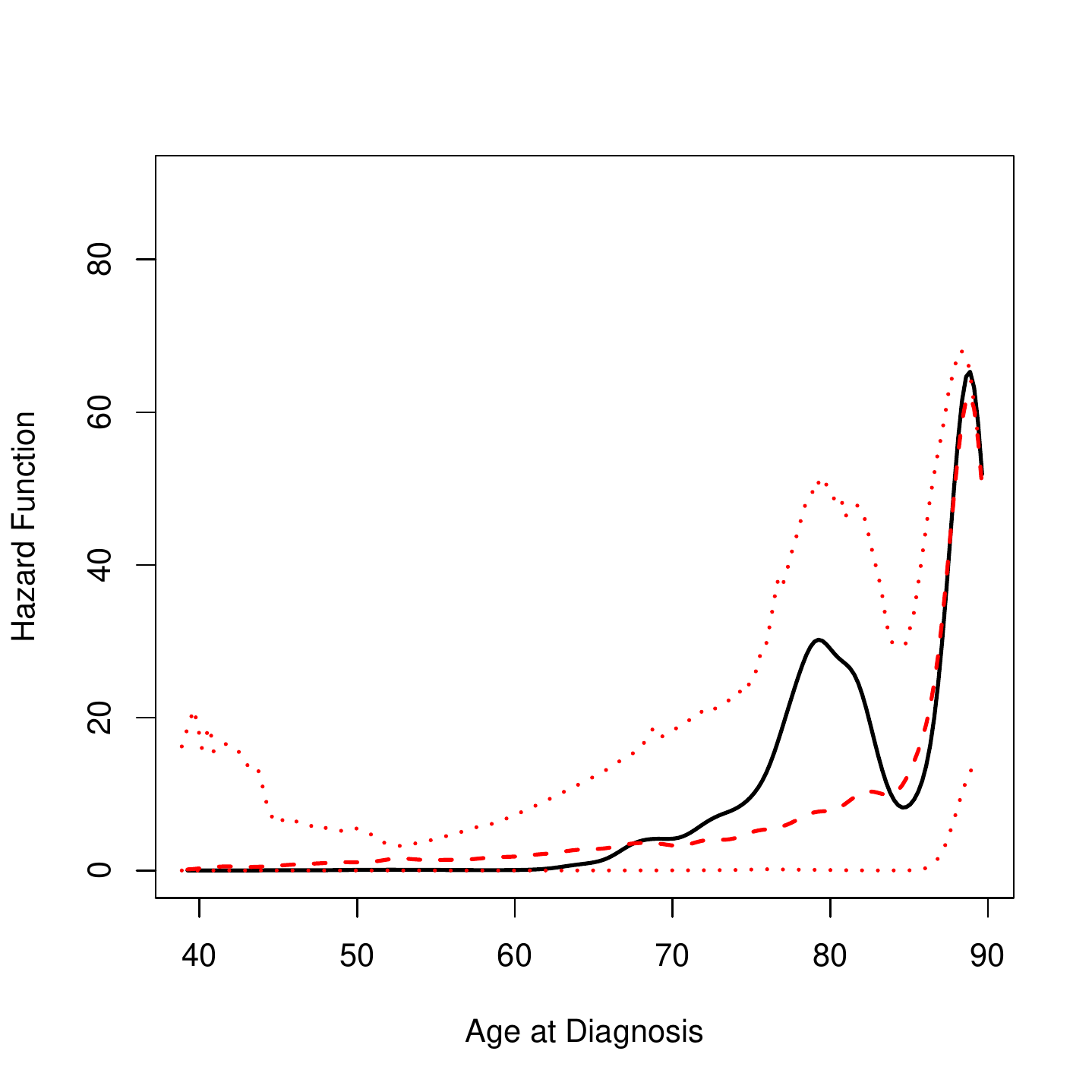}
%\caption{Figura da direita}
%\label{fig:dir}
\end{minipage}
\begin{minipage}[b]{0.48\linewidth}
\includegraphics[width=\textwidth]{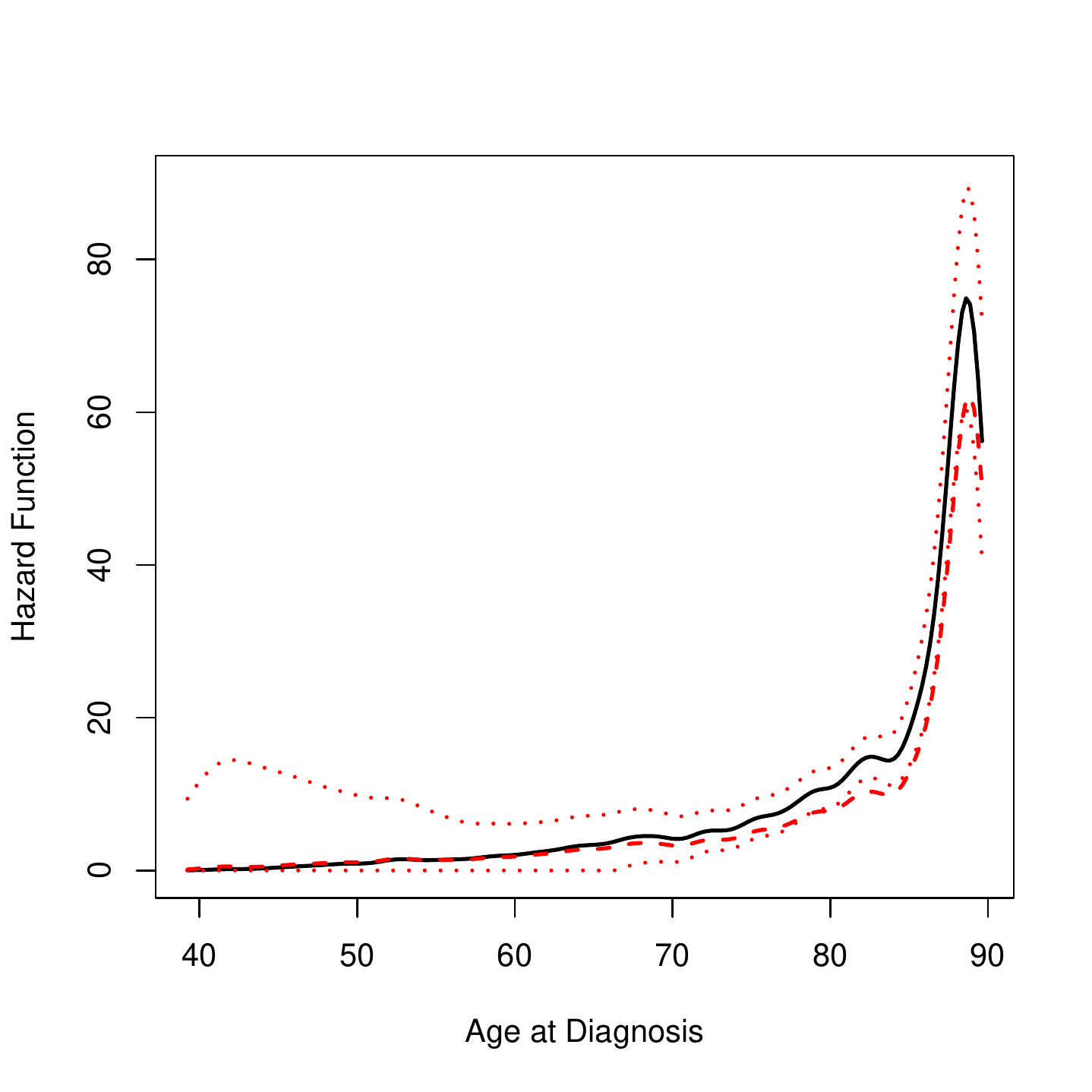}
%\caption{Figura da esquerda}
%\label{fig:esq}
%\end{minipage} \hfill
\end{minipage}
\begin{minipage}[b]{0.48\linewidth}
\includegraphics[width=\textwidth]{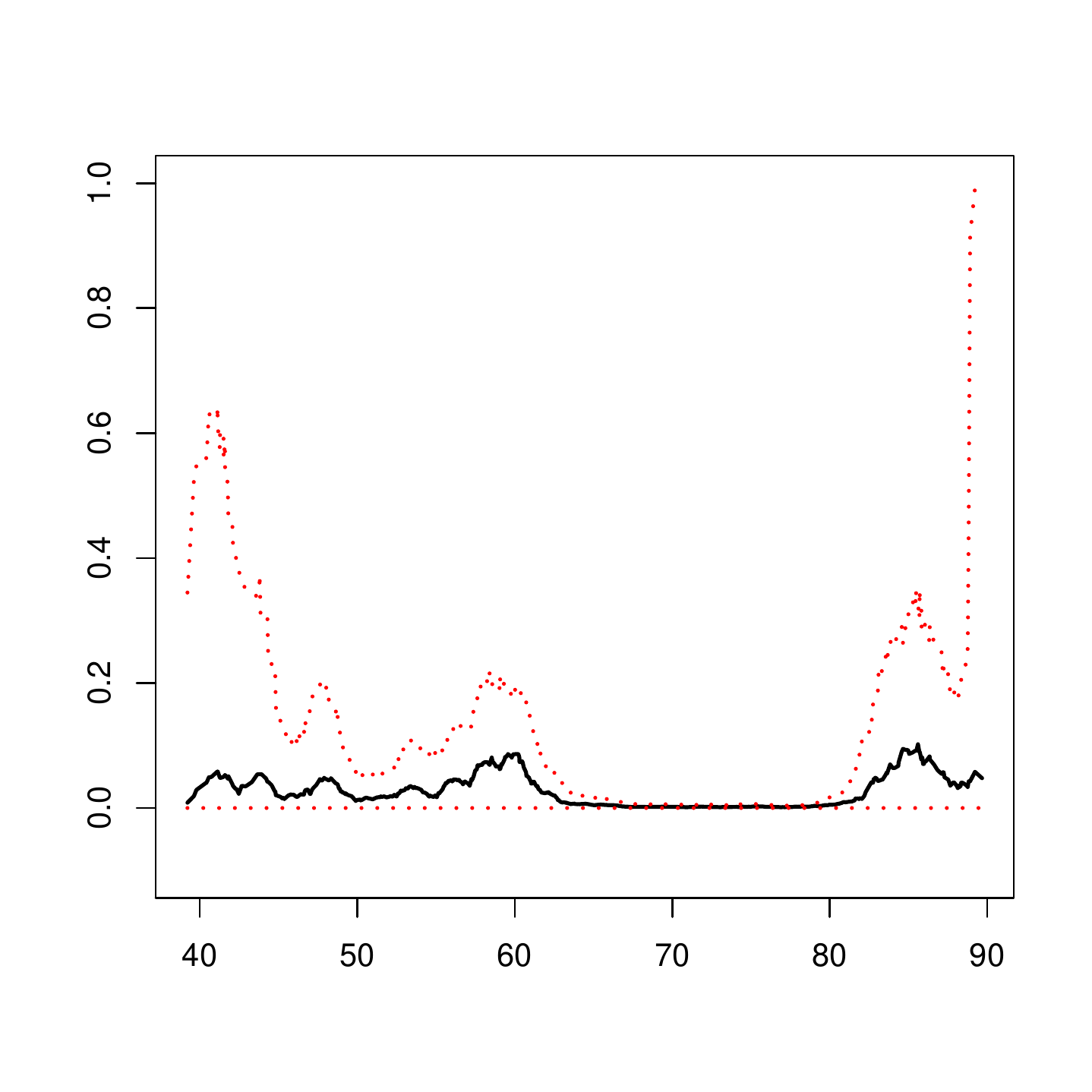}
%\caption{Figura da direita}
%\label{fig:dir}
\end{minipage}
\begin{minipage}[b]{0.48\linewidth}
\includegraphics[width=\textwidth]{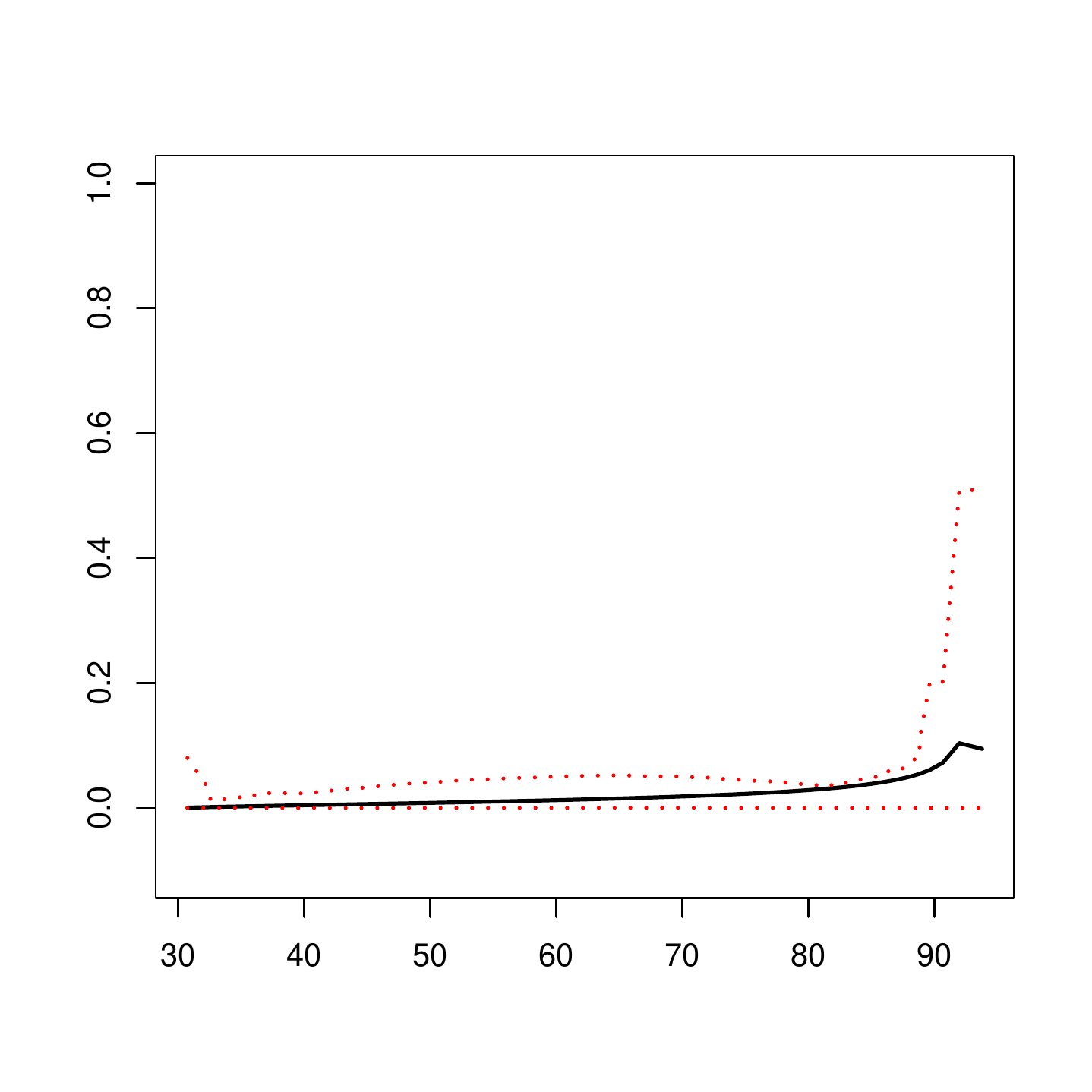}
%\caption{Figura da esquerda}
%\label{fig:esq}
%\end{minipage} \hfill
\end{minipage}
\caption{Top row: Nonparametric hazard estimator (left panel) and the semiparametric hazard estimator (rigth panel) for the age at diagnosis (solid line), ACS data (n=917), with pointwise confidence bands at level 95\% (dotted lines) and the naive estimator (dashed line). Bottom row: Nonparametric  estimator $G_n$ (left panel, solid line) with pointwise confidence bands at level 95 \% (dotted lines) and $G_{\widehat \theta}$ based on a beta model (rigth panel, dotted lines), ACS data (n=917).}\label{Gfunction}
\end{figure}

\subsection{ AIDS Blood Transfusion data}

\indent In this subsection we  use epidemiological data on transfusion-related Acquired Immune Deficiency Syndrome (AIDS). The AIDS Blood Transfusion data were collected from the Centers for Disease Control (CDC), which is from a registry database, a common source of medical data; see \cite{Bilker96} and \cite{Lawless89}. The variable of interest $(X^*)$ is the induction or incubation time, which is defined as the time elapsed from Human Immunodeficiency virus (HIV) infection to the clinical manifestation of AIDS. The CDC AIDS Blood Transfusion Data can be viewed as being doubly truncated. The data were retrospectively ascertained for all transfusion-associated AIDS cases in which the diagnosis of AIDS occurred prior to the end of the study, thus leading to right-truncation. Besides, because HIV was unknown prior to 1982, any cases of transfusion-related AIDS before this time would not have been properly classified and thus would have been missed. Thus, in addition to right-truncation, the observed data were also truncated from the left. See \cite{Bilker96} Section 5.2, for further discussions.

Data included 494 cases reported to the CDC prior to January 1, 1987, and diagnosed prior to July 1, 1986. Of the 494 cases, 295 had consistent data, and the infection could be attributed to a single transfusion or short series of transfusions. Our analyses are restricted to this subset, which is entirely reported in \cite{Lawless89}, Table 1. Values of $U^*$ were obtained by measuring the time from HIV infection to January 1, 1982; while $V^*$ was defined as time from HIV infection to the end of study (July 1, 1986). Note that the difference between $V^*$ and its respective $U^*$ is always $4.5$ years. The times were considered in months.

After checking the existence and uniqueness of the NPMLE, the semiparametric and the nonparametric kernel estimators for the
hazard rate function of $X^*$ were computed from the $n=295$ cases, together with 95\% bootstrap pointwise confidence bands. The results are displayed in Figure
\ref{RealdataAIDS}, top row. The
transformation $(t+49)/95$ has been used for the ages at
diagnosis and the truncation variables. With this transformation,
the $U^*$ is supported on the $(0,1)$ interval. For better analysis of the figures, the ages in the horizontal axis are reported in their original scale (months).  As in our first real data illustration for the semiparametric estimator, we assume a $Beta(\theta_1,\theta_2)$
model for $U^*$, and the parameters are estimated by maximizing the
conditional likelihood of the truncation times. This parametrization has been used since it permits a range of different curves to describe the data.

The optimal (LSCV) bandwidths for the nonparametric and semiparametric estimators were $h=0.61$ and $h=0.53$ respectively. For comparison purposes, the naive kernel
hazard estimator which does not correct  the double truncation issue
is also reported. Figure \ref{RealdataAIDS}, top row, reveals that the hazard increases with the induction times, which is in accordance with the literature. It is also seen that the proposed estimators are close to each other along their whole support, while the naive estimator clearly overestimates the hazard function.\\
In Figure \ref{RealdataAIDS}, bottom row, we display the parametric and nonparametric biasing functions together with the 95\% pointwise confidence bands based on the bootstrap. The two estimators are roughly equivalent, and they both suggest a sampling probability which decreases as the induction time increases. This decreasing shape of the function $G$ is responsible for the positive bias of the naive hazard estimator. Indeed, it can be proved in general that, when $G$ is non-increasing, the hazard rate corresponding to the observed $X$ is greater than the target. Both the nonparametric and the semiparametric estimators declare a mode around 65-68 months for the hazard, although with no epidemiological interpretation.

%leading to an interesting relation between the naive estimator and those corresponding to the decreasing biasing functions.$G(t)=P(U\leq t \leq V)$ and $F^*(x)=\alpha^{-1} \int_0^x G(t)F(dt)$ with $\alpha= \int_0^\infty G(t)F(dt)$.\\$F(x)=\alpha \int_0^x G(t)^{-1}F^*(dt)$, with $\alpha= \left[\int_0^\infty G(t)F(dt)\right]^{-1}$. So, $h^*(x)=\frac {f^*(x)}{1-F^*(x)}= \frac{G(x)f(x)}{\int_x^{\infty}G(t)F(dt)}\geq \frac{f(x)}{\int_x^{\infty}F(dt)}=h(x)$. This explains the underestimation of the naive estimator.

\begin{figure}[h]
\begin{minipage}[b]{0.49\linewidth}
\includegraphics[width=\textwidth]{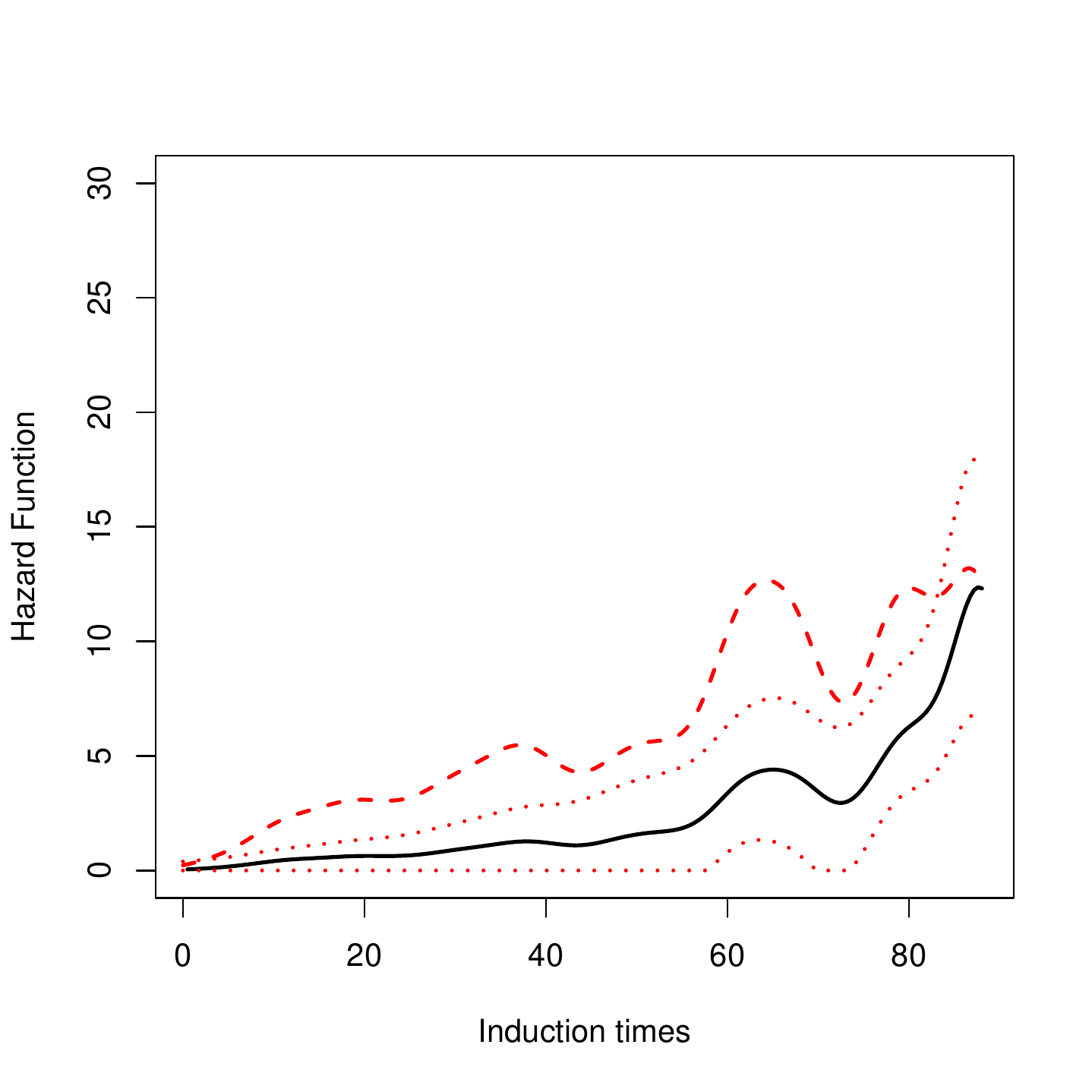}
%\caption{Figura da direita}
%\label{fig:dir}
\end{minipage}
\begin{minipage}[b]{0.49\linewidth}
\includegraphics[width=\textwidth]{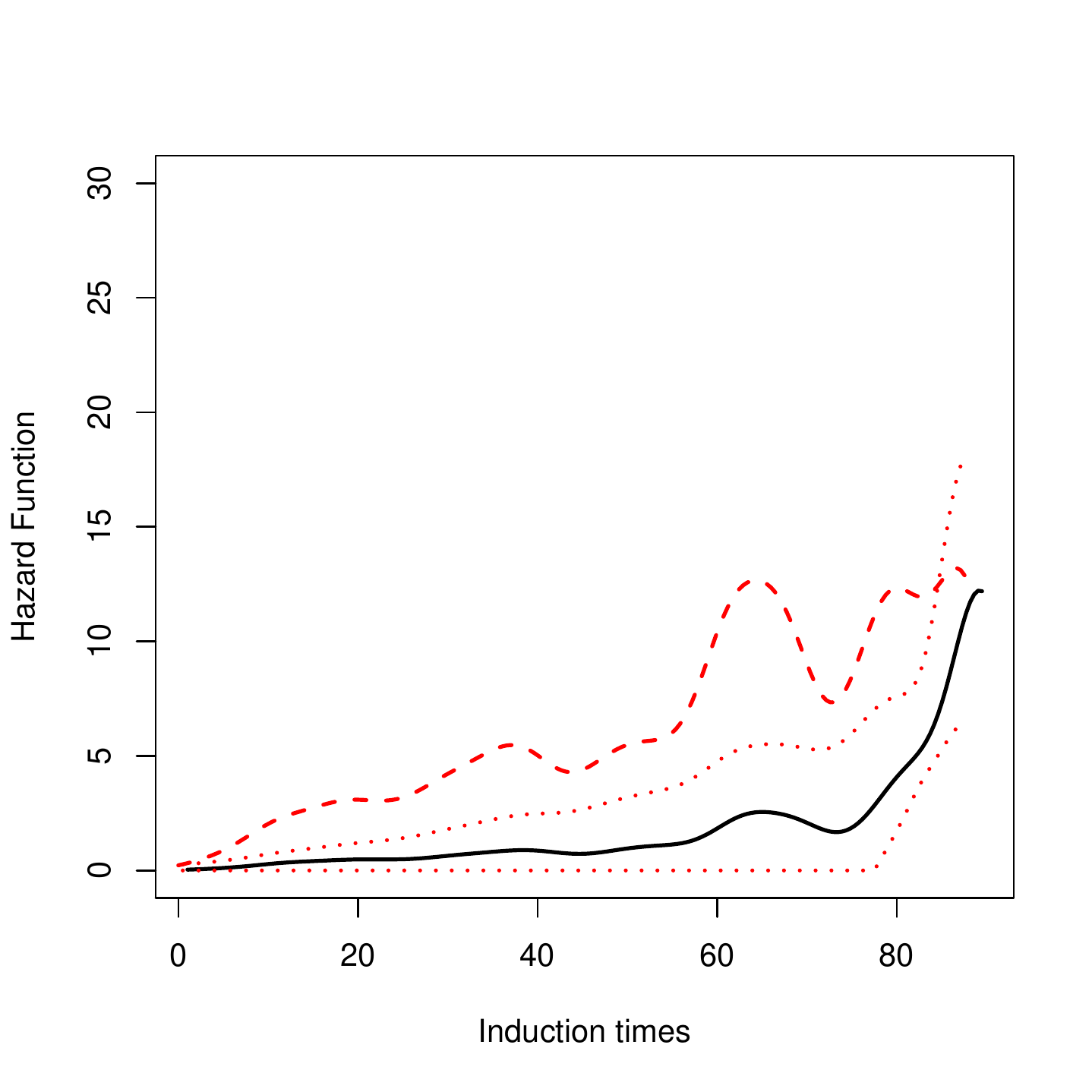}
%\caption{Figura da esquerda}
%\label{fig:esq}
%\end{minipage} \hfill
\end{minipage}
\begin{minipage}[b]{0.49\linewidth}
\includegraphics[width=\textwidth]{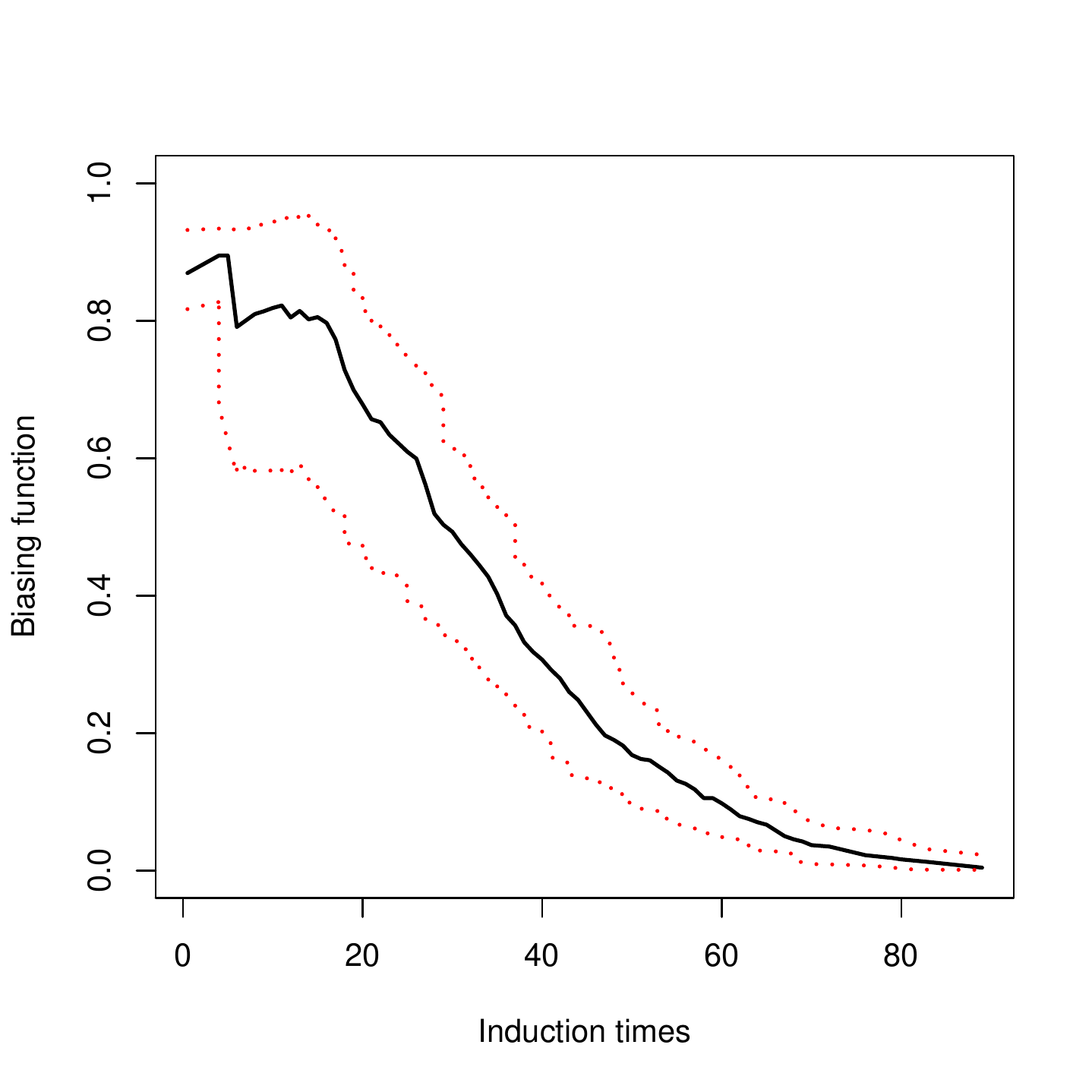}
%\caption{Figura da direita}
%\label{fig:dir}
\end{minipage}
\begin{minipage}[b]{0.49\linewidth}
\includegraphics[width=\textwidth]{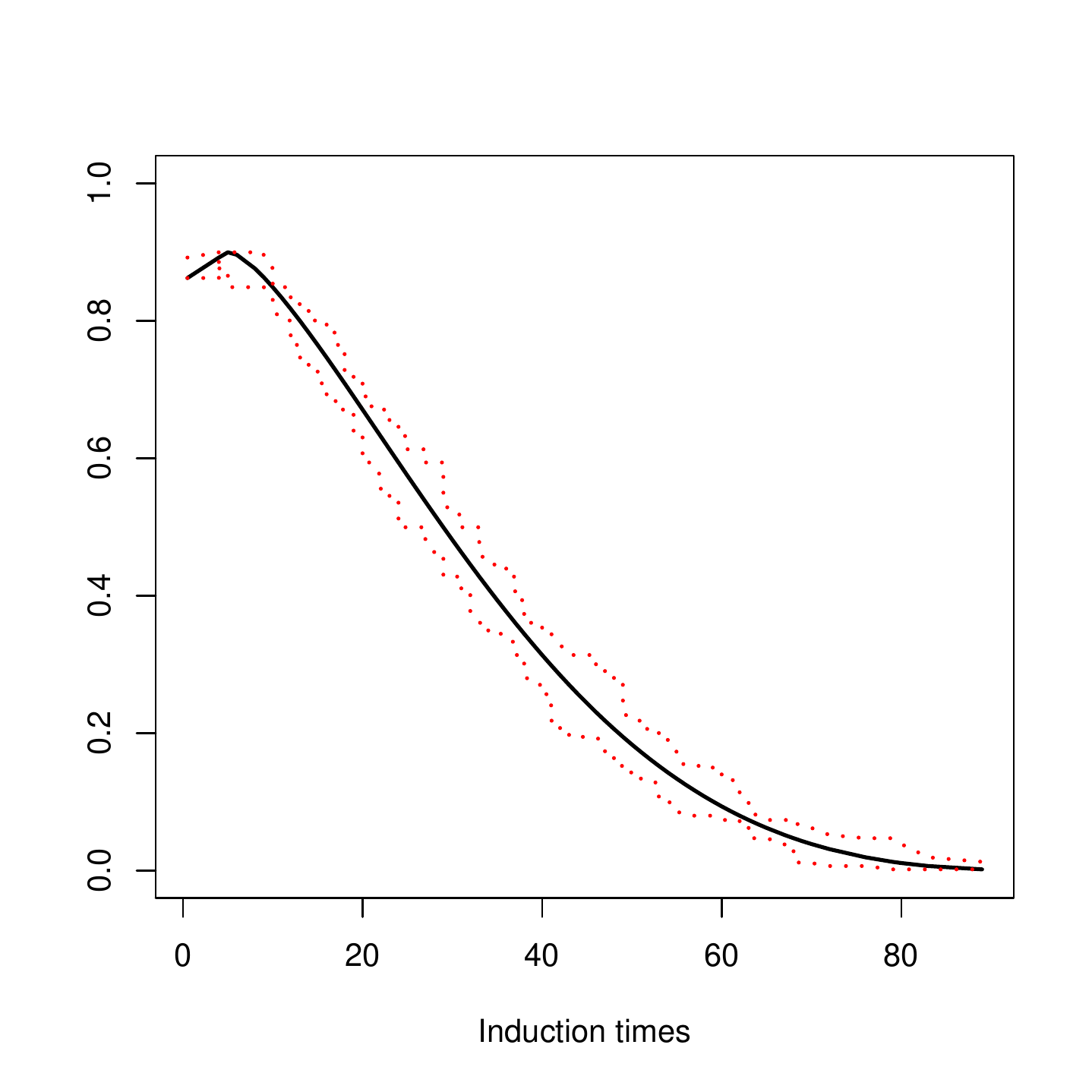}
%\caption{Figura da esquerda}
%\label{fig:esq}
%\end{minipage} \hfill
\end{minipage}
\caption{Top row: Nonparametric hazard estimator (left panel) and the semiparametric hazard estimator (rigth panel) for the age at diagnosis (solid line), AIDS data (n=295), with pointwise confidence bands at level 95\% (dotted lines) and the naive estimator (dashed line). Bottom row: Nonparametric  estimator $G_n$ (left panel, solid line) with pointwise confidence bands at level 95 \% (dotted lines) and $G_{\widehat \theta}$ based on a beta model (rigth panel, dotted lines), AIDS data (n=295).} \label{RealdataAIDS}
\end{figure}

\section{Conclusions}
\label{Section6}
In this paper we have introduced kernel hazard estimators for a variable which is observed under random double truncation. Two estimators have been proposed. The first one is purely nonparametric, and it is defined as a convolution of a kernel function with the NPMLE of the cumulative hazard. The second estimator is based on a parametric specification for the df of the truncation variables, thus being semiparametric. Asymptotic properties of the two estimators have been discussed, including a formula for the asymptotic MISE. A LSCV bandwidth selection criterion for the automatic application of the proposed smoothers has been derived.

The nonparametric and semiparametric estimators are asymptotically equivalent in the sense of having the same asymptotic MISE. However, simulations have shown that, for moderate sample sizes, the semiparametric estimator may outperform the nonparametric estimator. Importantly, the nonparametric estimator may be missleading due to its relatively large variance, providing spurious bumps in particular applications. Furthermore, the NPMLE may not exist and, therefore, the semiparametric estimator may be the only way out in estimation. The referred issues are particularly present with interval sampling when the sampling interval is very narrow. This is in agreement with the intuition that short sampling intervals may result in little, or too deteriorated, information on the target.

Two real data illustrations were provided.  For the ACS data, some features concerning the NPMLE, such as non-existence and non-uniqueness, were encountered. In order to make the application of the nonparametric estimator possible, the sample was reduced to the largest dataset for which the NPMLE exists and is unique. In this reduced dataset it was seen that the effect of double truncation was almost negligible, the proposed estimators being close to the ordinary kernel smoother. In contrast, for the  Blood Transfusion data the effect of double truncation was found critical and the standard kernel estimator exhibited a gross positive bias. Thus, in practice, taking the double truncation issue into account is very important.

In the two real data applications the semiparametric estimator provided confidence intervals much narrower than the nonparametric estimator, indicating the importance of modeling the truncation distibution. This was much more evident for the ACS data. In practice, the semiparametric estimator may be recommended when a suitable parametric family for the truncation distribution is available. A two-parameter beta model worked well in our applications. Besides, in general the semiparametric estimator is well defined, so no sample reduction is needed. Application of the semiparametric estimator to the full ACS dataset ($n=939$) provided an estimate similar to that in Figure \ref{Gfunction}, top-right plot (results not shown). The analysis of the full dataset through the nonparametric hazard estimate was not possible due to the aforementioned nonexistence of the NPMLE.

\section*{Appendix} \label{Appendix1}
%%%%%%%%%%%%%%%%%%%%%%%%%%%%
\subsection*{Proof of Theorem 1}

%%%%%%%%%%%%%%%%%%%%%%%%%%%%
\begin{proof}
For (i) introduce $\widetilde{\lambda}_{h}(x)=\alpha
\frac {G(x)^{-1}}{1-F(x)}f^*_{h}(x)$ where
\begin{eqnarray*}
f^*_{h}(x)=\frac{1}{n}\sum_{i=1}^{n}K_{h}(x-X_{i})
\end{eqnarray*} is the ordinary kernel density estimator computed from the observed data.
By \cite{Devroye79} we have $f^*_{h}(x)\rightarrow
f^*(x)$ almost surely, where $f^*$ is the density of $F^*(x)=P(X_1\leq x)$. Now,
if the support of $K$ is contained in $\left[ -a,a\right] $,
\begin{eqnarray*}
\left| \overline{\lambda}_{h}(x)-\widetilde{\lambda}_{h}(x)\right| \leq \alpha
f^*_{h}(x) \sup_{x-ah\leq y\leq x+ah}\left|
\frac{G(y)^{-1}}{1-F(y)}-\frac{G(x)^{-1}}{1-F(x)}\right|,
\end{eqnarray*}
and the supremum goes to zero as $h \rightarrow 0$ by the continuity of $G$ at $x$. This
ends with the proof to (i). Statement (ii) is proved similarly to
Section 2 of \cite{Parzen62}; by following such lines we obtain
\begin{eqnarray*}
\left( nh\right) ^{1/2}\left( \overline{\lambda}_{h}(x)-E\overline{\lambda}_{h}(x)\right) \rightarrow N(0,\alpha \frac{G(x)^{-1}}{1-F(x)}\lambda(x)R(K))
\end{eqnarray*}
in distribution. Now, a two-term Taylor expansion (and the fact that $K$ is even) gives $E\overline{\lambda}_{h}(x)=\lambda(x)+O(h^{2})$. Since $nh^{5}\rightarrow 0$,
this implies the claimed result.
\end{proof}

%%%%%%%%%%%%%%%%%%%%%%%%%%%%
\subsection*{LSCV bandwidth selection}

%%%%%%%%%%%%%%%%%%%%%%%%%%%%

 We give the details of the LSCV bandwidth selector for the semiparametric estimator $\lambda_{\widehat \theta,h}$.  The LSCV selector for the nonparametric estimator is analogous.

LSCV aims to estimate the $MISE$ and then to minimize the $MISE$ with respect to $h$.  Given the estimator $ \lambda_{\widehat \theta, h}$ of the hazard $\lambda$, the mean integrated squared error can be written as
\begin{eqnarray*}
MISE(\lambda_{\widehat \theta, h})= E[ISE(\lambda_{\widehat \theta, h})]
= E \Big[\int (\lambda_{\widehat \theta, h}(x)-\lambda(x))^2dx \Big]= E \Big[\int \lambda_{\widehat \theta, h}^2(x) \, dx - 2 \int \lambda_{\widehat \theta, h}(x) \lambda(x)dx + \int \lambda^2(x)dx \Big].
\label{ISE}
\end{eqnarray*}

\noindent The term $\int \lambda^2(x)dx$ does not depend on $h$, and so minimizing $MISE(\lambda_{\widehat \theta, h})$ is equivalent to minimizing
\begin{eqnarray*}
S(\lambda_{\widehat \theta, h})&= MISE(\lambda_{\widehat \theta, h})-\int \lambda^2(x)dx= E\left[ \int \lambda_{\widehat \theta, h}^2(x) \, dx - 2\alpha \int \lambda_{\widehat \theta, h}(x) \frac{G^{-1}(x)}{1-F(x)} F^\ast(dx)\right].
\end{eqnarray*}

\noindent Now, in order to construct an estimator of S$(\lambda_{\widehat \theta, h})$, let $\lambda_{\widehat \theta, h;-i}$ be the hazard estimator constructed from all data points except $X_i$, i.e.
\begin{equation*}
\lambda_{\widehat \theta, h;-i}(x)=\alpha_{\widehat \theta; -i}\frac{1}{n-1}\displaystyle \sum_{j\neq i}K_{h}(x-X_{j})\frac{G_{\widehat \theta; -i}^{-1}(X_{j})}{1-F_{\widehat \theta; -i}(X_{j})}, \label{LSCV}
\end{equation*}
where $G_{\widehat \theta; -i}(\cdot)$, $\alpha_{\widehat \theta; -i}(\cdot)$ and $F_{\widehat \theta; -i}(\cdot)$ are the estimators of $G$, $\alpha$ and $F(\cdot)$, defined in Section 2, except that the $i$-th data point is not used for estimating $\theta$. Introduce
\[
LSCV(h)=\\
=\int \lambda^2_{\widehat \theta, h}(x) \, dx - 2n^{-1} \displaystyle \sum_{i=1}^n   \alpha_{\widehat \theta; -i} \lambda_{\widehat \theta, h; -i}(X_i) \frac{G^{-1}_{\widehat \theta; -i}(X_i)}{1-F_{\widehat \theta; -i}(X_i)},
\]

\noindent and estimate the optimal $h$ by minimizing $LSCV(h)$ over $h$:
\begin{eqnarray*}
\widehat h_{LSCV}=\argmin_hLSCV(h).
\end{eqnarray*}

Note that
\begin{eqnarray*}
&& E \Big[n^{-1} \displaystyle \sum_{i=1}^n \alpha_{\widehat \theta; -i} \lambda_{\widehat \theta, h; -i}(X_i) \frac{G^{-1}_{\widehat \theta; -i}(X_i)}{1-F_{\widehat \theta; -i}(X_i)} \Big] \\
&& = E \Big[\alpha_{\widehat \theta; -1} \lambda_{\widehat \theta, h; -1}(X_1) \frac{G^{-1}_{\widehat \theta; -1}(X_1)}{1-F_{\widehat \theta; -1}(X_1)} \Big] \\
&& = E \Big[\alpha_{\widehat \theta; -1} \int \lambda_{\widehat \theta, h; -1}(x) \frac{G^{-1}_{\widehat \theta; -1}(x)}{1-F_{\widehat \theta; -1}(x)}  F^\ast(dx) \Big],
\end{eqnarray*}
and this is asymptotically equivalent to $E[\alpha \int \overline \lambda_{h;-1}(x) \frac{G^{-1}(x)}{1-F(x)} F^\ast(dx)] = E [\int \overline \lambda_{h;-1}(x) \Lambda(dx)] \linebreak = E [\int \overline \lambda_{h}(x) \lambda(x)dx]$, where $\Lambda$ stands for the cumulative hazard of $F$. Last equality follows from the fact that $E\{\lambda_{\theta, h}(x)\}$ depends only on the kernel and the bandwidth, and not on the sample size. Hence, $E [LSCV(h)]$ is asymptotically equivalent to $S(\lambda_{\widehat \theta, h})$, which suggests that we can expect $\widehat h_{LSCV}$ to be close to the minimizer of  $S(\lambda_{\widehat \theta, h})$, that is, the minimizer of   $MISE(\lambda_{\widehat \theta, h})$.

%%%%%%%%%%%%%%%%%%%%%%%%%%%%
\subsection*{Smoothed bootstrap}

%%%%%%%%%%%%%%%%%%%%%%%%%%%%

The smoothed bootstrap procedure can be described as follows. In order to simplify the presentation,
we restrict our attention  to the semiparametric estimator $\lambda_{\widehat \theta,h}$. The proposed method can be adapted
to the nonparametric estimator in an obvious way; see below. For fixed $B$ and for $b=1,\ldots,B$:

\begin{enumerate}
\item Let $X_{b,i}^{boot}$, $i=1, \ldots, n$, be an i.i.d.\ sample from $f_{\widehat \theta, h^0}(x)=\alpha_{\widehat \theta}n^{-1}\sum_{i=1}^n K_{h^0}(x-X_i)G_{\widehat \theta}(X_i)^{-1}$, where the pilot bandwidth $h^0$ is chosen to be $\widehat h^0_{LSCV}$ (other choices for $h^0$ are possible as well). Let $(U_{b,i}^{boot}, V_{b,i}^{boot})$, $i=1, \ldots, n$, be an i.i.d.\ sample from $T_{\widehat \theta}$.  Next, for each $i=1,\ldots,n$, we keep the triplet $(U_{b,i}^{boot}, X_{b,i}^{boot}, V_{b,i}^{boot})$ in the resample only if the condition $U_{b,i}^{boot} \le X_{b,i}^{boot} \le V_{b,i}^{boot}$ is fulfilled.   If not, the same resampling procedure is repeated until a triplet satisfying the inequality is found.
\item Let $\widehat\theta_b^{boot}$ and $\lambda_{\widehat \theta_b^{boot}, b,h}^{boot}$ be the estimator of $\theta$ (that is, $\widehat \theta$) and of the hazard $\lambda$ ($\lambda_{\widehat \theta,h}$) respectively, obtained from the bootstrap sample $(U_{b,i}^{boot}, X_{b,i}^{boot}, V_{b,i}^{boot})$, $i=1, \ldots, n$.
\end{enumerate}

\noindent Variability of $\widehat \theta$ and $\lambda_{\widehat \theta,h}$ is then estimated by that of the $B$ bootstrap evaluations $\widehat\theta_b^{boot}$ and $\lambda_{\widehat \theta_b^{boot}, b,h}^{boot}$, $1\leq b \leq B$. Note that the resampling plan above is an obvious bootstrap, as opposed to the simple bootstrap which directly resamples with replacement from the triplets $(U_i,X_i,V_i)$ \citep{Moreira10}. This allows for an easy introduction of the semiparametric information, since the truncation couple is resampled from the estimated parametric model $T_{\widehat \theta}$. When the focus is the nonparametric estimator $\lambda_h$, the parametric distribution $T_{\widehat \theta}$ is replaced by the NPMLE $T_n$ in Step 1 above; and, obviously, the semiparametric density estimator $f_{\widehat \theta, h^0}(x)$ is replaced by its nonparametric counterpart $f_{h^0}(x)=\alpha_n n^{-1}\sum_{i=1}^n K_{h^0}(x-X_i)G_n(X_i)^{-1}$.

\begin{acknowledgements}
\section{Acknowledges}
Work supported by the Grant MTM2017-89422-P (MINECO/AEI/FEDER, UE) and by Portuguese Funds through FCT — Funda\c c\~ao Ci\^encia e Tecnologia, within the Projects UIDB/00013/2020 and UIDP/00013/2020. Financial support from the Xunta de Galicia (Centro singular de investigaci\'on de Galicia accreditation 2019-2022) and the EU (ERDF), Ref. ED431G2019/06, is acknowledged too.
This study is also a result of the project DOCnet (NORTE-01-0145-FEDER-000003), supported by Norte Portugal Regional Operational Programme (NORTE 2020), under the PORTUGAL 2020 Partnership Agreement, through the European Regional Development Fund (ERDF).
\end{Acknowledges}


\begin{thebibliography}{}

\bibitem{Araujo18}
Ara{\'u}jo, C., O.~Laszczy{\'n}ska, M.~Viana, F.~Mel{\~a}o, A.~Henriques,
  A.~Borges, M.~Severo, M.~J. Maciel, I.~Moreira, and A.~Azevedo,
 Sex differences in presenting symptoms of acute coronary syndrome:
  the epiheart cohort study,
{\em BMJ Open\/}~{\em 8\/}(2) (2018)


\bibitem{Bilker96}
Bilker, W.B. and M.-C. Wang,
A semiparametric extension of the mann-whitney test for randomly
  truncated data,
 {\em Biometrics}, 52, 10--20 (1996)



\bibitem{Cai98}
Cai, Z., 
 Kernel density and hazard rate estimation for censored dependent
  data,
 {\em Journal of Multivariate Analysis\/}~{\em 67\/}(1), 23 -- 34 (2018)


%\bibitem[\protect\citeauthoryear{Chaieb, Rivest, and Abdous}{Chaieb
%  et~al.}{2006}]{Chaieb06}
%Chaieb, L., L.-P. Rivest, and B.~Abdous (2006).
%\newblock Estimating survival under a dependent truncation.
%\newblock {\em Biometrika\/}~{\em 93}, 655--669.

\bibitem{dUAVK19}
de U\~{n}a-\'{A}lvarez, J. and  Van Keilegom, I., 
 An asymptotic representation of Efron-Petrosian integrals with covariates: an asymptotic analysis,
 {\em Bernoulli} ~{27}, 249–-273 (2021)


\bibitem{Devroye79}
Devroye, L.P. and Wagner, T.J., 
The $L^1$ convergence of kernel density estimates,
{\em The Annals of Statistics}~{\em 7}, 1136--1139 (1979)


\bibitem{Diehl88}
Diehl, S. and W.~Stute ,
 Kernel density and hazard function estimation in the presence of
  censoring,
 {\em Journal of Multivariate Analysis\/}~{\em 25\/}(2), 299 -- 310 (1988)

%\bibitem[\protect\citeauthoryear{Ding}{Ding}{2012}]{Ding2012}
%Ding, A.~A. (2012, Oct).
%\newblock Copula identifiability conditions for dependent truncated data model.
%\newblock {\em Lifetime Data Analysis\/}~{\em 18\/}(4), 397--407.

\bibitem{Efron99}
Efron, B. and V.~Petrosian 
Nonparametric methods for doubly truncated data,
{\em Journal of the American Statistical Association\/}~{\em 94}, 824--834 (1999)


%\bibitem[\protect\citeauthoryear{Emura, Wang, and Hung}{Emura
%  et~al.}{2011}]{Emura2011}
%Emura, T., W.~Wang, and H.-N. Hung (2011).
%\newblock Semi-parametric inference for copula models for truncated data.
%\newblock {\em Statistica Sinica\/}~{\em 21\/}(1), 349--367.

%\bibitem[\protect\citeauthoryear{Emura and Wang}{Emura and
%  Wang}{2012}]{EMURA2012}
%Emura, T. and W.~Wang (2012).
%\newblock Nonparametric maximum likelihood estimation for dependent truncation
%  data based on copulas.
%\newblock {\em Journal of Multivariate Analysis\/}~{\em 110}, 171 -- 188.


%\bibitem[\protect\citeauthoryear{Emura}{Emura}{2015}]{Emura2015}
%Emura, T.and~Murotani, K. (2015).
%\newblock An algorithm for estimating survival under a copula-based dependent
%  truncation model.
%\newblock {\em TEST\/}~{\em 24\/}(4), 734--751.


\bibitem{Lawless89}
Kalbfleisch, J.~D. and J.~F. Lawless,
 Inference based on retrospective ascertainment: An analysis of the
  data on transfusion-related aids,
 {\em American Statistical Association\/}~{\em 84}, 360--372 (1989)


\bibitem{Lemdani07}
Lemdani, M. and E.~Ould-Sa\"{i}d,
 Asymptotic behavior of the hazard rate kernel estimator under
  truncated and censored data,
 {\em Communications in Statistics - Theory and Methods\/}~{\em
  36\/}(1), 155--173 (2007)

%\bibitem[\protect\citeauthoryear{Lior and Xie}{Lior and Xie}{2017}]{Rennert17}
%Lior, R. and Xie, S. (2017).
%\newblock Cox regression model with doubly truncated data.
%\newblock {\em Biometrics\/}~{\em 74\/}(2), 725--733.

\bibitem{Lynden71}
Lynden-Bell, D.,
A method for allowing for known observational selection in small
  samples applied to 3cr quasars,
 {\em Monthly Notices of the Royal Astronomical Society\/}~{\em 155},  95--118 (1971)



\bibitem{Mandel2017}
Mandel, M., de~Uña-Álvarez J., Simon ~D. K., and Betensky ~R. A.,
Inverse probability weighted cox regression for doubly truncated
  data,
{\em Biometrics\/}~{\em 74\/}(2), 481--487 (2018)

\bibitem{Moreira10}
Moreira, C. and J.~de~U\~na{-}\'Alvarez,
 Bootstrappping the npmle for doubly truncated data,
{\em Journal of Nonparametric Statistics\/}~{\em 22}, 567--583 (2010)


\bibitem{Moreira010}
Moreira, C., J.~de~U\~na{-}\'Alvarez, and R.~Crujeiras,
 Dtda: an r package to analyze randomly truncated data,
{\em Journal of Statistical Software\/}~{\em 37}, 1--20 (2010)

\bibitem{Moreira110}
Moreira, C. and J.~de~U\~na{-}\'Alvarez (2010a),
 A semiparametric estimator of survival for doubly truncated data,
 {\em Statistics in Medicine\/}~{\em 29}, 3147--3159 (2010)

%\bibitem[\protect\citeauthoryear{Moreira and de Uña-Álvarez}{Moreira and
%  de~U\~na{-}\'Alvarez}{2010b}]{Moreira2010b}
%Moreira, C. and J.~de~U\~na{-}\'Alvarez (2010b).
%\newblock Bootstrapping the NPMLE for doubly
%truncated data.
%\newblock {\em Journal of Nonparametric Statistics\/}~{\em 22}, 567--583.



\bibitem[\protect\citeauthoryear{Moreira and de~U\~na{-}\'Alvarez}{Moreira and
  de~U\~na{-}\'Alvarez}{2012}]{Moreira2012}
Moreira, C. and J.~de~U\~na{-}\'Alvarez (2012).
\newblock Kernel density estimation with doubly truncated data.
\newblock {\em Electron. J. Statist.\/}~{\em 6}, 501--521.


\bibitem[\protect\citeauthoryear{Moreira and Keilegom}{Moreira and
  Keilegom}{2013}]{Moreira13}
Moreira, C. and I.~V. Keilegom (2013).
\newblock Bandwidth selection for kernel density estimation with doubly
  truncated data.
\newblock {\em Computational Statistics \& Data Analysis\/}~{\em 61}, 107 --  123.

\bibitem[\protect\citeauthoryear{M\"{u}ller and Wang}{M\"{u}ller and
  Wang}{1994}]{Wang94}
M\"{u}ller, H.~G. and J.~L. Wang (1994).
\newblock Hazard rate estimation under random censoring with varying kernels
  and bandwidths.
\newblock {\em Biometrika\/}~{\em 50\/}(1), 61--76.


%\bibitem[\protect\citeauthoryear{Nelsen}{Nelsen}{2006}]{Nelsen2006}
%Nelsen, R.~B. (2006).
%\newblock {\em An Introduction to Copulas (Springer Series in Statistics)}.
%\newblock Secaucus, NJ, USA: Springer-Verlag New York, Inc.

\bibitem[\protect\citeauthoryear{Nichols, Townsend, Scarborough, and
  Rayner}{Nichols et~al.}{2014}]{Nichols2014}
Nichols, M., N.~Townsend, P.~Scarborough, and M.~Rayner (2014).
\newblock {Cardiovascular disease in Europe 2014: epidemiological update}.
\newblock {\em European Heart Journal\/}~{\em 35\/}(42), 2950--2959.

\bibitem[\protect\citeauthoryear{Parzen}{1962}]{Parzen62}
Parzen, E. (1962).
\newblock On estimation of a probability density function and mode
\newblock {\em Annals of Mathematical Statistics }~{\em 33}, 1065--1076.

\bibitem[\protect\citeauthoryear{Rennert and Xie}{2019}]{RennertXie2019}
Rennert, L. and S. X. Xie (2019).
\newblock Bias induced by ignoring double truncation inherent
in autopsy-confirmed survival studies of
neurodegenerative diseases
\newblock {\em Statistics in Medicine }~{\em 38}, 3599--3613.

\bibitem[\protect\citeauthoryear{Silverman} {1986}]{Silverman87}
Silverman, B.W. (1986).
\newblock Density Estimation for Statistics and Data Analysis.
\newblock {\em Monographs on Statistics and Applied Probability}~{\em 26}
\newblock{\em Chapman and Hall}.

\bibitem[\protect\citeauthoryear{Shen}{Shen}{2010}]{Shen08}
Shen, P. (2010).
\newblock Nonparametric analysis of doubly truncated data.
\newblock {\em Annals of the Institute of Statistical Mathematics\/}~{\em 62},  835--853.

\bibitem[\protect\citeauthoryear{Turnbull,  W.}{1976}]{Turnbull76}
Turnbull, Bruce W. (1976).
\newblock The Empirical Distribution Function with Arbitrarily Grouped, Censored and Truncated Data.
\newblock {\em Journal of the Royal Statistical Society. Series B\/}~{\em 38\/}, 290--295.


\bibitem[\protect\citeauthoryear{Wand and Jones}{Wand and Jones}{1995}]{Wand95}
Wand, M.~P. and M.~C. Jones (1995).
\newblock {\em Kernel Smoothing}, Volume~60 of {\em Monographs on Statistics
  and Applied Probability}.
\newblock London: Chapman and Hall Ltd.

\bibitem[\protect\citeauthoryear{Vardi} {1985}]{Vardi85}
Vardi, Y. (1985).
\newblock Empirical distributions in selection bias models.
\newblock {\em Ann. Statist}~{\em 13}, 178–-203.



\bibitem[\protect\citeauthoryear{Woodroofe}{Woodroofe}{1985}]{Woodroofe85}
Woodroofe, M. (1985).
\newblock Estimating a distribution function with truncated data.
\newblock {\em The Annals of Statistics\/}~{\em 13}, 163--177.


\bibitem[\protect\citeauthoryear{Xiao and  Hugdens}{Xiao and Hudgens}{2019}]{Xiao19}
Xiao, J. and M.~G. Hudgens (2019).
\newblock On nonparametric maximum likelihood estimation with double truncation.
\newblock {\em Biometrika\/}~{\em 106\/}(4), 989--996.


\bibitem[\protect\citeauthoryear{Zhou}{Zhou}{1999}]{Zhou999}
Zhou, Y. (1999).
\newblock Asymptotic representations for kernel density and hazard function
  estimators with left truncation.
\newblock {\em Statistica Sinica\/}~{\em 9\/}(2), 521--533.

%\bibitem[\protect\citeauthoryear{Zhu and Wang}{Zhu and Wang}{2012}]{Zhu2012}
%Zhu, H. and M.-C. Wang (2012).
%\newblock Analysing bivariate survival data with interval sampling and
% application to cancer epidemiology.
%\newblock {\em Biometrika\/}~{\em 99\/}(2), 345--361.

\bibitem[\protect\citeauthoryear{Zhu and Wang}{Zhu and Wang}{2014}]{Zhu2014}
Zhu, H. and M.-C. Wang (2014).
\newblock Nonparametric inference on bivariate survival data with interval
  sampling: association estimation and testing.
\newblock {\em Biometrika\/}~{\em 101\/}(3), 519--533.









\end{thebibliography}
\end{document}